\documentclass[12pt]{article}

\usepackage{xr}
\usepackage{natbib}
\usepackage[pdftex, final]{graphicx}
\DeclareGraphicsExtensions{.png,.pdf,.jpg}
\usepackage{enumerate}
\usepackage{url} 
\usepackage{soul}
\usepackage{algorithm}
\usepackage{algpseudocode}
\usepackage{amsmath}
\usepackage{bm}
\usepackage{bbm}
\usepackage{verbatim}
\usepackage{amssymb,amsthm,amsmath}
\usepackage{dsfont}
\usepackage{mathrsfs}
\usepackage{subcaption}
\usepackage{placeins}

\usepackage{booktabs}
\usepackage{multirow}

\def\E{\mathds{E}}

\def\R{\mathds{R}}

\mathchardef\mhyphen="2D

\allowdisplaybreaks

\newtheorem{proposition}{Proposition}

\newtheorem{remark}{Remark}

\def\simiid{\stackrel{\mbox{\scriptsize{iid}}}{\sim}}

\usepackage{mathtools}
\usepackage[colorlinks=true,linkcolor=blue,citecolor=blue]{hyperref}
\usepackage{cleveref}

\usepackage[normalem]{ulem}

\newtheorem{thm}{Theorem}[section]
\newtheorem{lem}[thm]{Lemma}

\usepackage{tabularx}
\begin{document}

\def\spacingset#1{\renewcommand{\baselinestretch}%
{#1}\small\normalsize} \spacingset{1}


  \title{\bf Gaussian credible intervals in Bayesian nonparametric estimation of the unseen}

\author{
Claudia Contardi \\
Department of Mathematics\\
University of Pavia, Pavia, Italy \\[0.4cm]
Emanuele Dolera \\
Department of Mathematics\\
University of Pavia, Pavia, Italy \\[0.4cm]
Stefano Favaro \\  Department of Economics and Statistics\\ 
University of Torino and Collegio Carlo Alberto, Torino, Italy
 }
\maketitle

\bigskip
\begin{abstract}
The unseen-species problem assumes $n\geq1$ samples from a population of individuals belonging to different species, possibly infinite, and calls for estimating the number $K_{n,m}$ of hitherto unseen species that would be observed if $m\geq1$ new samples were collected from the same population. This is a long-standing problem in statistics, which has gained renewed relevance in biological and physical sciences, particularly in settings with large values of $n$ and $m$. In this paper, we adopt a Bayesian nonparametric approach to the unseen-species problem under the Pitman-Yor prior, and propose a novel  methodology to derive large $m$ asymptotic credible intervals for $K_{n,m}$, for any $n\geq1$. By leveraging a Gaussian central limit theorem for the posterior distribution of $K_{n,m}$, our method improves upon competitors in two key aspects: firstly, it enables the full parameterization of the Pitman-Yor prior, including the Dirichlet prior; secondly, it avoids the need of Monte Carlo sampling, enhancing computational efficiency. We validate the proposed method on synthetic and real data, demonstrating that it improves the empirical performance of competitors by significantly narrowing the gap between asymptotic and exact credible intervals for any $m\geq1$.
\end{abstract}

\noindent%
\vfill

\spacingset{1.9} 


\section{Introduction}\label{sec1}

The estimation of the number of unseen species is a long-standing problem in statistics, dating back to the seminal work of \citet{Fis(43)} on  ``species extrapolation". It assumes that $n\geq1$ random samples $(X_{1},\ldots,X_{n})$ are collected from an unknown discrete distribution $P$ on $\mathbb{S}$, with $\mathbb{S}$ being a space of species' labels or symbols, and calls for estimating 
\begin{displaymath}
K_{n,m}=|\{X_{n+1},\ldots,X_{n+m}\}\setminus \{X_{1},\ldots,X_{n}\}|,
\end{displaymath}
namely the number of hitherto unseen species that would be observed if $m\geq1$ additional samples $(X_{n+1},\ldots,X_{n+m})$ were collected from the same distribution $P$. First introduced in ecology \citep{Cha(84),Cha(92),Bun(93)}, the unseen-species problem has more recently found applications in biological and physical sciences, where it poses significant challenges in handling large values of $n$ and $m$ \citep{Kro(99),Gao(07),Ion(09),Dal(13)}. See \citet{Den(19)} for an overview with emphasis on large-scale biological data. Further applications include, e.g., information theory and theoretical computer science \citep{Haa(95),Mot(06),Flo(07),Bub(13),Cai(18)}, empirical linguistics and natural large language modeling \citep{Efr(76),Thi(87),Orl(04),Oha(12),Ben(18),Kal(24)}, and forensic DNA analysis \citep{Cer(17),FN(24)}. 

\subsection{Background and motivation}

A frequentist nonparametric approach to the unseen-species problem was proposed by \citet{Goo(56)} and \citet{Efr(76)}, then developed rigorously by \citet{Orl(17)}. This is a distribution-free approach, in the sense that it does not rely on any assumption on $P$, leading to estimates of $K_{n,m}$ that are minimax optimal for any $n$ and $m\leq n\log n$, with such a range being the best (largest) possible \citep{Orl(17),Wu(16),Wu(19)}. From a Bayesian nonparametric (BNP) perspective, it is natural to specify a prior distribution for $P$, an approach that was first investigated in \citet{Lij(07)} by focussing on the Pitman-Yor prior \citep{Pit(97)}, which is a prior indexed by $\alpha\in[0,1)$ and $\theta>-\alpha$, with $\alpha=0$ being the celebrated Dirichlet prior \citep{Fer(73)}. Under the Pitman-Yor prior, \citet{Lij(07)} showed that the posterior distribution of $K_{n,m}$, given $(X_{1},\ldots,X_{n})$, depends on the sampling information only through the sample size $n$ and the number $K_{n}$ of species in $(X_{1},\ldots,X_{n})$. Such a distribution is in closed-form, with the posterior mean estimate $\hat{K}_{n,m}$ that can be easily evaluated for any value of $n$ and $m$ \citep{Fav(09)}.  See \citet{Bal(23)} for an up-to-date overview.

Uncertainty quantification for estimates of $K_{n,m}$ has been addressed under the BNP approach, but remains an open problem under the distribution-free approach, especially when $m>n$ \citep{Orl(17)}. For $\alpha\in[0,1)$, exact credible intervals for $K_{n,m}$ can be derived, for any $n$ and $m$, by Monte Carlo sampling the posterior distribution through the predictive distributions of the Pitman-Yor prior \citep{Bal(23)}. Further, if $\alpha\in(0,1)$ then large $m$ asymptotic credible intervals can be derived using the method proposed by \citet{Fav(09)}. Denoting by $K_{m}^{(n)}$ a random variable whose distribution is the posterior distribution of $K_{n,m}$ given $K_{n}=j$, \citet{Fav(09)} showed that, as $m\rightarrow+\infty$
\begin{displaymath}
\frac{K_{m}^{(n)}}{(\theta+n+m)^{\alpha}-(\theta+n)^{\alpha}}\stackrel{\text{w}}{\longrightarrow}S_{\alpha,\theta}^{(n,j)},
\end{displaymath}
where $S_{\alpha,\theta}^{(n,j)}$ is a scaled Mittag-Leffler random variable \citep[Chapter 0]{Pit(06)}. Given values of $(n,j)$ and $(\alpha,\theta)$, with $(\alpha,\theta)$ estimated by means of empirical or fully Bayes procedures, for $m$ sufficiently large the distribution of $K_{m}^{(n)}$ is approximated by the distribution of $c_{\alpha,\theta,n}(m)S_{\alpha,\theta}^{(n,j)}$, with $c_{\alpha,\theta,n}(m)=(\theta+n+m)^{\alpha}-(\theta+n)^{\alpha}$. Mittag-Leffler credible intervals for $K_{n,m}$ are then derived by Monte Carlo sampling $c_{\alpha,\theta,n}(m)S_{\alpha,\theta}^{(n,j)}$. In particular, the  scaling $c_{\alpha,\theta,n}(m)$ is determined in such a way that $\hat{K}_{n,m}$ coincides with the expected value of $c_{\alpha,\theta,n}(m)S_{\alpha,\theta}^{(n,j)}$, ensuring that intervals are centered on the BNP estimator $\hat{K}_{n,m}$ for any $m\geq1$. 

While the method proposed by \citet{Fav(09)} addresses uncertainty quantification for large values of $m$, it comes with notable limitations. Firstly, Monte Carlo sampling $c_{\alpha,\theta,n}(m)S_{\alpha,\theta}^{(n,j)}$ is computationally expensive, despite the recent advances on sampling scaled Mittag-Leffler distributions \citep{Qu(21)}, which limits the practical (computational) benefits of asymptotic credible intervals over exact ones derived by sampling the posterior distribution, unless $m$ is extremely large. Secondly, empirical analyses by \citet{Fav(09)} show that Mittag-Leffler credible intervals are shorter than the exact intervals, with the gap decreasing when $m$ enters the regime $m\gg\theta+n$, namely $m$ is much larger than $\theta+n$. Thirdly, the method of \citet{Fav(09)} fails to extend to the case $\alpha=0$ due to a degenerate behaviour of the limiting posterior distribution \citep{Bal(23)}. These theoretical and empirical limitations highlight the motivation for this paper, which introduces an alternative methodology to uncertainty quantification for BNP estimates of $K_{n,m}$.

\subsection{Preview of our contributions}

We propose a novel method to derive large $m$ asymptotic credible intervals for $K_{n,m}$, which allows to deal with $\alpha\in[0,1)$ and avoids the use of Monte Carlo sampling. Under the Pitman-Yor prior, for $\alpha\in[0,1)$ and $\theta>-\alpha$, we rely on the large $m$ asymptotic behaviour of $K_{m}^{(n)}$ assuming that both the sampling information $(n,j)$ and the parameter $\theta$ are large. In particular, we set $n=\nu m$, $j=\rho m$ and $\theta=\tau m$, with $\nu,\rho,\tau>0$, and show that, as $m\rightarrow+\infty$
\begin{displaymath}
\frac{K_{m}^{(n)}-m\mathscr{M}_{\alpha,\tau,\nu,\rho}}{\sqrt{m\mathscr{S}_{\alpha,\tau,\nu,\rho}}}\stackrel{\text{w}}{\longrightarrow}N(0,1),
\end{displaymath}
where
\begin{displaymath}
 \mathbb{E}\left[K_m^{(n)}\right] =m\mathscr{M}_{\alpha, \tau, \nu, \varrho}+ O(1)
\end{displaymath}
and
\begin{displaymath}
 \operatorname{Var}\left(K_m^{(n)}\right) =m \mathscr{S}_{\alpha, \tau, \nu, \varrho}^2 + O(1),
\end{displaymath}
for some (explicit) functions $\mathscr{M}_{\alpha, \tau, \nu, \varrho}$ and $\mathscr{S}_{\alpha, \tau, \nu, \varrho}$ of $(\alpha,\tau,\nu,\varrho)$, respectively, and where $N(0,1)$ is the standard Gaussian random variable. Given values of $(\nu m,\rho m)$ and $(\alpha,\tau m)$, with $(\alpha,\tau m)$ estimated by means of empirical or fully Bayes  procedures, for $m$ sufficiently large the distribution of $K_{m}^{(n)}$ is approximated by a Gaussian distribution with mean $m\mathscr{M}_{\alpha, \tau, \nu, \varrho}$ and variance $m \mathscr{S}_{\alpha, \tau, \nu, \varrho}^2$. Gaussian credible intervals for $K_{n,m}$, with a prescribed (asymptotic) level, are then derived from the Gaussian quantiles, enhancing computational efficiency.  

Given the sampling information $(n,j)$ and the parameter $(\alpha,\theta)$, the proposed methodology and that of \citet{Fav(09)} employ different approximations of the posterior distribution of $K_{n,m}$, for $m$ sufficiently large. Both approximations lead to large $m$ asymptotic credible intervals that are centered on the BNP estimator $\hat{K}_{n,m}$. Figure \ref{fig1_intro} shows that our method outperforms that of \citet{Fav(09)} on synthetic data generated from various natural distributions, namely Zipf, Dirichlet-Multinomial and Uniform distributions; the same synthetic datasets were analyzed in \citet[Figure 3]{Orl(17)}. In particular, compared to Mittag-Leffler credible intervals, Gaussian credible intervals provide greater coverage of the exact credible intervals, for any $m\geq1$. A comprehensive empirical validation of our methodology is presented in the paper, encompassing both synthetic and real datasets. For real data, we present an application to the Expressed Sequence Tags (ESTs) data considered in \citet[Section 3]{Fav(09)}, confirming the behaviour displayed in Figure \ref{fig1_intro}.

\begin{figure}[h]
\begin{minipage}{0.5 \textwidth}
\begin{center}
\medskip
A) Zipf
\medskip

\includegraphics[width = \textwidth]{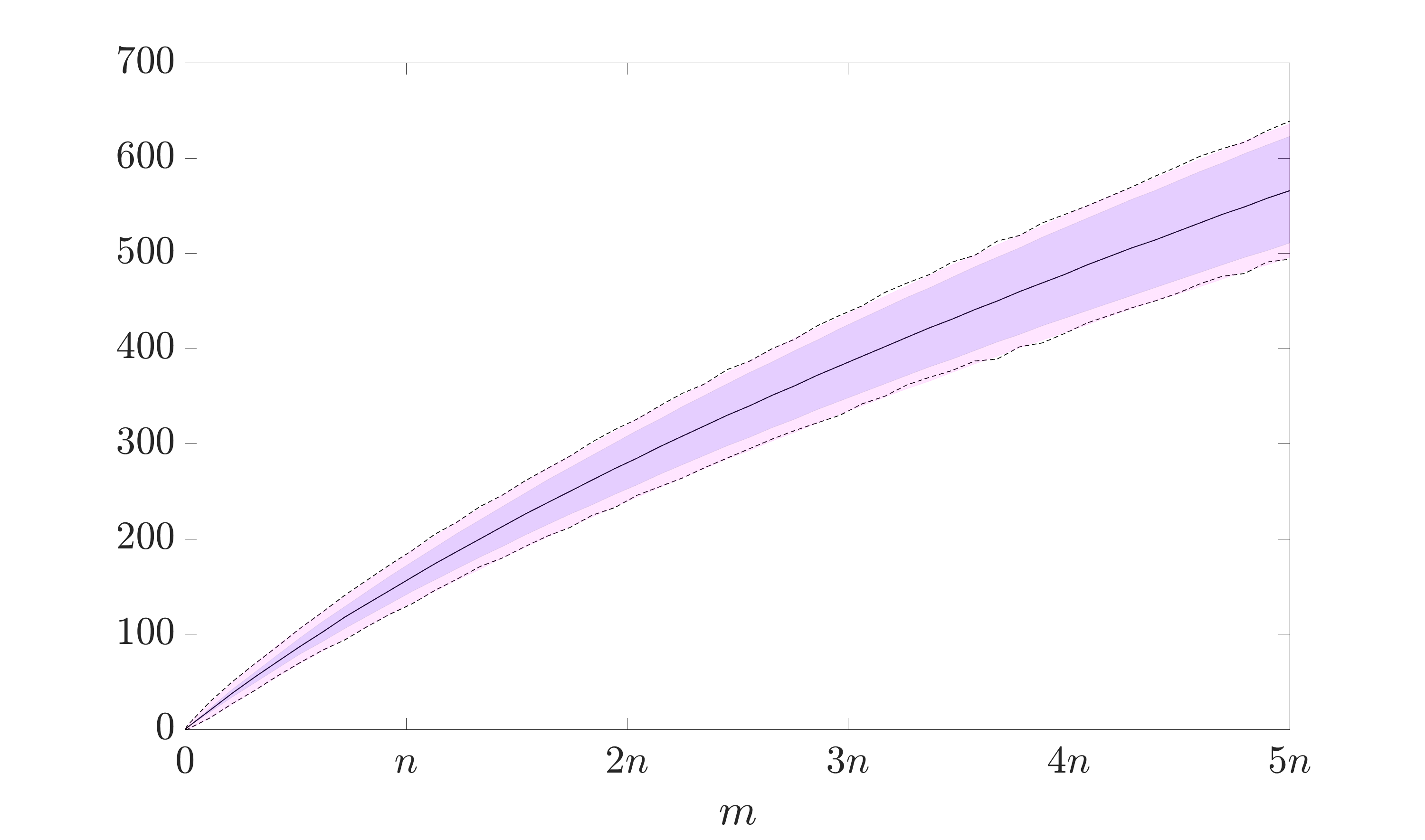}
\end{center}
\end{minipage}
\begin{minipage}{0.5 \textwidth}
\begin{center}pdf
\medskip
B) Zipf
\medskip

\includegraphics[width = \textwidth]{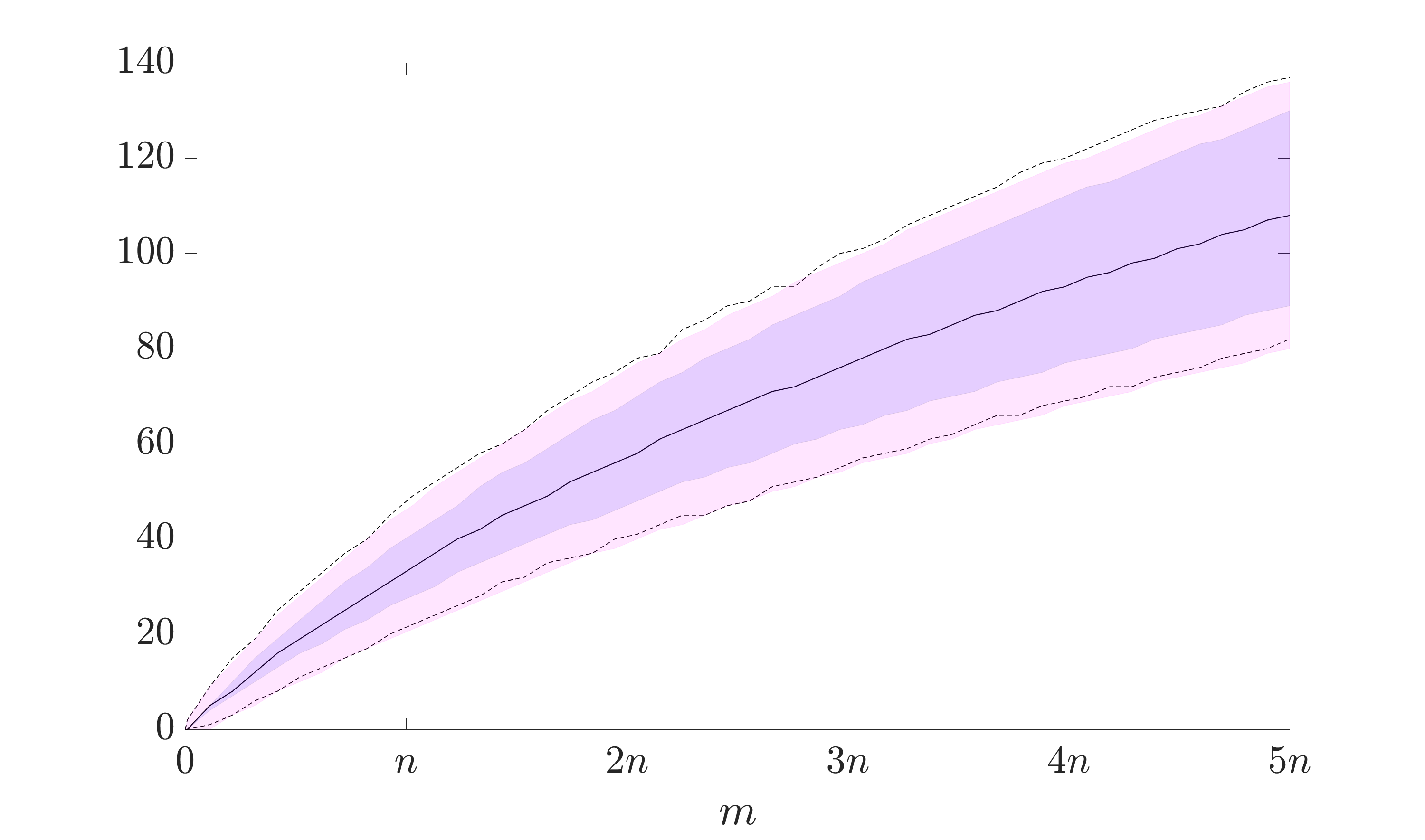}
\end{center}
\end{minipage}

\begin{minipage}{0.5 \textwidth}
\begin{center}
\medskip
C) P\'olya
\medskip

\includegraphics[width = \textwidth]{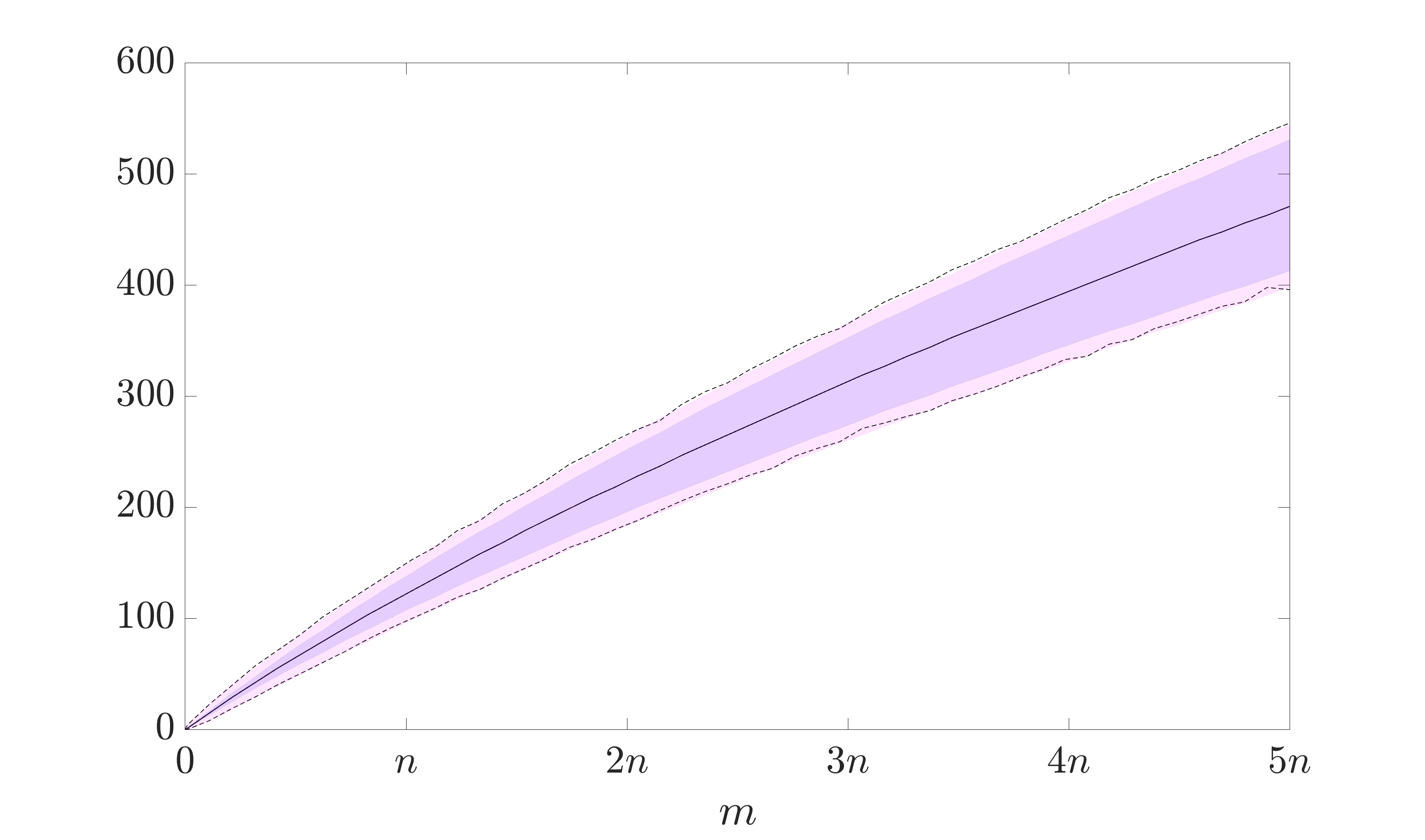}
\end{center}
\end{minipage}
\begin{minipage}{0.5 \textwidth}
\begin{center}
\medskip
D) Uniform
\medskip

\includegraphics[width = \textwidth]{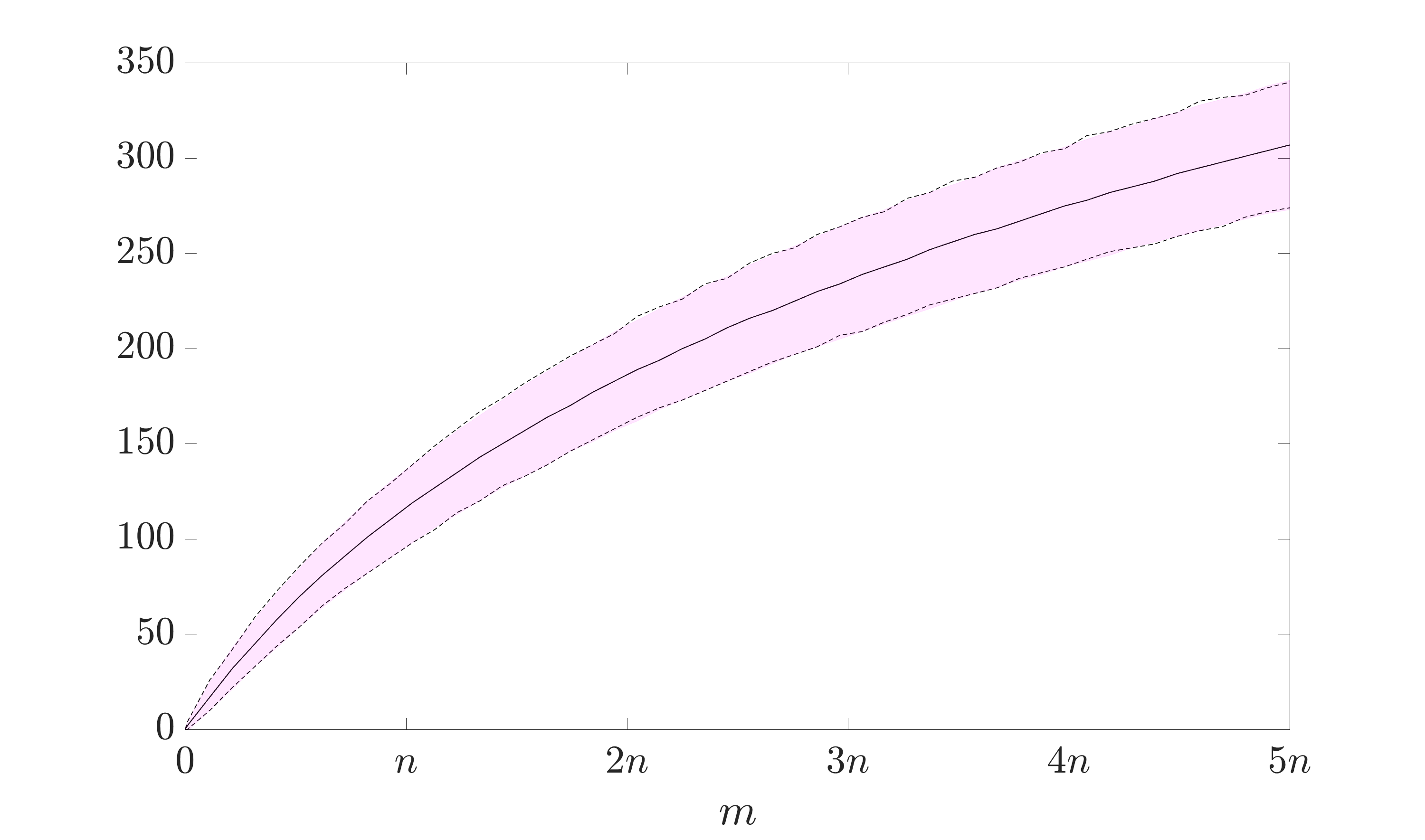}
\end{center}
\end{minipage}
\caption{\scriptsize{BNP estimates of $K_{n,m}$ (solid line --) with 95\% exact credible intervals (dashed line - -), Mittag-Leffler credible intervals (violet) and Gaussian credible intervals (pink), as a function of $m \in [0, 5n]$. Synthetic datasets generated from the following discrete distributions: A) Zipf distribution on $\{0,1,\ldots,300\}$ with parameter $2$, $n=977$, $j=300$, and estimated $(\alpha,\theta)=(0.54,\,26.67)$; B) Zipf distribution on $\{0,1,\ldots,100\}$ with parameter $1.5$, $n=1877$, $j=100$, and estimated $(\alpha,\theta)=(0.38,\,4.66)$; C) P\'olya distribution on $\{0,1,\ldots,500\}$ with parameter (2, 2, 500, 500, ..., 500), $n=2,000$, $j=227$, and estimated $(\alpha,\theta)=(0.69,\,1.80)$; D) Uniform distribution on $\{0,1,\ldots,500\}$, with $n=2,000$, $j=447$, and estimated $(\alpha,\theta)=(0,\,178.48)$. The parameter $(\alpha,\theta)$ is estimated through an empirical Bayes procedure \citep[Section 3]{Fav(09)}.}}
\label{fig1_intro}
\end{figure}

\subsection{Outline of the paper}

The paper is organized as follows. In Section \ref{sec2} we recall the BNP approach to the unseen-species problem, including the modeling assumptions and posterior inferences. Section \ref{sec3} contains the Gaussian CLT for $K_{m}^{(n)}$, with corresponding large $m$ asymptotic credible intervals for $K_{n,m}$. The empirical performance of our methodology is investigated in Section \ref{sec4}. Section \ref{sec5} concludes with a discussion and some directions for future research. Technical preliminary results, proofs and additional numerical illustrations are deferred to the Appendix.


\section{The BNP approach to the unseen-species problem}\label{sec2}

Following the BNP approach of \citet{Lij(07)}, we consider $n\geq1$ observations with values in the space of species' labels or symbols $\mathbb{S}$, modeled as random samples $\mathbf{X}_{n} = (X_{1},\ldots,X_{n})$ such that
\begin{align}\label{eq:exchangeable_model}
\begin{split}
X_1,\ldots,X_{n}\,|\,P &\, \simiid\, P, \\
P&\, \sim\,\text{PYP}(\alpha,\theta),
\end{split}
\end{align}
where $\text{PYP}(\alpha,\theta)$ is the Pitman-Yor prior indexed by $\alpha\in[0,1)$ and $\theta>-\alpha$. For short, we say that $\mathbf{X}_{n}$ is a random sample from $\text{PYP}(\alpha,\theta)$. A simple and intuitive definition of $P\sim\text{PYP}(\alpha,\theta)$ follows from its stick-breaking construction \citep{Per(92)}. Specifically, let: i) $(V_{i})_{i\geq1}$ be independent random variables, with each $V_{i}$ following a Beta distribution with parameter $(1-\alpha,\theta+i\alpha)$; ii) $(S_{j})_{j\geq1}$ be random variables following a non-atomic distribution $\nu$ on $\mathbb{S}$ and independent of each other as well as of the $V_{i}$'s. If $P_{1}=V_{1}$ and $P_{j}=V_{j}\prod_{1\leq i\leq j-1}(1-V_{i})$ for $j\geq1$, so that $P_{j}\in(0,1)$ for any $j\geq1$ and $\sum_{j\geq1}P_{j}=1$ almost surely, then $P=\sum_{j\geq1}P_{j}\delta_{S_{j}}\sim\text{PYP}(\alpha,\theta)$, with the Dirichlet prior corresponding to $\alpha=0$. 

\begin{remark}
If $(P_{(j)})_{j\geq1}$ are the decreasingly ordered stick-breaking random probabilities $P_{j}$'s of $P\sim\text{PYP}(\alpha,\theta)$, then, for $\alpha\in(0,1)$, as $j\rightarrow+\infty$ the $P_{(j)}$'s follow a power-law distribution of exponent $c=\alpha^{-1}$ \citep{Pit(97)}. The parameter $\alpha\in(0,1)$ controls the power-law tail of $P$ through the small $P_{(j)}$'s: the larger $\alpha$, the heavier the tail of $P$. As $\alpha\rightarrow0$, the Dirichlet prior features geometric tails \citep[Chapter 4]{Pit(06)}. We refer to Appendix \ref{asymptotics} for an alternative description of the power-law behaviour driven by $\alpha\in[0,1)$.
\end{remark}

\subsection{Estimates}

Under the model \eqref{eq:exchangeable_model}, \citet[Propostion 1]{Lij(07)} computes the posterior distribution of $K_{n,m}$, given $\mathbf{X}_{n}$. For $\alpha\in(0,1)$, this is expressed in terms of the generalized factorial coefficient 
\begin{displaymath}
\mathscr{C}(u,v;a,b):=\frac{1}{v!}\sum_{i=0}^{v}(-1)^{i}{v\choose i}(-ia-b)_{(u)}
\end{displaymath}
for $a>0$, $b\geq0$ and $u,v\in\mathbb{N}_{0}$ such that $v\leq u$, where $(a)_{(u)}$ denotes the $u$-th rising factorial of $a$, i.e. $(a)_{(u)}:=\prod_{0\leq i\leq u-1}(a+i)$; see Appendix \ref{app_comb_gencoeff}. If the sample $\mathbf{X}_{n}$ features $K_{n}=j$ species with (empirical) frequencies $(N_{1,n},\ldots,N_{1,K_{n}})=(n_{1},\ldots,n_{j})$ then for $k\in\{0,1,\ldots,m\}$
\begin{equation}\label{post_py_k}
\text{Pr}[K_{m}^{(n)}=k]=\text{Pr}[K_{n,m}=k\,|\,\mathbf{X}_{n}]=\frac{\left(j+\frac{\theta}{\alpha}\right)_{(k)}}{(\theta+n)_{(m)}}\mathscr{C}(m,k;\alpha,-n+j\alpha).
\end{equation}
From \eqref{post_py_k}, \citet[Proposition 1]{Fav(09)} provides a BNP estimator of $K_{n,m}$ as posterior expectation, i.e., 
\begin{equation}\label{est_py_k}
\hat{K}_{n,m}=\E[K_{m}^{(n)}]=\E[K_{n,m}\,|\,\mathbf{X}_{n}]=\left(j+\frac{\theta}{\alpha}\right)\left(\frac{(\theta+n+\alpha)_{(m)}}{(\theta+n)_{(m)}}-1\right).
\end{equation}
While the estimator \eqref{est_py_k} can be easily evaluated for any $n$ and $m$, the computational burden for evaluating \eqref{post_py_k} becomes overwhelming as $m$ increases, due the generalized factorial coefficients.

For $\alpha=0$, the posterior distribution of $K_{n,m}$, given $\mathbf{X}_{n}$, follows from \eqref{post_py_k} by taking the limit as $\alpha\rightarrow0$ \citep{Lij(07)}. The resulting distribution is expressed in terms of Stirling numbers, whose evaluation remains computationally unfeasible for large values of $m$; see Appendix \ref{app_comb_stirling}. The corresponding BNP estimator of $K_{n,m}$ as posterior expectation is 
\begin{equation}\label{est_dp_k}
\hat{K}_{n,m}=\E[K_{m}^{(n)}]=\E[K_{n,m}\,|\,\mathbf{X}_{n}]=\sum_{i=1}^{m}\frac{\theta}{\theta+n+i-1},
\end{equation}
which can be easily evaluated for any $n$ and $m$. See \citet{Bal(23)} for further details on the case $\alpha=0$.

\subsection{Exact credible intervals} \label{sec 2.2}

Exact credible intervals for $K_{n,m}$ are derived by Monte Carlo sampling the posterior distribution of $K_{n,m}$, given $\mathbf{X}_{n}$. By exploiting the closed-form expression of the posterior distribution \eqref{post_py_k}, as well as the corresponding expression for $\alpha=0$, one may implement the inverse transform algorithm to sample $K_{m}^{(n)}$. However, due to the evaluation of the distribution of $K_{m}^{(n)}$, the inverse transform algorithm becomes computationally unfeasible for large values of $m$. As an alternative approach, one may consider to sample $K_{m}^{(n)}$ by relying on the predictive distribution, or generative scheme, of the Pitman-Yor prior \citep{Pit(95)}; see Appendix \ref{predictive}. The use of the predictive distribution reduces the problem of Monte Carlo sampling $K_{m}^{(n)}$ to the problem of sampling $m$ Bernoulli random variables, which can be easily performed for any value of $n$ and $m$ \citep{Bal(23)}; in particular, for $\alpha=0$ the Bernoulli random variables are independent. This approach is reported in Algorithm \ref{alg:k}.

\begin{algorithm}
\caption{Monte Carlo sampling $K_{m}^{(n)}$}\label{alg:k}
\begin{algorithmic}
\Require{$n$, $K_{n}$, $m$, $\alpha$, $\theta$} 
\Function{MonteCarloK}{$n$, $K_{n}$, $m$, $\alpha$, $\theta$}
  \State {$K$ $\gets$ {$K_{n}$}}
    \For{$i \gets 0$ to $m-1$}                    
        \State {$b$ $\gets$ {$\text{Random sample from Bernoulli}\left(\frac{\theta+\alpha K}{\theta+n+i}\right)$}}
         \State {$K$ $\gets$ {$K+b$}}
    \EndFor
    \State \Return {$K_{m}^{(n)}=K-K_{n}$}
\EndFunction
\end{algorithmic}
\end{algorithm}

\subsection{Large $m$ asymptotic credible intervals}\label{sec 2.3}

For $\alpha\in(0,1)$, \citet{Fav(09)} proposed a method to derive large $m$ asymptotic credible intervals for $K_{n,m}$. In particular, \citet[Proposition 2]{Fav(09)} shows that, as $m\rightarrow+\infty$
\begin{equation}\label{aslimit}
\frac{K_{m}^{(n)}}{(\theta+n+m)^{\alpha}-(\theta+n)^{\alpha}}\stackrel{\text{w}}{\longrightarrow}S_{\alpha,\theta}^{(n,j)}\stackrel{\text{d}}{=}B_{j+\theta/\alpha,n/\alpha-j}S_{\alpha,(\theta+n)/\alpha},
\end{equation}
where $B_{j+\theta/\alpha,n/\alpha-j}$ and $S_{\alpha,(\theta+n)/\alpha}$ are independent random variables such that: i) $B_{a,b}$ is Beta distributed with parameter $a,b>0$; ii) $S_{\alpha,q}$ is Mittag-Leffler distributed with parameter $\alpha\in(0,1)$ and $q>0$, i.e. with density function $f_{S_{\alpha,q}}(s)\propto s^{q-1/\alpha-1}f_{\alpha}(y^{1/\alpha})$, where $f_{\alpha}$ is the positive $\alpha$-Stable density \citep{Zol(86)}. From \eqref{aslimit}, if $c_{\alpha,\theta,n}(m)=(\theta+n+m)^{\alpha}-(\theta+n)^{\alpha}$ then
\begin{equation}
K_{m}^{(n)}\stackrel{d}{\approx} c_{\alpha,\theta,n}S_{\alpha,\theta}^{(n,j)},
\label{ML approx}
\end{equation}
namely for $m$ sufficiently large the distribution of $K_{m}^{(n)}$ is approximated by means of the distribution of $c_{\alpha,\theta,n}(m)S_{\alpha,\theta}^{(n,j)}$. In particular, the scaling $c_{\alpha,\theta,n}(m)$ is determined in such a way that
\begin{displaymath}
\hat{K}_{n,m}=\E[K_{m}^{(n)}]=\E[c_{\alpha,\theta,n}S_{\alpha,\theta}^{(n,j)}].
\end{displaymath}
Mittag-Leffler credible intervals for $K_{n,m}$, centered on the BNP estimator $\hat{K}_{n,m}$, are then derived by Monte Carlo sampling $c_{\alpha,\theta,n}(m)S_{\alpha,\theta}^{(n,j)}$. We refer to \citet{Qu(21)} for recent developments on (exact) sampling $S_{\alpha,(\theta+n)/\alpha}$, improving over the original strategy of \citet{Fav(09)}. 

Instead, for $\alpha=0$, it follows from \citet[Theorem 2.3]{Kor(73)} that, as $m\rightarrow+\infty$
\begin{equation}\label{lim_dp_k}
\frac{K_{m}^{(n)}}{\log m}\stackrel{\text{w}}{\longrightarrow}\theta,
\end{equation}
namely the posterior distribution has a large $m$ limiting behaviour that is degenerate at $\theta>0$. See also \citet{Bal(23)} for details. Such a degenerate limiting behaviour prevents from extending the Monte Carlo sampling procedure of \citet{Fav(09)} to the case $\alpha=0$.


\section{Gaussian credible intervals}\label{sec3}

For the BNP approach to the unseen-species problem, we present a new method to derive large $m$ asymptotic credible intervals for $K_{n,m}$, which improves over the method of \citet{Fav(09)}: firstly, it allows to deal with $\alpha\in[0,1)$, thus including the case $\alpha=0$; secondly, it avoids the use of Monte Carlo sampling, enhancing computational efficiency. Following the notation of Section \ref{sec2}, for $\alpha\in[0,1)$ and $\theta>-\alpha$ let $\mathbf{X}_{n}$ by a collection of $n\geq1$ random samples from $\text{PYP}(\alpha,\theta)$, namely the model \eqref{eq:exchangeable_model}, such that $\mathbf{X}_{n}$ features $K_{n}=j$ species with (empirical) frequencies $(N_{1,n},\ldots,N_{1,K_{n}})=(n_{1},\ldots,n_{j})$. Our methodology relies on a large $m$ approximation of the posterior distribution of $K_{n,m}$, given $\mathbf{X}_{n}$, which is obtained assuming that both the sampling information $(n,j)$ and the parameter $\theta$ are large. In particular, we set $n=\nu m$, $j=\rho m$ and $\theta=\tau m$, with $\nu,\rho,\tau>0$, and provide a (weak) law of large numbers (LLN) and a Gaussian central limit theorem (CLT) for $K_{m}^{(n)}$, as $m\rightarrow+\infty$. The CLT is then applied to derive a large $m$ Gaussian approximation of the distribution of $K_{m}^{(n)}$. Here, we describe the approach that leads to the Gaussian credible intervals for $K_{n,m}$.

\subsection{Preliminary results}

For $\alpha\in[0,1)$ and $\theta>-\alpha$, let $K_{m}^{\ast}$ be the (random) number of species in $m\geq1$ random samples from $\text{PYP}(\alpha,\theta+n)$, such that $K^{\ast}_{m}\in\{1,\ldots,m\}$, and let $B_{a,b}$ be a Beta random variable with parameter $a,b>0$. If we denote by $Q(n, p)$ a Binomial random variable with parameter $n\in\mathbb{N}$ and $p\in(0,1)$, then according to \citet[Proposition 1]{Bal(23)} there hold:
\begin{itemize}
\item[i)] for $\alpha\in(0,1)$
\begin{equation}\label{bino_rap}
K_{m}^{(n)}\stackrel{\text{d}}{=}Q\left(K_{m}^{\ast},B_{\frac{\theta}{\alpha}+j,\frac{n}{\alpha}-j}\right);
\end{equation}
\item[ii)] for $\alpha=0$
\begin{equation}\label{bino_rap0}
K_{m}^{(n)}\stackrel{\text{d}}{=}Q\left(K_{m}^{\ast},\frac{\theta}{\theta+n}\right).
\end{equation}
\end{itemize}
As discussed in \citet{Bal(23)}, the compound Binomial representations \eqref{bino_rap}-\eqref{bino_rap0} follow from the quasi-conjugacy and conjugacy properties of the Pitman-Yor and Dirichlet priors, respectively. See also \citet{Dol(20),Dol(21)} for a more detailed account on \eqref{bino_rap}.

The next theorem from \citet{Con(24)} provides a LLN and a Gaussian CLT for $K^{\ast}_{m}$, as $m\rightarrow+\infty$, under the assumption that $\theta+n$ increases linearly in $m$, i.e. $\theta+n=\lambda m$, with $\lambda>0$.

\begin{thm}\label{thm_main}
For $m\in\mathbb{N}$ let $K^{\ast}_{m}\in\{1,\ldots,m\}$ be the (random) number of species in $m$ random samples from $\text{PYP}(\alpha,\theta+n)$, such that: $\alpha\in[0,1)$ and $\theta+n=\lambda m$, for some $\lambda>0$. If
\begin{displaymath}
\mathfrak{m}_{\alpha, \lambda}=\begin{cases} \frac{\lambda}{\alpha}\left[\left(1+\frac{1}{\lambda}\right)^{\alpha}-1\right] & \text{for } \alpha \in (0, 1)\\[0.4cm]
\lambda \log \left(1+\frac{1}{\lambda}\right) &  \text{for } \alpha = 0
\end{cases} 
\end{displaymath}
and
\begin{displaymath}
\mathfrak{s}_{\alpha, \lambda}^2= \begin{cases}  \frac{\lambda}{\alpha}\left[\left(1+\frac{1}{\lambda}\right)^{2\alpha}\left(1-\frac{\alpha}{1+\lambda}\right)-\left(1+\frac{1}{\lambda}\right)^{\alpha}\right]& \text{for } \alpha \in (0, 1)\\[0.4cm]
\lambda \log \left(1 + \frac{1}{\lambda}\right)  - \frac{\lambda}{\lambda +1}&  \text{for } \alpha = 0,
\end{cases}
\end{displaymath}
then, as $m\rightarrow+\infty$ there hold:
\begin{itemize}
\item[i)]
\begin{equation}\label{mom1}
\nonumber \mathbb{E}\left[K_m^*\right] =m\mathfrak{m}_{\alpha, \lambda} + O(1)
\end{equation}
and
\begin{equation}\label{mom2}
 \operatorname{Var}(K_m^*) =m \mathfrak{s}_{\alpha, \lambda}^2 + O(1);
\end{equation}
\item[ii)] 
\begin{equation}\label{lln}
\frac{K_{m}^*}{m}\stackrel{\text{p}}{\longrightarrow}\mathfrak{m}_{\alpha, \lambda};
\end{equation}
\item[iii)]
\begin{equation}\label{clt}
\frac{K_{m}^*-m\mathfrak{m}_{\alpha, \lambda}}{\sqrt{m\mathfrak{s}_{\alpha, \lambda}^{2}}}\stackrel{\text{w}}{\longrightarrow}N(0,1).
\end{equation}
\end{itemize}
Furthermore, for any $\lambda>0$ there hold that $\mathfrak{m}_{0, \lambda}=\lim_{\alpha\rightarrow0}\mathfrak{m}_{\alpha, \lambda}$ and also that $\mathfrak{s}_{0, \lambda}^2=\lim_{\alpha\rightarrow0}\mathfrak{s}_{\alpha, \lambda}^2$.
\end{thm}

\subsection{Main result}

By relying on the compound Binomial representations \eqref{bino_rap}-\eqref{bino_rap0} and Theorem \ref{thm_main}, the next theorem provides a LLN and a Gaussian CLT for $K_{m}^{(n)}$. See Appendix \ref{proof1} and Appendix \ref{proof2} for the proof.

\begin{thm}\label{thm_main_post}
For $n,m\in\mathbb{N}$ let $K^{(n)}_{m}\in\{0,1,\ldots,m\}$ be distributed according to the posterior distribution of $K_{n,m}$, given $\mathbf{X}_{n}$ from $\text{PYP}(\alpha,\theta)$, with $\mathbf{X}_{n}$ featuring $K_{n}=j$ species, such that: $n = \nu m$, $j = \varrho m$ and $\theta = \tau m$, for some $\nu,\rho,\tau>0$, and $\tau + \nu = \lambda$, for some $\lambda>0$. If
\begin{displaymath}
\mathscr{M}_{\alpha, \tau, \nu, \varrho}=\begin{cases}  \frac{\tau + \varrho \alpha }{\alpha } \left[ -1 + \left(  \frac{\lambda +1 }{\lambda}\right)^\alpha\right] & \text{for } \alpha \in (0, 1)\\[0.4cm]
\tau \log \left(1+\frac{1}{\lambda}\right) &  \text{for } \alpha = 0
\end{cases} 
\end{displaymath}
and
\begin{displaymath}
\mathscr{S}^{2}_{\alpha, \tau, \nu, \varrho} = \begin{cases} \frac{\tau + \varrho \alpha }{\lambda } \left(  \frac{\lambda +1 }{\lambda}\right)^\alpha \left\{ \frac{\lambda } {\alpha} \left[-1 + \left(  \frac{\lambda +1 }{\lambda}\right)^\alpha \right] - \frac{\tau + \varrho \alpha }{\lambda+1} \left(  \frac{\lambda +1 }{\lambda}\right)^\alpha\right\} & \text{for } \alpha \in (0, 1)\\[0.4cm]
\tau \log\left(1+ \frac{1}{\lambda}\right)  - \frac{\tau^2}{\lambda(\lambda+1)}&  \text{for } \alpha = 0,
\end{cases}
\end{displaymath}
then, as $m\rightarrow+\infty$ there hold:
\begin{itemize}
\item[i)] 
\begin{equation}\label{mom_post_1}
\nonumber \mathbb{E}\left[K_m^{(n)}\right]=m\mathscr{M}_{\alpha, \tau, \nu, \varrho}+ O(1)
\end{equation}
and
\begin{equation}\label{mom_post_2}
 \operatorname{Var}\left(K_m^{(n)}\right)=m \mathscr{S}_{\alpha, \tau, \nu, \varrho}^2 + O(1);
\end{equation}
\item[ii)]
\begin{equation}\label{lln_post}
\frac{K_m^{(n)}}{m}\stackrel{\text{p}}{\longrightarrow}\mathscr{M}_{\alpha, \tau, \nu, \varrho};
\end{equation}
\item[iii)]
\begin{equation}\label{clt_post}
\frac{K_m^{(n)}-m\mathscr{M}_{\alpha, \tau, \nu, \varrho}}{\sqrt{m\mathscr{S}_{\alpha, \tau, \nu, \varrho}^{2}}}\stackrel{\text{w}}{\longrightarrow}N(0,1).
\end{equation}
\end{itemize}
Furthermore, for any $\lambda>0$ there hold that $\mathscr{M}_{0, \tau, \nu, \varrho}=\lim_{\alpha\rightarrow0}\mathscr{M}_{\alpha, \tau, \nu, \varrho}$ and also that $\mathscr{S}_{0, \tau, \nu, \varrho}^2=\lim_{\alpha\rightarrow0}\mathscr{S}_{\alpha, \tau, \nu, \varrho}^2$.
\end{thm}

\subsubsection{Sketch of the proof of Theorem \ref{thm_main_post}}

The asymptotic expansions \eqref{mom_post_2} follow by combining \eqref{bino_rap} and \eqref{bino_rap0}, the asymptotic expansions \eqref{mom2} and standard properties of conditional expectation;  see sections \ref{sec:proof_asy_1} and \ref{sec:proof_asy_2} for details. The proof of the LLN \eqref{lln_post} relies on \eqref{mom_post_2} and Chebychev inequality; see section \ref{sec:lln_proof} for details. 

The proof of the CLT  \eqref{clt_post} relies on Theorem \ref{thm_main}, in combination with Proposition \ref{BE_Q} and Proposition \ref{p2} below. In particular, the structure of the proof is the same for $\alpha\in(0,1)$ and for $\alpha=0$, with some differences that are of technical nature, and can be appreciated in the proof of Proposition \ref{BE_Q}, which is deferred to Appendix \ref{sec:d1} (for $\alpha =0$) and to Appendix \ref{sec:d2} (for $\alpha \in (0, 1)$). From \eqref{bino_rap} and \eqref{bino_rap0}, which in the regime $\theta = \tau m ,\,  n = \nu m, \, j = \varrho m$, become
\begin{enumerate}[i)]
\item for $\alpha \in (0, 1)$,
\begin{displaymath} 
 K_m^{(n)} \stackrel{\text{d}}{ =} Q\left(K_m^*, \, B_{\left(\frac{\tau}{\alpha} + \varrho\right)\, m, \,  \left(\frac{\nu}{\alpha} - \varrho\right)\, m }\right)
\end{displaymath}
\item for $\alpha = 0$,
\begin{displaymath}
 K_m^{(n)} \stackrel{\text{d}}{ =} Q\left(K_m^*, \, \frac{\tau}{\lambda}\right),
 \end{displaymath}
 \end{enumerate}
 define

\begin{equation}
    \label{eq: def_Qm}
    Q_m(z) = \begin{cases} Q\left(\left\lfloor mz \right\rfloor, \, B_{\left(\frac{\tau}{\alpha} + \varrho\right)\, m, \,  \left(\frac{\nu}{\alpha} - \varrho\right)\, m }\right) & \text{if } \alpha \in (0, 1)\\
    Q\left(\left\lfloor mz \right\rfloor, \, \frac{\tau}{\tau + \nu}\right) & \text{if } \alpha =0
    \end{cases}
\end{equation}
where $\left\lfloor mz \right\rfloor : = \max\left(k \in \mathbb{N} \, : \, k \le mz\right)$. 

The next propositions are critical for the proof of Theorem \ref{thm_main_post}, showing how Theorem \ref{thm_main} interplay with the asymptotic expansions \eqref{mom_post_2}. The proof is deferred to  Appendix \ref{sec:d1} and to Appendix \ref{sec:d2}.
 
 \begin{proposition}[Berry-Esseen theorem for $Q_m(z)$] \label{BE_Q}
Let $\alpha\in[0,1)$ and $z>0$. Further, let $\mu:  (0, +\infty) \to \mathbb{R}$ and , $\sigma :  (0, +\infty) \to \mathbb{R}$ be functions defined as
\begin{equation*}
    \mu(z)  = z \cdot \frac{\tau + \varrho \alpha}{\tau + \nu}  
\end{equation*}
and
\begin{equation*}
    \sigma^2(z)  = z \cdot \frac{(\tau + \varrho \alpha)(\nu - \varrho \alpha)}{(\tau +\nu)^2}  \left[ 1+ \frac{\alpha z}{\tau + \nu} \right] 
\end{equation*}
Then, as $m \to +\infty$, 
    \begin{equation}
        \mathbb{E}[Q_m(z)] = m \, \mu(z) + O(1)
        \label{mean_Q}
    \end{equation}
    \begin{equation}
       \operatorname{Var}(Q_m(z)) = m \, \sigma^2(z) + O(1)
       \label{var_Q}
    \end{equation}
If
    \begin{equation*}
        V_m(z): = \frac{Q_m(z) - m\,  \mu(z)}{\sqrt{m \, \sigma^2(z)}}.
    \end{equation*}
then, for every $0< \zeta_0 < \mathfrak{m}_{\alpha, \lambda}< \zeta_1 <+\infty $ there exist  $\bar{m} = \bar{m}(\zeta_0, \zeta_1) \in \mathbb{N}$, and a continuous function $C = C_{\zeta_0, \zeta_1}: \left[\zeta_0, \zeta_1\right] \rightarrow (0, +\infty)$, such that for every $z \in \left[\zeta_0, \zeta_1\right]$ and every $m \ge \bar{m}$
    \begin{equation}
    \label{eq:be_Qm}
\left\|F_{V_m(z)}  - \Phi \right\|_\infty   \le C(z) \, \phi(m).
    \end{equation}
    where
    \begin{displaymath}
    \phi(m) = \begin{cases} m^{-\frac{1}{6}}  & \text{if } \alpha \in (0, 1)\\
    m^{-\frac{1}{2}} & \text{if } \alpha =0
    \end{cases}
    \end{displaymath}
 This implies in particular that for every $z \in \left[\zeta_0, \zeta_1\right]$, $ V_m(z) \stackrel{w}{\longrightarrow} N(0,1)$ as $m \rightarrow + \infty$.
\end{proposition}
\begin{proposition} \label{p2}
     For any $\alpha \in [0, 1)$, the following two equalities hold:
    \begin{equation}
        \label{eq:muQ(M)}
        \mu\left(\mathfrak{m}_{\alpha, \lambda}\right) = \mathscr{M}_{\alpha, \tau, \nu, \varrho}
    \end{equation}
    and
    \begin{equation}
        \label{eq:sigmaQ(M)}
        \sigma^2\left(\mathfrak{m}_{\alpha, \lambda}\right) + \mathfrak{s}_{\alpha, \lambda}^{2}\cdot \left(\mu'\left(\mathfrak{m}_{\alpha, \lambda}\right)\right)^2 = \mathscr{S}_{\alpha, \tau, \nu, \varrho}
^2.
    \end{equation}
\end{proposition}

Now, we conclude the proof of the CLT \eqref{clt_post}. The line of reasoning is the same as in \citet[Section 2.3]{Con(24)}, with $K_m^*$ playing the role of $Z_n$ and $Q_m(z)$ that of $R_n(z)$. Denoting by $F_{m}$ the cumulative distribution function of the random variable $\frac{K_m^{(n)} - m\mathscr{M}_{\alpha, \tau, \nu, \varrho}}{\sqrt{m} \mathscr{S}_{\alpha, \tau, \nu, \varrho}}$, our aim is to prove
\begin{displaymath}
\lim_{m \rightarrow + \infty } F_m(x) := \lim_{ m \rightarrow + \infty } P\left[K_m^{(n)} \le m\mathscr{M}_{\alpha, \tau, \nu, \varrho} + \sqrt{m} \mathscr{S}_{\alpha, \tau, \nu, \varrho} \, x \right] = \Phi(x),
\end{displaymath}
where $\Phi$ denotes the cumulative distribution function of the standard normal distribution. By resorting to the compound Binomial representations \eqref{bino_rap} (for $\alpha\in(0,1)$) and \eqref{bino_rap0} (for $\alpha=0$), and by means standard properties of conditional probability (see Appendix \ref{sec:d} for details), we write
\begin{equation}
\label{F=I1+I2}
F_m(x) = \mathcal{I}^{(m)}_1(x) + \mathcal{I}^{(m)}_2(x), 
\end{equation}
with
\begin{displaymath}
 \mathcal{I}^{(m)}_1(x) : =  \int_0^{+\infty}\Phi\left( \frac{\sqrt{m}\, \left[\mathscr{M}_{\alpha, \tau, \nu, \varrho} - \mu(z)\right] + \mathscr{S}_{\alpha, \tau, \nu, \varrho} \, x}{\sigma(z)} \right)\ \mu_{\frac{K_m^*}{m}}(\mathrm{d} z)
\end{displaymath}
and
\begin{align*}
 \mathcal{I}^{(m)}_2(x) & := \int_0^{+\infty} \left\{F_{V_m(z)}\left( \frac{\sqrt{m}\, \left[\mathscr{M}_{\alpha, \tau, \nu, \varrho} - \mu(z)\right] + \mathscr{S}_{\alpha, \tau, \nu, \varrho} \, x}{\sigma(z)} \right)  \right. \\ 
 & \quad  \left. - \Phi\left( \frac{\sqrt{m}\, \left[\mathscr{M}_{\alpha, \tau, \nu, \varrho} - \mu(z)\right] + \mathscr{S}_{\alpha, \tau, \nu, \varrho} \, x}{\sigma(z)} \right)\right\}\ \mu_{\frac{K_m^*}{m}}(\mathrm{d} z).
\end{align*}
Accordingly, the proof of the CLT \eqref{clt_post} is completed by showing that, for every $x\in\mathbb{R}$ there hold
\begin{equation}\label{rv_1}
\lim_{m \to + \infty} \mathcal{I}^{(m)}_1(x) =\Phi(x)
\end{equation}
and
\begin{equation}\label{rv_2}
\lim_{m \to + \infty} \mathcal{I}^{(m)}_2(x) =0.
\end{equation}
This is done in Proposition \ref{lem:lem3_ter} (for Equation \eqref{rv_1}) and Proposition \ref{lem:lem4_ter} (for Equation \eqref{rv_2}) below.
\begin{proposition}\label{lem:lem3_ter} 
For every $x \in \mathbb{R}$, 
\begin{displaymath}
\lim_{m \to + \infty} \mathcal{I}^{(m)}_1(x) :=\lim_{m \rightarrow+\infty} \mathbb{E} \left[ \Phi \left(\frac{\sqrt{m} \left[\mu\left(\frac{K_m^*}{m}\right)- \mathscr{M}_{\alpha, \tau, \nu, \varrho}\right] +\mathscr{S}_{\alpha, \tau, \nu, \varrho }\, x}{\sigma\left(\frac{K_m^*}{m}\right)}\right)\right] =\Phi(x).
\end{displaymath}
\end{proposition}
\begin{proof}
It follows from Proposition \ref{p2} and \eqref{clt} of Theorem \ref{thm_main}, trough Lemma \ref{lem:lemma2_ter}, that as $m \to +\infty$
\begin{displaymath}
\sqrt{n} \left[ \mu\left(\frac{K_m^*}{m}\right) - \mu(\mathfrak{m}_{\alpha, \lambda})\right]\stackrel{\text{w}}{\longrightarrow} \mu'(\mathfrak{m}_{\alpha, \lambda})Y.
\end{displaymath}
Further, since $\sigma$ is a continuous function, by means of \eqref{lln}, as $m\rightarrow+\infty$, it holds
\begin{displaymath}
\sigma\left(\frac{K_m^*}{m}\right)\stackrel{\text{p}}{\longrightarrow} \sigma(\mathfrak{m}_{\alpha, \lambda}).
\end{displaymath}
Since $\mathscr{M}_{\alpha, \tau, \nu, \varrho} = \mu(\mathfrak{m}_{\alpha, \lambda})$ by Proposition \ref{p2}, by means of Slutsky's theorem, as $m\rightarrow+\infty$
\begin{displaymath}
\frac{\sqrt{m} \left[\mu\left(\frac{K_m^*}{m}\right)- \mathscr{M}_{\alpha, \tau, \nu, \varrho}\right] +\mathscr{S}_{\alpha, \tau, \nu, \varrho }\, x}{\sigma\left(\frac{K_m^*}{m}\right)}\stackrel{\text{w}}{\longrightarrow}\frac{\mu'(\mathfrak{m}_{\alpha, \lambda})Y + \mathscr{S}_{\alpha, \tau, \nu, \varrho} \, x}{\sigma\left(\mathfrak{m}_{\alpha, \lambda}\right)}.
\end{displaymath}
Since $\Phi$ is a bounded and continuous function, the proof is completed by Portmanteau theorem.
\end{proof}

\begin{proposition}\label{lem:lem4_ter}
For every $x \in \mathbb{R}$, 
\begin{displaymath}
\left| \mathcal{I}^{(m)}_2(x) \right| \le \int_0^{+\infty} \left\|F_{V_m(z)} - \Phi \right\|_{\infty} \mu_{\frac{K_m^*}{m}}(\mathrm{d} z)\rightarrow 0.
\end{displaymath}
\end{proposition}
\begin{proof}
Fix $\varepsilon>0$ and choose $\delta = \delta(\varepsilon)>0$ such that $\Phi(\delta) = 1 - \varepsilon/2$. Now, let
\begin{displaymath}
\tilde{m} = \tilde{m}(\varepsilon, \zeta_0, \zeta_1): = \min \left\{m \in \mathbb{N}\text{ : } z_0 - \frac{\delta \mathfrak{s}_{\alpha, \lambda}}{\sqrt{m}} >\zeta_0,\text{ and } z_0 - \frac{\delta \Sigma}{\sqrt{m}} <\zeta_1, \right\},
\end{displaymath}
which exists because $\mathfrak{m}_{\alpha, \lambda} \in [\zeta_0, \zeta_1]$. We set $\tilde{\zeta}_0 := z_0 - \frac{\delta \mathfrak{s}_{\alpha, \lambda}}{\sqrt{\tilde{m}}}$ and $\tilde{\zeta}_1 := z_0 + \frac{\delta \mathfrak{s}_{\alpha, \lambda}}{\sqrt{\tilde{m}}} $, and write
\begin{align}\label{par31}
&\int_0^{+\infty} \left\|F_{V_m(z)} - \Phi \right\|_{\infty} \mu_{\frac{K_m^*}{m}}(\mathrm{d} z)\\
&\notag\quad = \int_{\tilde{\zeta}_0}^{\tilde{\zeta}_1} \left\|F_{V_m(z)} - \Phi \right\|_{\infty} \mu_{\frac{K_m^*}{m}}(\mathrm{d} z) + \int_{\mathbb{R} \smallsetminus [\tilde{\zeta}_0, \tilde{\zeta}_1]} \left\|F_{V_m(z)} - \Phi \right\|_{\infty} \mu_{\frac{K_m^*}{m}}(\mathrm{d} z)
\end{align}
and treat the terms on the right-hand side of \eqref{par31} separately. For the first term on the right-hand side of \eqref{par31}, by means of \eqref{eq:be_Qm}, for every $m \ge \bar{m} $ it holds
\begin{equation}\label{par32}
\int_{\tilde{\zeta}_0}^{\tilde{\zeta}_1}\left\|F_{V_m(z)} - \Phi \right\|_{\infty} \mu_{\frac{K_m^*}{m}}(\mathrm{d} z) \le \phi(m) \int_{\tilde{\zeta_0}}^{\tilde{\zeta}_1}C(z) \ \mu_{\frac{K_m^*}{m}}(\mathrm{d} z) \le \phi(m) \, \mathcal{M}_C,
\end{equation}
where $\mathcal{M}_C: = \max_{z \in [\zeta_0, \zeta_1]} C(z)$ exists because $C$ is continuous; in particular, $\mathcal{M}_C$ is bounded on $[\tilde{\zeta}_0, \tilde{\zeta}_1]$. For the second term on the right-hand side of \eqref{par31},
\begin{align}\label{par33}
\int_{\mathbb{R} \smallsetminus [\tilde{\zeta}_0, \tilde{\zeta}_1]} \left\|F_{V_m(z)} - \Phi \right\|_{\infty} \mu_{\frac{K_m^*}{m}}(\mathrm{d} z) & \le \int_{\mathbb{R} \smallsetminus [\tilde{\zeta}_0, \tilde{\zeta}_1]}\mu_{\frac{K_m^*}{m}}(\mathrm{d} z) \\
&\notag = P\left[\frac{K_m^*}{m} \notin \left[\mathfrak{m}_{\alpha, \lambda}- \frac{\delta \mathfrak{s}_{\alpha, \lambda}}{\sqrt{\tilde{m}}},\mathfrak{m}_{\alpha, \lambda} + \frac{\delta \mathfrak{s}_{\alpha, \lambda}}{\sqrt{\tilde{m}}} \right] \right]\\
&\notag\le P\left[\frac{K_m^*}{m} \notin \left[\mathfrak{m}_{\alpha, \lambda} - \frac{\delta \mathfrak{s}_{\alpha, \lambda}}{\sqrt{m}},\mathfrak{m}_{\alpha, \lambda} + \frac{\delta \mathfrak{s}_{\alpha, \lambda}}{\sqrt{m}} \right] \right]\\
&\notag = P \left[\frac{K_m^*- m \,  \mathfrak{m}_{\alpha, \lambda}}{\sqrt{m} \mathfrak{s}_{\alpha, \lambda}} \notin [-\delta, \delta]\right].
\end{align}
From \eqref{par31}, with \eqref{par32} and \eqref{par33}, we obtained that, for every $n \ge \max(\bar{m}, \tilde{m})$ it holds
\begin{displaymath}
\int_0^{+\infty} \left\|F_{V_m(z)} - \Phi \right\|_{\infty} \mu_{\frac{K_m^*}{m}}(\mathrm{d} z) \le \phi(m) \, \mathcal{M}_C + P \left[\frac{K_m^*- m \, \mathfrak{m}_{\alpha, \lambda}}{\sqrt{m} \mathfrak{s}_{\alpha, \lambda}} \notin [-\delta, \delta]\right].
\end{displaymath}
Hence,
\begin{align*}
0  & \le \limsup_{m \to + \infty} \int_0^{+\infty} \left\|F_{V_m(z)} - \Phi \right\|_{\infty} \mu_{\frac{K_m^*}{m}}(\mathrm{d} z) \\
& \le \lim_{m \to + \infty} \left\{\phi(m) \, \mathcal{M}_C + P \left[\frac{K_m^*- m \, \mathfrak{m}_{\alpha, \lambda}}{\sqrt{m} \, \mathfrak{s}_{\alpha, \lambda}}\notin [-\delta, \delta]\right]\right\}\\
& = 0+ 2-2 \Phi(\delta)\\
& = \varepsilon,
\end{align*}
where the last identities follows from \eqref{clt} and the definition of $\delta$ such that $\Phi(\delta) = 1 - \varepsilon/2$, respectively. The proof is completed by the arbitrariness of $\varepsilon$. 
\end{proof}

This completes the proof of the CLT \eqref{clt_post}, and the proof of Theorem \ref{thm_main_post}. See Appendix \ref{proof1} and Appendix \ref{sec:d}.

\subsection{Credible intervals}
Theorem \ref{thm_main_post} provides  a large $m$ Gaussian approximation of the distribution of $K_{m}^{(n)}$; in particular the approximation is centered on the BNP estimator $\hat{K}_{n,m}$. From \eqref{clt_post},
\begin{displaymath}
K_{m}^{(n)} \stackrel{d}{\approx} N \left(m \mathscr{M}_{\alpha, \tau, \nu, \varrho},\,  m\mathscr{S}_{\alpha, \tau, \nu, \varrho}^{2} \right)
\end{displaymath}
namely for $m$ sufficiently large the distribution of $K_{m}^{(n)}$ is approximated by a Gaussian distribution with mean $m \mathscr{M}_{\alpha, \tau, \nu, \varrho}$ and variance $m\mathscr{S}_{\alpha, \tau, \nu, \varrho}^{2}$. This approximation is applied to construct Gaussian credible intervals for $K_{n,m}$ with a prescribed (asymptotic) probability level for $K_m^{(n)}$. Given $\delta \in (0, 1)$, a $(1-\delta)$-level symmetric credible interval centered at the mean is
\begin{equation}\label{eq: intervals}
\text{Pr} \left[ K_{m}^{(n)} \in \left[ m \mathscr{M}_{\alpha, \tau, \nu, \varrho}  -q_{\frac{\delta}{2}} \sqrt{m\mathscr{S}_{\alpha, \tau, \nu, \varrho}^{2}}, \,  m \mathscr{M}_{\alpha, \tau, \nu, \varrho}  - q_{1- \frac{\delta}{2} } \sqrt{m\mathscr{S}_{\alpha, \tau, \nu, \varrho}^{2}}\right]  \right] \approx 1-\delta,
\end{equation}
where $q_\varepsilon := \Phi^{-1}(\varepsilon)$ is the $\varepsilon$-quantile of the standard Gaussian distribution. The Gaussian credible interval \eqref{eq: intervals} is fully analytical, thus avoiding the use of Monte Carlo sampling algorithms.


\section{Numerical illustrations}\label{sec4}

We validate our methodology on synthetic and real data, comparing its performance with that of state-of-the-art procedures for the construction of exact  and asymptotic credible intervals for $K_{m, n}$.

Under the Pitman-Yor prior, the BNP approach to estimate $K_{n,m}$ requires the specification of the prior's parameters $\alpha\in[0,1)$ and $\theta>-\alpha$. Here, we adopt an empirical Bayes approach to estimate $(\alpha,\theta)$, which relies on the marginal distribution of $\mathbf{X}_{n}=(X_{1},\ldots,X_{n})$ from $\text{PYP}(\alpha,\theta)$ \citep[Section 3]{Fav(09)}. For $\mathbf{X}_{n}$ featuring $K_n$ distinct species with (empirical) frequencies $\mathbf{N}_{n} = (N_{1,n},\dots,N_{K_{n},n})$, such a distribution is Ewens-Pitman’s formula \citep{Pit(95)},
\begin{equation}
\label{EPSampling}
Pr\left[ K_n = j, \mathbf{N}_{n} = \mathbf{n} \right] = \frac{\prod_{i = 1}^{k-1}(\theta+ i \alpha)}{(\theta + 1)_{(n-1)}} \, \prod_{j = 1}^k (1-\alpha)_{(n_j -1)} .
\end{equation}
The empirical Bayes approach estimates the parameters $(\alpha, \theta)$ with the values $(\hat{\alpha}_n, \hat{\theta}_n)$ that maximizes the distribution (marginal likelihood) \eqref{EPSampling} corresponding to the observed sample $(j, n_1, \dots, n_j)$, i.e.,
\begin{equation}
\label{MLE}
(\hat{\alpha}_n, \hat{\theta}_n) = \arg \max_{\{(\alpha, \theta) \, : \, \alpha \in [0, 1), \theta >-\alpha\}} \left\{ \frac{\prod_{i = 1}^{k-1}(\theta+ i \alpha)}{(\theta + 1)_{(n-1)}} \, \prod_{j = 1}^k (1-\alpha)_{(n_j -1)} \right\}.
\end{equation}
Alternatively, one may specify a prior on $(\alpha,\theta)$ and pursue fully Bayes estimates of prior's parameters. Due to the large sample size of the datasets considered, we do not expect substantial differences between empirical Bayes and fully Bayes. See \citet[Section 6]{Bal(23)} for a discussion of the estimation of $(\alpha,\theta)$ and issues thereof, especially with respect to $\theta>0$.

\subsection{Synthetic data}
We test our method on synthetic data generated from a collection of discrete distributions taken from \citet[Figure 3]{Orl(17)}. More precisely, we generate four datasets, which correspond to the four panels of Figure \ref{fig1_intro}: A) Zipf distribution on $\{0, 1, . . . , 300\}$ with parameter 2; B) Zipf distribution on $\{0, 1, . . . , 100\}$ with parameter $1.5$; C) P\'olya distribution on $\{0, 1, . . . , 500\}$ with parameter $\mathbf{\beta} = (\beta_1, ..., \beta_{500})$ with $\beta_1 = \beta_2 = 2$ and $\beta_i = 500$ for all $i \ge 3$; D) Uniform distribution on $\{0, 1, . . . , 500\}$. For each dataset, Table \ref{tab:0} collects the sample size $n$, the number of distinct species $j$, and the empirical Bayes estimates $(\hat{\alpha}_n, \hat{\theta}_n)$ of $(\alpha,\theta)$.

\begin{table}[ht]
\footnotesize
    \begin{tabularx}{\textwidth}{
    >{\raggedright \hsize=1.5\hsize \arraybackslash}      
   >{\centering \hsize=0.3\hsize \arraybackslash}X 
    >{\centering \hsize=0.3\hsize \arraybackslash}X
    >{\centering \hsize=0.3\hsize \arraybackslash}X
   >{\centering \hsize=0.3 \hsize \arraybackslash}X
>{\centering \hsize=0.3 \hsize \arraybackslash}X}
\toprule
 Dataset &  $n$ & $j$ & $\hat{\alpha}_n$ &  $\hat{\theta}_n$  \\
\midrule
A) Zipf  & 977   & 300 &  0.54   & 26.67 \\[1ex]
B) Zipf  & 1877   & 100 &  0.38   & 4.66 \\[1ex]
C) P\'olya  & 2000   & 215&  0.64   & 2.39\\[1ex]
D) Uniform & 2000   & 447 & 0   & 178.48\\
\bottomrule
\end{tabularx}
\caption{\scriptsize{Sample size $n$, number of distinct species $j$ in the sample, and empirical Bayes estimates of $(\alpha,\theta)$ for datasets A, B, C and D.}}
\label{tab:0} 
\end{table}

For datasets A, B, C and D, Figure \ref{fig1_intro} displays the BNP estimates of $K_{n,m}$ with the 95\% exact credible intervals, Mittag-Leffler credible intervals and Gaussian credible intervals as a function of $m$, for $m \in [0, 5n]$. Table \ref{tab:1} contains the BNP estimates of $K_{n,m}$ for $m=2,\,2n,\,3,\,4n,\,5n$, and the corresponding 95\% exact credible intervals, Mittag-Leffler credible intervals with their coverages of exact intervals, and Gaussian credible intervals with their coverages of exact intervals. The coverage is defined as the ratio between the length of the intersection of the (rounded to the nearest integer) exact credible interval with the (rounded) approximate credible interval, and the length of the (rounded) exact interval. Table \ref{tab:2} contains the same quantities as Table \ref{tab:1}, for larger values of $m$. We refer to Appendix \ref{app:E} for figures displaying the coverage of Mittag-Leffler and Gaussian credible intervals as a function of $m \in [0, 5n]$. Exact intervals in Table \ref{tab:1} are derived by Monte Carlo sampling, through the inverse transform algorithm, the posterior distribution of $K_{n,m}$, given $\mathbf{X}_{n}$. Instead, due to the larger values of $m$, the exact intervals in Table \ref{tab:2} are derived by Monte Carlo sampling, through Algorithm \ref{alg:k}. In both algorithms, we applied 2000 Monte Carlo samples.

\begin{table}[ht]
\scriptsize
    \begin{tabularx}{\textwidth}{
    >{\raggedright \hsize=0.8\hsize \arraybackslash}      
   >{\centering \hsize=0.5\hsize \arraybackslash}X 
    >{\centering \hsize=0.1\hsize \arraybackslash}X
    >{\centering \hsize=0.2\hsize \arraybackslash}X
   >{\centering \hsize=0.3 \hsize \arraybackslash}X
    >{\centering \hsize=0.3 \hsize \arraybackslash}X
     >{\centering \hsize=0.25 \hsize \arraybackslash}X
>{\centering \hsize=0.3 \hsize \arraybackslash}X
>{\centering \hsize=0.25 \hsize \arraybackslash}X}
\toprule
\multirow{2}{*}{Dataset}  & \multirow{2}{*}{$m$}  & \multirow{2}{*}{$\hat{K}_{n,m}$} &\multirow{2}{*}{95\% Exact C.I.} &  \multicolumn{2}{c}{Mittag-Leffler C.I.}   &  \multicolumn{2}{c}{Gaussian C.I.}\\ 
\cmidrule(lr){5-6} \cmidrule(lr){7-8}
 &  &  & & 95\% C.I. & Coverage (\%) & 95\% C.I. & Coverage (\%) \\
\midrule
A) Zipf, $n = 977$ &  $n$   & 156 & (130, 184)   &  (141, 173) & 59.3 &  (129, 183) & 98.1 \\
 & $2n$   & 280 &  (241, 321)   & (252, 309) & 71.3 & (239, 320) & 98.8\\
& $3n$   & 386 &  (335, 439)   & (348, 426) & 75.0 & (334, 437) & 98.1 \\
& $4n$   & 480 &  (423, 541)    & (433, 530) & 82.2 & (419, 541)& 100\\
& $5n$   & 566 &  (501, 638)   & (511, 625)& 83.2 & (496, 636) & 98.5  \\[1ex]
 B) Zipf, $n = 1877$  &  $n$   & 33 & (22, 47)    &  (28, 40) & 48.0 &  (21, 46) & 96.0\\
 & $2n$   & 57 &  (40, 77)   & (47, 69) &  59.5 & (39, 76)& 97.3 \\
& $3n$   & 77 &  (57, 102)   & (63, 92) & 64.4 & (55, 99) & 93.3 \\
& $4n$   & 93 &  (69, 119)   & (77, 112)& 70.0& (68, 119) & 100 \\
& $5n$   & 108 &  (80, 137)  & (89, 129)  & 70.2 &(80, 136) & 98.2 \\[1ex]
C) P\'olya, $n = 2000$  & $n$   & 122 &(98, 149)  & (107, 139)& 62.7 & (96, 149) &100 \\
& $2n$   & 224 &  (185, 265)  & (195, 254)& 73.8 & (183, 264) & 98.8\\
& $3n$   & 313 &  (263, 369)    & (273, 356)& 78.3 & (261, 366) & 97.2\\
& $4n$   & 395 &  (334, 460)  & (344, 449) & 83.3 & (332, 458)  & 98.4 \\
& $5n$   & 471 &  (398, 549)  & (410, 535) & 82.8 & (398, 544) & 96.7 \\[1ex]
D) Uniform, $n = 2000$  &  $n$   & 116 &   (96, 137) &--& -- & (96, 137) & 100 \\
 & $2n$   & 186 &  (160, 211)   & -- & -- & (160, 212) & 100\\
& $3n$   & 236 &  (206, 265)   & -- & -- & (207, 266) & 98.3 \\
& $4n$   & 275 &  (244, 309)   & -- & -- & (243, 307) &96.9\\
& $5n$   & 307 &  (274, 341)   & -- & -- & (273, 341) &100 \\[1ex]
\bottomrule
\end{tabularx}
\caption{\scriptsize{Additional sample $m$, BNP estimates of $K_{n,m}$,  95\% exact C.I., Mittag-Leffler C.I. with their coverages (of the the exact C.I.) and Gaussian credible intervals with their coverages (of the exact C.I.). All values are rounded to the nearest integer.}}
\label{tab:1} 
\end{table}

\begin{table}[ht]
\scriptsize
    \begin{tabularx}{\textwidth}{
    >{\raggedright \hsize=0.8\hsize \arraybackslash}      
   >{\centering \hsize=0.5\hsize \arraybackslash}X 
    >{\centering \hsize=0.1\hsize \arraybackslash}X
    >{\centering \hsize=0.2\hsize \arraybackslash}X
   >{\centering \hsize=0.3 \hsize \arraybackslash}X
    >{\centering \hsize=0.3 \hsize \arraybackslash}X
     >{\centering \hsize=0.25 \hsize \arraybackslash}X
>{\centering \hsize=0.3 \hsize \arraybackslash}X
>{\centering \hsize=0.25 \hsize \arraybackslash}X}
\toprule
\multirow{2}{*}{Dataset}  & \multirow{2}{*}{$m$}  & \multirow{2}{*}{$\hat{K}_{n,m}$} & \multirow{2}{*}{95\% Exact C.I.} &  \multicolumn{2}{c}{Mittag-Leffler C.I.}   &  \multicolumn{2}{c}{Gaussian C.I.}\\ 
\cmidrule(lr){5-6} \cmidrule(lr){7-8}
 &  &  & & 95\% C.I. & Coverage (\%) & 95\% C.I. & Coverage (\%) \\
\midrule
 A) Zipf, $n = 977$ &  $10n$   & 923 & (813, 1030)   &  (834, 1019) & 85.3 &  (817, 1029) & 97.7 \\
 & $50n$   & 2582 &  (2299, 2863)   & (2332, 2851) & 92.0 & (2311, 2854) & 96.3\\
& $100n$   & 3904 &  (3519, 4304)   & (3526, 4311) & 99.1 & (3501, 4307) & 100 \\
& $1000n$   & 14493 &  (13064, 16031)   & (13088, 16002)& 98.2 & (13038, 15949) & 97.2  \\[1ex]
 B) Zipf, $n = 1877$  &  $10n$   & 165 & (127, 205)    &  (134, 196) & 79.5 &  (125, 204) & 98.7\\
 & $50n$   & 381 &  (304, 467)   & (309, 453) &  88.3 & (301, 460)& 95.7 \\
& $100n$   & 525 &  (421, 636)   & (427, 625) & 92.1 & (419, 632) & 98.1 \\
& $1000n$   & 1400 &  (1131, 1667)  & (1137, 1664)  & 98.1 &(1132, 1668) & 99.8 \\[1ex]
 C) P\'olya, $n = 2000$  & $10n$   & 799 & (685, 913)  & (701, 908) & 90.8 & (682, 915) &100 \\
& $50n$   & 2502 &  (2171, 2857)  & (2195, 2845)&94.8 & (2165, 2839) & 97.4\\
& $100n$   & 3998 &  (3491, 4547)    & (3508, 4547)& 98.4 & (3467, 4529) & 98.3\\
& $1000n$   & 18139 &  (15824, 20559)  & (15915, 20629) & 98.1 & (15776, 20501) & 98.8 \\[1ex]
D) Uniform, $n = 2000$  &  $10n$   & 414 &   (375, 457) &--& -- & (375, 453) & 95.1 \\
 & $50n$   & 687 &  (636, 739)   & -- & -- & (636, 738) & 99.0\\
& $100n$   & 809 &  (757, 864)   & -- & -- & (753, 864) & 100\\
& $1000n$   & 1218 &  (1150, 1286)   & -- & -- & (1150, 1286) &100 \\
\bottomrule
\end{tabularx}
\caption{\scriptsize{Additional sample $m$, BNP estimates of $K_{n,m}$,  95\% exact C.I., Mittag-Leffler C.I. with their coverages (of the the exact C.I.) and Gaussian credible intervals with their coverages (of the exact C.I.). All values are rounded to the nearest integer.}}
\label{tab:2} 
\end{table}

Figure \ref{fig1_intro}, as well as from Table \ref{tab:1} and the rows of Table \ref{tab:2} corresponding to $m = 10n$ and $m=50n$, show that Mittag-Leffler credible intervals have a smaller coverage than Gaussian credible intervals. As expected, the performance of Mittag-Leffler credible intervals improves as $m$ grows. In particular, for  $m=100 n$ and $m=1000 n$ the coverage of the Mittag-Leffler credible intervals eventually matches or outperforms that of the Gaussian credible intervals, which in any case maintain a coverage of at least 97\%. Such a behaviour provides a further indication on the magnitude of the additional sample size $m$ for which the Mittag-Leffler regime of approximation may be more suited than the Gaussian regime. Furthermore, the Gaussian intervals appear to be near-to-optimal, in the sense that their length is always comparable to that of the exact intervals even when they are larger. For the Uniform dataset, where the empirical Bayes estimate for $\alpha$ is $0$, Gaussian credible intervals sill work, preserving their good performance in terms of coverage of the exact interval. Instead, the method of \citet{Fav(09)} does not apply for $\alpha=0$, as discussed in Section \ref{sec 2.3}.

\subsection{Real data}
For real data analysis, we consider the same Expressed Sequence Tags (EST) datasets previously analyzed in \citet{Fav(09)}. These datasets are generated by sequencing cDNA libraries  consisting of millions of genes and one of the main quantities of interest is the number of distinct genes. Due to the cost of the sequencing procedure, only a small portion of the cDNA library is typically sequenced. Given the resulting sequenced sample of size $n$, it is required to estimate the number of new genes $K_{m,n}$ to appear in an additional sample of size $m$. On the basis of such estimates, geneticists decide whether it is worth proceeding with sequencing and, if so, up to which additional sample size. The five libraries considered are: i) tomato flower cDNA library \citep{Qua(00)}; ii) two cDNA libraries of the amitochondriate protist \emph{Mastigamoeba balamuthi}, one of which has undergone a normalization protocol \citep{Sus(04)}; iii)two \emph{Naegleria gruberi} cDNA libraries from cells grown respectively in aerobic and anaerobic conditions \citep{Sus(04)}.

For each EST dataset, Table \ref{table_fav} collects the sample size $n$, the number of distinct species $j$, and the empirical Bayes estimates $(\hat{\alpha}_n, \hat{\theta}_n)$ of $(\alpha,\theta)$; Table \ref{table_fav} coincides with \citet[Table 1]{Fav(09)}, where the same datasets are analyzed. For these datasets, Table \ref{tab:3} contains the BNP estimates of $K_{n,m}$ for $m=2,\,2n,\,3,\,4n,\,5n$, and the corresponding 95\% exact credible intervals, Mittag-Leffler credible intervals with their coverages of exact intervals, and Gaussian credible intervals with their coverages of exact intervals.  Table \ref{tab:4} contains the same quantities as Table \ref{tab:3}, for larger values of $m$. We refer to Appendix \ref{app:E} for figures displaying the coverage of Mittag-Leffler and Gaussian credible intervals as a function of $m \in [0, 5n]$.

\begin{table}[ht]
\footnotesize
    \begin{tabularx}{\textwidth}{
    >{\raggedright \hsize=1.7\hsize \arraybackslash}      
   >{\centering \hsize=0.3\hsize \arraybackslash}X 
    >{\centering \hsize=0.3\hsize \arraybackslash}X
    >{\centering \hsize=0.3\hsize \arraybackslash}X
   >{\centering \hsize=0.3 \hsize \arraybackslash}X
>{\centering \hsize=0.3 \hsize \arraybackslash}X}
\toprule
 Library  &  $n$ & $j$  & $\hat{\alpha}_n$ &  $\hat{\theta}_n$  \\[1ex]
\midrule
 Tomato flower  & 2586   & 1825 &  0.612   & 741.0 \\[1ex]
  \emph{Mastigamoeba} & 715   & 460 &  0.770   & 46.0 \\[1ex]
  \emph{Mastigamoeba}--normalized & 363   & 248 &  0.700   & 57.0\\[1ex]
  \emph{Naegleria} aerobic  & 959   & 473 & 0.670   & 46.3\\[1ex]
  \emph{Naegleria} anaerobic  & 969   & 631  &  0.660   & 155.5 \\
  \bottomrule
\end{tabularx}
\caption{\scriptsize{Sample size $n$, number of distinct species $j$ in the sample, and empirical Bayes estimates of $(\alpha,\theta)$ for the five EST datasets.}}
\label{table_fav}
\end{table}

\begin{table}[h!]
\scriptsize
    \begin{tabularx}{\textwidth}{
    >{\raggedright \hsize=0.8\hsize \arraybackslash}      
   >{\centering \hsize=0.5\hsize \arraybackslash}X 
    >{\centering \hsize=0.1\hsize \arraybackslash}X
    >{\centering \hsize=0.2\hsize \arraybackslash}X
   >{\centering \hsize=0.3 \hsize \arraybackslash}X
    >{\centering \hsize=0.3 \hsize \arraybackslash}X
     >{\centering \hsize=0.25 \hsize \arraybackslash}X
>{\centering \hsize=0.3 \hsize \arraybackslash}X
>{\centering \hsize=0.25 \hsize \arraybackslash}X}
\toprule
\multirow{2}{*}{Library}  & \multirow{2}{*}{$m$}  & \multirow{2}{*}{$\hat{K}_{n,m}$} &\multirow{2}{*}{95\% Exact C.I.} &  \multicolumn{2}{c}{Mittag-Leffler C.I.}   &  \multicolumn{2}{c}{Gaussian C.I.}\\ 
\cmidrule(lr){5-6} \cmidrule(lr){7-8}
 &   && &  95\% C.I. & Coverage (\%) & 95\% C.I. & Coverage (\%) \\
\midrule
 Tomato flower &  $n$   & 1281 & (1223, 1339)   &  (1244, 1321) & 66.4 &  (1222, 1340) & 100 \\
$n = 2586$ & $2n$   & 2354 &  (2264, 2446)   & (2287, 2427) & 76.9 & (2262, 2445) & 99.5\\
& $3n$   & 3305 &  (3184, 3424)   & (3211, 3409) & 82.5 & (3185, 3425) & 99.6 \\
& $4n$   & 4173 &  (4031, 4318)    & (4054, 4304)& 87.1 & (4028, 4319)& 100 \\
& $5n$   & 4980 &  (4815, 5146)   & (4838, 5136)& 90.0 & (4811, 5148) & 100  \\[1ex]
  \emph{Mastigamoeba}  &  $n$   & 346 & (312, 379)    &  (323, 369) & 68.7 &  (312, 379) &100\\
$n = 715$ & $2n$   & 654 &  (596, 706)   & (610, 697) &79.1 & (599, 708)&97.3 \\
& $3n$   & 939 &  (866, 1014)   & (875, 1001) & 85.1 & (865, 1012) & 98.6 \\
& $4n$   & 1208 &  (1119, 1301)   & (1126, 1288)& 89.0 & (1116, 1299) &98.9 \\
& $5n$   & 1465 &  (1357, 1578)  & (1366, 1562)  & 88.7 &(1356, 1573) & 97.7 \\[1ex]
 \emph{Mastigamoeba}--norm. & $n$   & 180 &(157, 202)  & (164, 197)&73.3 & (157, 203) &100 \\
$n = 363$ & $2n$   & 336 &  (299, 371)  & (306, 367)&84.7 & (299, 372) &100\\
& $3n$   & 477 &  (429, 525)    & (435, 522)& 90.6 & (428, 526) & 100\\
& $4n$   & 608 &  (546, 671)  & (555, 666) &88.8 & (548, 668)  & 96.0\\
& $5n$   & 732 &  (660, 803)  & (668, 801) &93.0 & (662, 803) &98.6\\[1ex]
 \emph{Naegleria} aerobic  &  $n$   & 307 & (272, 343)  &  (284, 331) & 66.2 & (272, 342) & 98.6 \\
$n = 959$ & $2n$   & 566 &  (514, 621)   & (524, 611) & 81.3 & (511, 622) & 100\\
& $3n$   & 798 &  (730, 873)   & (739, 861) & 85.3 & (726, 871) & 98.6 \\
& $4n$   & 1012 &  (921, 1099)   & (937, 1091) & 86.5 & (923, 1101) & 98.9\\
& $5n$   & 1212 &  (1108, 1319)   & (1122, 1307) & 87.7 & (1109, 1315) &97.6 \\[1ex]
 \emph{Naegleria} anaerobic  &  $n$   & 439 & (402, 476)  &  (415, 465) & 67.6 &  (402, 476)  & 100 \\
$n = 969$ & $2n$   & 812 &  (753, 871)  & (767, 860) & 78.8& (754, 870) & 98.3   \\
& $3n$   & 1146 &  (1065, 1219)  & (1083, 1213) & 84.4 & (1069, 1223) & 97.4  \\
& $4n$   & 1454 &  (1365, 1550)   & (1373, 1538) & 89.2 &  (1360, 1547) & 98.4\\
& $5n$   & 1741 &  (1635, 1855)   & (1645, 1843) & 90.0 & (1632, 1851) & 98.2 \\
\bottomrule
\end{tabularx}
\caption{\scriptsize{Additional sample $m$, BNP estimates of $K_{n,m}$,  95\% exact C.I., Mittag-Leffler C.I. with their coverages (of the the exact C.I.) and Gaussian credible intervals with their coverages (of the exact C.I.). All values are rounded to the nearest integer.}}
\label{tab:3} 
\end{table}

\begin{table}[h!]
\scriptsize
    \begin{tabularx}{\textwidth}{
    >{\raggedright \hsize=0.8\hsize \arraybackslash}      
   >{\centering \hsize=0.5\hsize \arraybackslash}X 
    >{\centering \hsize=0.1\hsize \arraybackslash}X
    >{\centering \hsize=0.2\hsize \arraybackslash}X
   >{\centering \hsize=0.3 \hsize \arraybackslash}X
    >{\centering \hsize=0.3 \hsize \arraybackslash}X
     >{\centering \hsize=0.25 \hsize \arraybackslash}X
>{\centering \hsize=0.3 \hsize \arraybackslash}X
>{\centering \hsize=0.25 \hsize \arraybackslash}X}
\toprule
\multirow{2}{*}{Library}  & \multirow{2}{*}{$m$}  & \multirow{2}{*}{$\hat{K}_{n,m}$} & \multirow{2}{*}{95\% Exact C.I.} &  \multicolumn{2}{c}{Mittag-Leffler C.I.}   &  \multicolumn{2}{c}{Gaussian C.I.}\\ 
\cmidrule(lr){5-6} \cmidrule(lr){7-8}
 &  &  & & 95\% C.I. & Coverage (\%) & 95\% C.I. & Coverage (\%) \\
\midrule
 Tomato flower &  $10n$   & 8432 & (8171, 8705)   &  (8188, 8687) & 93.4 &  (8164, 8700) & 99.1 \\
$n = 2586$ & $50n$   & 25926 &  (25137, 26706)   & (25176, 26712) & 97.5 & (25154, 26698) & 98.4\\
& $100n$   & 40888 &  (39690, 42082)   & (39705, 42128) & 99.4 & (39684, 42092) & 100 \\
& $1000n$   & 113848 &  (110620, 117142)   & (110555, 117300)& 100 & (110540, 117157) & 100  \\[1ex]
  \emph{Mastigamoeba}  &  $10n$   & 2634 & (2442, 2825)    &  (2463, 2809) &90.3 &  (2448, 2819) &96.9\\
$n = 715$ & $50n$   & 9718 &  (9010, 10370)   & (9089, 10367) &94.0 & (9063, 10372)&96.1 \\
& $100n$   & 16797 &  (15664, 17928)   & (15711, 17920) & 97.6 & (15674, 17921) & 99.2 \\
& $1000n$   & 58889 &  (54962, 62878)  & (55079, 62824)  & 97.8 &(54976, 62803) &98.9 \\[1ex]
 \emph{Mastigamoeba}--norm. & $10n$   & 1280 &(1158, 1392)  & (1169, 1393)& 95.3 & (1163, 1397) &97.9 \\
$n = 363$ & $50n$   & 4344 &  (3951, 4707)  & (3966, 4729)& 98.0 & (3969, 4720) & 97.6\\
& $100n$   & 7203 &  (6581, 7809)    & (6576, 7840)& 100 & (6585, 7820) & 99.7\\
& $1000n$   & 22759 &  (20826, 24856)  & (20779, 24774) &98.0 & (20825, 24694) & 96.0 \\[1ex]
 \emph{Naegleria} aerobic  &  $10n$   & 2084 & (1920, 2253)  &  (1926, 2246) &96.1 & (1917, 2252) & 99.7 \\
$n = 959$ & $50n$   & 6781 &  (6260, 7299)   & (6265, 7306) &99.5 & (6270, 7293) & 98.5\\
& $100n$   & 11030 &  (10209, 11923)   & (10190, 11883) & 97.7 & (10207, 11852) & 95.9 \\
& $1000n$   & 33286 &  (30777, 35851)   & (30752, 35862) & 100 & (30833, 35740) &96.7 \\[1ex]
 \emph{Naegleria} anaerobic  &  $10n$   & 2994 & (2809, 3175)  &  (2826, 3161) & 91.5 &  (2817, 3171)  & 96.7  \\
$n = 969$ & $50n$   & 9679 &  (9145, 10213)  & (9135, 10218) & 100 & (9140, 10218) & 100   \\
& $100n$   & 15671 &  (14809, 16530)  & (14791, 16544) & 100 & (14807, 16535) & 100  \\
& $1000n$   & 46683 &  (44226, 49260)   & (44062, 49283) & 100 & (44137, 49229) & 99.4 \\
\bottomrule
\end{tabularx}
\caption{\scriptsize{Additional sample $m$, BNP estimates of $K_{n,m}$,  95\% exact C.I., Mittag-Leffler C.I. with their coverages (of the the exact C.I.) and Gaussian credible intervals with their coverages (of the exact C.I.). All values are rounded to the nearest integer.}}
\label{tab:4} 
\end{table}

For the \emph{Naegleria} aerobic dataset, Figure \ref{fig:2} displays BNP estimates of $K_{n,m}$ with 95\% exact credible intervals, Mittag-Leffler credible intervals and Gaussian credible intervals as a function of $m$, for $m \in [0, 5n]$. See Appendix \ref{app:E} for similar plots for the other datasets described in Table \ref{table_fav}.

\begin{figure}[h!]
\begin{center}
\medskip
\small{\textit{Naegleria} aerobic}
\medskip

\includegraphics[width = \textwidth]{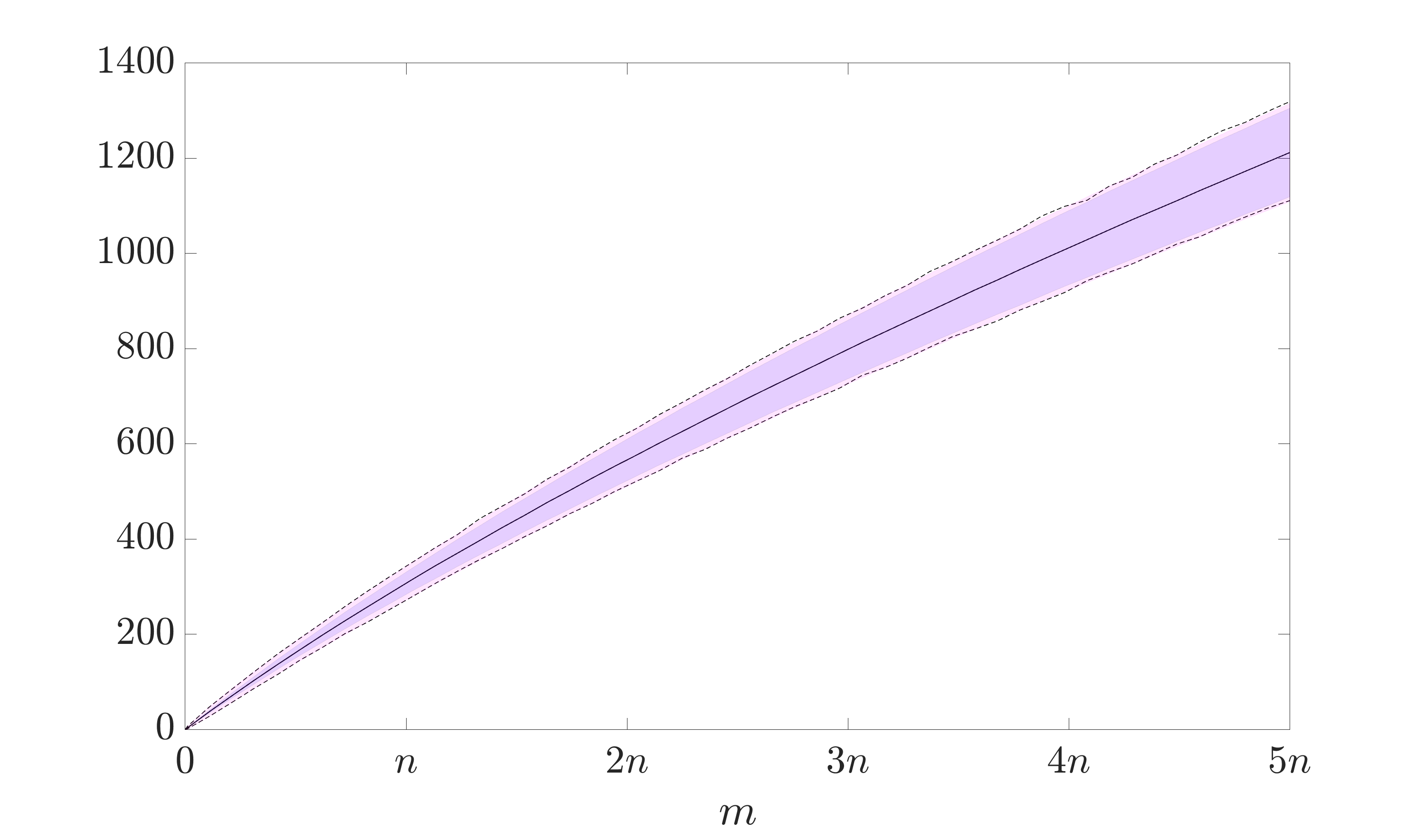}
\end{center}
\caption{\scriptsize{BNP estimates of $K_{n,m}$ (solid line --) with 95\% exact credible intervals (dashed line - -), Mittag-Leffler credible intervals (violet) and Gaussian credible intervals (pink), as a function of $m$, for $m \in [0, 5n]$.}}
\label{fig:2}
\end{figure}

Tables \ref{tab:3} and \ref{tab:4} confirm the behavior already observed on synthetic data, with Gaussian credible intervals providing a better coverage than Mittag-Leffler credible intervals for values of $m$ corresponding to $m = n, 2n, 3n, 4n, 5n, 10n$ and $50n$, and displaying near-to-optimal length. Again, the Mittag-Leffler credible intervals display a substantial improvement in performance as $m$ grows, with their coverage eventually matching or outperforming that of the Gaussian credible intervals    for values of $m = 100n, m = 1000n$. It shall be mentioned that, also when outperformed, the Gaussian intervals maintain in any case a coverage of at least 96\%.


\section{Discussion}\label{sec5}

The unseen-species problem is arguably the most popular example of ``species-sampling” problem. Given $n\geq1$ observed sample $\mathbf{X}_{n}=(X_{1},\ldots,X_{n})$ from a population of individuals belonging to different species $\mathbb{S}$, a broad class of ``species-sampling” (extrapolation) problems call for estimating features of the unknown species composition of $m\geq1$ additional unobservable samples $\mathbf{X}_{m}=(X_{n+1},\ldots,X_{n+m})$ from the same population. If $(N_{s,n})_{s\in\mathbb{S}}$ and $(N_{s,m})_{s\in\mathbb{S}}$ denote the (empirical) frequencies of species in $\mathbf{X}_{n}$ and $\mathbf{X}_{m}$, respectively, for $u,v\geq$ we set
\begin{equation}\label{ss_probl}
K_{n,m}(u,v)=\sum_{s\in\mathbb{S}}I(N_{s,n}=u)I(N_{s,m}=v),
\end{equation}
with $I(\cdot)$ being the indicator function, namely the number of species with frequency $u$ in $\mathbf{X}_{n}$ and with frequency $v$ in $\mathbf{X}_{m}$. The number of unseen species is recovered from \eqref{ss_probl} by taking
\begin{displaymath}
K_{n,m}=\sum_{v=1}^{m}K_{n,m}(0,v).
\end{displaymath}
We refer to \citet{Bal(23)} for an up-to-date overview on ``species-sampling” problems of the form \eqref{ss_probl}, with applications to biological data. While the emphasis is on the BNP approach, \citet{Bal(23)} also discusses the most recent advances in the distribution-free approach.

\subsection{Gaussian credible intervals for $K_{n,m}(0,v)$}

Among ``species-sampling” problems of the form \eqref{ss_probl}, the number of unseen rare species certainly stands out \citep{Den(19)}, with a rich literature under the distribution-free and BNP approach \citep{Fav(13),Hao(20)}. It calls for estimating $K_{n,m}(0,v)$, namely the number of hitherto unseen species that would be observed with frequency $v\geq1$ in the $m$ additional samples. Under the BNP approach with a Pitman-Yor prior, with $\alpha\in[0,1)$ and $\theta>-\alpha$, \citet[Proposition 1]{Bal(23)} provides a compound Binomial representation of the posterior distribution of $K_{n,m}(0,v)$, given $\mathbf{X}_{n}$. In particular, let $K_{n,m}^{(n)}(0,v)$ be random variable with such a distribution. Following the notation of \eqref{bino_rap}-\eqref{bino_rap0}, and denoting by $K^{\ast}_{m}(v)$ the (random) number of species with frequency $v$ in $m\geq1$ random samples from $\text{PYP}(\alpha,\theta+n)$, such that $K^{\ast}_{m}(v)\in\{0,1,\ldots,m\}$, there hold
\begin{itemize}
\item[i)] for $\alpha\in(0,1)$
\begin{equation}\label{bino_rap_freq}
K_{n,m}^{(n)}(0,v)\stackrel{\text{d}}{=}Q\left(K^{\ast}_{m}(v),\,B_{\frac{\theta}{\alpha}+k,\,\frac{n}{\alpha}-k}\right);
\end{equation}
\item[ii)] for $\alpha=0$
\begin{equation}\label{bino_rap0_freq}
K_{n,m}^{(n)}(0,v)\stackrel{\text{d}}{=}Q\left(K^{\ast}_{m}(v),\,\frac{\theta}{\theta+n}\right).
\end{equation}
\end{itemize}
The representations \eqref{bino_rap_freq}-\eqref{bino_rap0_freq} lead to posterior mean estimates of $K_{n,m}(0,v)$ that can be easily evaluated for any value of any $n$ and $m$ \citep[Equation 24 and Equation 30]{Fav(13)}. 

Exact credible intervals for $K_{n,m}(0,v)$ can be derived by Monte Carlo sampling the posterior distribution of $K_{n,m}(0,v)$, given $\mathbf{X}_{n}$. One may exploit the representations \eqref{bino_rap_freq}-\eqref{bino_rap0_freq}, as well as the predictive distributions of the Pitman-Yor prior, to sample $K_{n,m}^{(n)}(0,v)$, though with a computational burden that becomes overwhelming as $m$ increases. In particular, Algorithm \ref{alg:k} does not extend to $K_{n,m}^{(n)}(0,v)$; see \citet{Bal(23)} for details. Large $m$ asymptotic credible intervals for $K_{n,m}(0,v)$ can derived by relying on the large $m$ asymptotic behaviour of $K_{n,m}^{(n)}(0,v)$ in \citet[Theorem 4]{Fav(13)}, still involving the scaled Mittag-Leffler distribution. However, this approach would suffer from the same limitations as the approach of \citet{Fav(09)} for $K_{n,m}$. Alternatively, one may consider to extend our approach to in order to derive Gaussian credible intervals for $K_{n,m}(0,v)$, which is a promising direction for future work. This would require to extend Theorem \ref{thm_main} to $K^{\ast}_{m}(v)$, and then combine it with the representations \eqref{bino_rap_freq}-\eqref{bino_rap0_freq} along the same lines of Theorem \ref{thm_main_post}.

\subsection{More directions for future work}

Beyond the extension of our approach to other ``species-sampling”, this work opens several opportunities for future research. For example, one may consider the problem of deriving exact credible intervals for $K_{n,m}$, namely for any $n$ and $m$, by avoiding the use of Monte Carlo sampling the posterior distribution. By relying on the representations \eqref{bino_rap_freq}-\eqref{bino_rap0_freq}, this problem would require to refine the CLT \eqref{clt} with a Berry-Eseen type inequality or, better, to develop a concentration inequality for $K^{\ast}_{m}$. See \citet[Equation 11]{Con(24)} for a preliminary result in this direction, though limited to the case $\alpha=0$. Another research direction could involve the use of the Gaussian credible interval \eqref{eq: intervals} as a confidence interval under the semi-parametric approach of \citet{Fav(23)}, which assumes the tail of $P$ to be regularly varying of index $\alpha\in(0,1)$. This problem would open to the study of frequentist properties of our approach to construct Gaussian credible intervals under regular variation.


\section*{Acknowledgments}

Emanuele Dolera and Stefano Favaro received funding from the European Research Council (ERC) under the European Union’s Horizon 2020 research and innovation programme under grant agreement No 817257. Stefano Favaro also gratefully acknowledges the Italian Ministry of Education, University and Research (MIUR), “Dipartimenti di Eccellenza" grant 2013-2027. Stefano Favaro is also affiliated to IMATI-CNR “Enrico Magenes" (Milan, Italy).


\clearpage

\appendix\newpage\markboth{Appendix}{Appendix}
\renewcommand{\thesection}{\Alph{section}}
\numberwithin{equation}{section}
\numberwithin{figure}{section}


\section{Non-central generalized factorial coefficients}\label{app_comb}

We recall some definitions and basic results on signless Stirling numbers of the first type and on generalized factorial coefficients. We refer to the monograph by \citet[Chapter 2]{Cha(05)} for a comprehensive account on these combinatorial numbers and generalizations thereof. 

\subsection{Stirling numbers}\label{app_comb_stirling}

For $t>0$, the $(u,v)$-th (central or centered) signless Stirling number of the first type, denoted by $|s(u,v)|$, is the defined as the $v$-th coefficient in the expansion of $(t)_{(u)}$ into powers, i.e., 
\begin{equation}\label{stir1}
(t)_{(u)}=\sum_{v=0}^{u}|s(u,v)|t^{v}.
\end{equation}
It is assumed that $|s(0,0)|=1$, $|s(u,0)|=0$ for $u>0$ and $|s(u,v)|=0$ for $v>u$. As a generalization of \eqref{stir1}, for any $b>0$ let us consider the expansion of $(t+b)_{(u)}$ into powers, i.e., 
\begin{equation}\label{stir2}
(t+b)_{(u)}=\sum_{v=0}^{u}|s(u,v;b)|t^{v}.
\end{equation}
It is assumed that $|s(0,0;b)|=1$, $|s(u,0;b)|=(b)_{(u)}$ for $u>0$ and $|s(u,v;b)|=0$ for $v>u$. The $v$-th coefficient $|s(u,v;b)|$ in the expansion  \eqref{stir2} is defined as the non-central $(u,v)$-th signless Stirling number of the first type. The following identity shows a relationships between Stirling numbers of the first type and their corresponding non-central Stirling numbers, i.e., 
\begin{equation}\label{stir_ide}
|s(u,v;b)|=\sum_{i=v}^{u}{u\choose i}(b)_{(u-i)}|s(i,v)|.
\end{equation}
See \citet[Chapter 2]{Cha(05)} for the proof of Equation \eqref{stir_ide}, and for generalizations thereof.

\subsection{Generalized factorial coefficients}\label{app_comb_gencoeff}

For $t\in\mathbb{R}^{+}$, $a\in\mathbb{R}$ and $u\in\mathbb{N}_{0}$, let us consider the rising factorial of $at$ of order $u$, i.e., $(at)_{(u)}=\prod_{0\leq i\leq u-1}(at+i)$. The $(u,v)$-th (centered) generalized factorial coefficient, denoted by $\mathscr{C}(u,v;a)$, is defined as the  $v$-th coefficient in the expansion of $(at)_{(u)}$ into rising factorials, i.e.,
\begin{equation}\label{eq:genfact3}
(at)_{(u)}=\sum_{v=0}^{u}\mathscr{C}(u,v;a)(t)_{(v)}.
\end{equation}
It is assumed that $\mathscr{C}(0,0;a)=1$, $\mathscr{C}(u,0;a)=0$ for $u>0$, $\mathscr{C}(u,v;a)=0$ for $v>u$. As a generalization of \eqref{eq:genfact3}, for any $b>0$ let us consider the expansion of $(at-b)_{(u)}$ into rising factorials, i.e., 
\begin{equation}\label{eq:genfact5}
(at-b)_{(u)}=\sum_{v=0}^{u}\mathscr{C}(u,v;a,b)(t)_{(v)}.
\end{equation}
It is assumed that $\mathscr{C}(0,0;a,b)=1$, $\mathscr{C}(u,0;a,b)=(-b)_{(u)}$ for $u>0$, $\mathscr{C}(u,v;a,b)=0$ for $v>u$. The $v$-th coefficient $\mathscr{C}(u,v;a,b)$ in the expansion \eqref{eq:genfact5} is defined as the non-central $(u,v)$-th generalized factorial coefficient. In particular, an explicit expression for the $v$-th coefficient is
\begin{equation}\label{eq:genfact1}
\mathscr{C}(u,v;a,b)=\frac{1}{v!}\sum_{j=0}^{v}(-1)^{j}{v\choose j}(-ja-b)_{(u)}.
\end{equation}
The following identity shows an asymptotic relationships, as $a\rightarrow0$, between the non-central generalized factorial coefficients and the non-central signless Stirling numbers of the first type:
\begin{equation}\label{eq:genfact2}
\lim_{a\rightarrow0}\frac{\mathscr{C}(u,v;a,b)}{a^{k}}=|s(u,v;-b)|;
\end{equation}
See \citet[Chapter 2]{Cha(05)} for the proof of Equation \eqref{eq:genfact1} and the proof of Equation \eqref{eq:genfact2}.


\section{Sampling properties of the Pitman-Yor prior}\label{app_pyp}

\subsection{Predictive distributions}\label{predictive}

Due to the discreteness of the Pitman-Yor prior, a random sample $\mathbf{X}_{n}$ from $P\sim\text{PYP}(\alpha,\theta)$ induces a random partition of $[n]=\{1,\ldots,n\}$ into $K_{n}=j\leq n$ blocks, labelled by $\{S_{1}^{\ast},\ldots,S^{\ast}_{K_{n}}\}$, with frequencies $(N_{1,n},\ldots,N_{K_{n},n})=(n_{1},\ldots,n_{k})$ such that the $n_{i}$'s are positive and $\sum_{1\leq i\leq j}n_{i}=n$. The distribution of the random partition is determined by the predictive distribution, or generative scheme, of the Pitman-Yor prior \citep[Proposition 9]{Pit(95)}, i.e., 
\begin{displaymath}
\text{Pr}[X_{1}\in\cdot]=\nu(\cdot),
\end{displaymath}
and, for $n\geq1$,
\begin{equation}\label{eq:pred}
\text{Pr}[X_{n+1}\in\cdot\,|\, \mathbf{X}_{n}]=\frac{\theta+j\alpha}{\theta+n}\nu(\cdot)+\frac{1}{\theta+n}\sum_{i=1}^{j}(n_{i}-\alpha)\delta_{S_{i}^{\ast}}(\cdot).
\end{equation}
The expression in \eqref{eq:pred} provides the conditional distribution of the random partition of $[n+1]$ obtained after sampling one additional point $X_{n+1}$, given the previous partition of $[n]$. This is a linear combination of: i) the probability $(\theta+j\alpha)/(\theta+n)$ that $X_{n+1}$ belongs to a new species, i.e., creating a new block in the partition of the set $[n]$; ii) the probability  $(n_{i}-\alpha)/(\theta+n)$ that $X_{n+1}$ is of species $S^{\ast}_{i}$, i.e., increasing by $1$ the size of the block $S^{\ast}_{i}$ in the partition of the set $[n]$, for $i=1,\ldots,j$. A larger value of $\alpha\in(0,1)$ corresponds to a higher probability of observing new species in the sample. If $\alpha=0$, i.e. the Dirichlet prior, the probabilities in \eqref{eq:pred} become proportional to the empirical frequencies of species, and the probability of generating a new one no longer depends on the number of observed species.

\subsection{Sampling formula}\label{sampling}

For any $r \in [n]$ let $M_{r,n}$ be the number of distinct species with frequency $r$ in a random sample $\mathbf{X}_{n}$ from $P\sim\text{PYP}(\alpha,\theta)$, such that $M_{r,n}=\sum_{1\leq i\leq K_{n}}I(N_{i,n}=r)$ such that $\sum_{1\leq r\leq n}M_{r,n}=K_{n}$ and $\sum_{1\leq r\leq n}rM_{r,n}=n$. In particular, let us consider the (partition) set
\begin{displaymath}
\mathcal{M}_{n,j}=\left\{j\in\{1,\ldots,n\}\text{ and }(m_{1},\ldots,m_{n})\text{ : }m_{i}\geq0,\,\sum_{i=1}^{n}m_{i}=j,\,\sum_{i=1}^{n}im_{i}=n\right\}. 
\end{displaymath}
The distribution of the random variable $\mathbf{M}_{n}=(M_{1,n}, \ldots,M_{n,n})$ is referred to as the Ewens-Pitman sampling formula \citep[Proposition 9]{Pit(95)}, and it is such that for $(m_{1},\ldots,m_{n})\in\mathcal{M}_{n,j}$
\begin{equation}\label{eq_ewe_py}
\text{Pr}[\mathbf{M}_{n}=(m_{1},\ldots,m_{n})]=n!\frac{\left(\frac{\theta}{\alpha}\right)_{(\sum_{i=1}^{n}m_{i})}}{(\theta)_{(n)}}\prod_{i=1}^{n}\left(\frac{\alpha(1-\alpha)_{(i-1)}}{i!}\right)^{m_{i}}\frac{1}{m_{i}!}.
\end{equation}
The Ewens-Pitman sampling formula generalizes the Ewens sampling formula, which corresponds to $\alpha=0$. The distribution $K_{n}$ follows by marginalizing \eqref{eq_ewe_py}. For  $j\in\{1,\ldots,n\}$, if $\alpha\in(0,1)$ then
\begin{equation}\label{eq_dist_py}
\text{Pr}[K_{n}=j]=\frac{\left(\frac{\theta}{\alpha}\right)_{(j)}}{(\theta)_{(n)}}\mathscr{C}(n,j;\alpha),
\end{equation}
whereas if $\alpha=0$ then
\begin{equation}\label{eq_dist_dp}
\text{Pr}[K_{n}=j]=\frac{\theta^{j}}{(\theta)_{(n)}}|s(n,j)|.
\end{equation}
We refer to the monograph by \citet[Chapter 2 and Chapter 3]{Pit(06)} for a detailed account on distributional properties of the random partition induced by sampling from $P\sim\text{PYP}(\alpha,\theta)$.

\subsection{Large $n$ asymptotics}\label{asymptotics}

At the sampling level, the power-law tail behaviour of the PYP prior emerges from the analysis of the large $n$ asymptotic behaviour of $K_{n}$ and $M_{r,n}$. For $\alpha\in(0,1)$ and $\theta>-\alpha$ let $S_{\alpha,\theta/\alpha}$ be a Mittag-Leffler random variable \citet[Theorem 3.8]{Pit(06)} shows that, as $n\rightarrow+\infty$,
\begin{equation} \label{eq:sigma_diversity}
\frac{K_{n}}{n^{\alpha}}\stackrel{\text{a.s.}}{\longrightarrow} S_{\alpha,\theta/\alpha}
\end{equation}
and
\begin{equation} \label{eq:sigma_diversity_m}
\frac{M_{r,n}}{n^{\alpha}}\stackrel{\text{a.s.}}{\longrightarrow}  \frac{\alpha(1-\alpha)_{(r-1)}}{r!}S_{\alpha,\theta/\alpha}.
\end{equation}
See also \citet{DF(20a)} and \citet{DF(20b)} for refinements of \eqref{eq:sigma_diversity}. For $\alpha=0$, as $n\rightarrow+\infty$, \citet[Theorem 2.3]{Kor(73)} shows that, as $n\rightarrow+\infty$
\begin{equation}
\frac{K_{n}}{\log n}\stackrel{\text{a.s.}}{\longrightarrow}\theta
\end{equation}
and 
\begin{equation}
M_{r,n}\stackrel{\text{a.s.}}{\longrightarrow}P_{\theta/r},
\end{equation}
where $P_{\theta/r}$ is a Poisson random variable with parameter $\theta/r$. From \eqref{eq:sigma_diversity}, for large $n$, $K_{n}$ grows as $n^{\alpha}$. This is precisely the growth of the number of distinct species in $n\geq1$ random samples from a power-law distribution of exponent $c=\alpha^{-1}$. Also, from \eqref{eq:sigma_diversity} and \eqref{eq:sigma_diversity_m}, $p_{\alpha,r}=\alpha(1-\alpha)_{(r-1)}/r!$ is the large $n$ asymptotic proportion of the number of distinct species with frequency $r$. Then, $p_{\alpha,r}\propto r^{-\alpha -1}$ for large $r$, which is the distribution of the number of distinct species with frequency $r$ arising from a power-law distribution of exponent $c=\alpha^{-1}$.


\section{Proof of \eqref{mom_post_2} and of the LLN \eqref{lln_post} \label{proof1}}
The proof of \eqref{mom_post_2} follows the same lines when $\alpha \in (0, 1)$ or $\alpha = 0$, relying on representations \eqref{bino_rap} and \eqref{bino_rap0} respectively; we split the two cases for the sake of clarity. The LLN \eqref{lln_post}  follows directly from \eqref{mom_post_2} and its proof is carried out at the end of the section.
\subsection{Proof of \eqref{mom_post_2} in the case $\alpha \in (0, 1)$} \label{sec:proof_asy_1}
We resort to representation \eqref{bino_rap} which, in the regime $\theta = \tau m ,\,  n = \nu m, \, j = \varrho m$, becomes
$$  K_m^{(n)} \stackrel{\text{d}}{ =} Q\left(K_m^*, \, B_{\left(\frac{\tau}{\alpha} + \varrho\right)\, m, \,  \left(\frac{\nu}{\alpha} - \varrho\right)\, m }\right)
$$
By the law of total expectation and recalling the expression for the mean of the beta-binomial distribution, the asymptotics for the mean of $K_m^*$ in \eqref{mom2} and the definitions of $\mathfrak{m}_{\alpha, \lambda}$  and $\mathscr{M}_{\alpha, \tau, \nu, \varrho}$ in the case $\alpha \in (0, 1)$, 
\begin{align*}
 \mathbb{E} \left[K_m^{(n)} \right] &= \mathbb{E} \left[ \mathbb{E} \left[Q\left(K_m^*, \, B_{\left(\frac{\tau}{\alpha} + \varrho\right)\, m, \,  \left(\frac{\nu}{\alpha} - \varrho\right)\, m }\right) \, \bigg| \, K_m^* \right] \right]\\
 &  = \mathbb{E} \left[ K_m^* \cdot \frac{\tau + \varrho \alpha}{\tau + \nu}\right]\\
 & = m \cdot \mathfrak{m}_{\alpha, \lambda} \,  \frac{\tau + \varrho \alpha}{\tau + \nu} + O(1)\\
 & = m \cdot \mathscr{M}_{\alpha, \tau, \nu, \varrho} + O(1)
\end{align*}
By the law of total variance, and recalling the expression for the moments of the beta-binomial distribution, the asymptotics in \eqref{mom2} and the definitions of $\mathfrak{m}_{\alpha, \lambda}, \mathfrak{s}_{\alpha, \lambda}$ and $\mathscr{S}_{\alpha, \tau, \nu, \varrho}$ in the case $\alpha \in (0, 1)$, 
\begin{align*}
 & \operatorname{Var}\left(K_m^{(n)} \right) \\
 &\quad = \mathbb{E} \left[ \operatorname{Var} \left(Q\left(K_m^*, \, B_{\left(\frac{\tau}{\alpha} + \varrho\right)\, m, \,  \left(\frac{\nu}{\alpha} - \varrho\right)\, m }\right) \, \bigg| \, K_m^* \right) \right] \\
 &\quad  \quad + \operatorname{Var} \left( \mathbb{E}\left[Q\left(K_m^*, \, B_{\left(\frac{\tau}{\alpha} + \varrho\right)\, m, \,  \left(\frac{\nu}{\alpha} - \varrho\right)\, m }\right) \, \bigg| \, K_m^* \right] \right)\\
 & \quad  = \mathbb{E} \left[ K_m^* \, \frac{(\tau + \varrho \alpha)(\nu - \varrho \alpha)}{(\tau+ \nu)^2} \, \left(1 - \frac{\alpha}{(\tau + \nu)m+\alpha}\right)  
+ \frac{\left(K_m^*\right)^2}{m} \, \frac{(\tau + \varrho \alpha)(\nu - \varrho \alpha)}{(\tau+ \nu)^2 } \, \frac{\alpha m}{(\tau+\nu) m + \alpha} \right] \\
 &\quad  \quad +  \operatorname{Var} \left( K_m^* \cdot \frac{\tau + \varrho \alpha}{\tau + \nu}\right)\\
 &\quad  =  \left[  m \cdot \mathfrak{m}_{\alpha, \lambda} + O(1)  \right] \cdot \frac{(\tau + \varrho \alpha)(\nu - \varrho \alpha)}{(\tau+ \nu)^2} \left[ 1 + O\left(\frac{1}{m}\right)\right] + m \cdot \mathfrak{s}_{\alpha, \lambda} \,  \frac{(\tau + \varrho \alpha)^2}{(\tau + \nu)^2} + O(1)\\
 & \quad \quad + \frac{ m^2\cdot \mathfrak{m}_{\alpha, \lambda}^2 +O(m)}{m} \cdot  \frac{(\tau + \varrho \alpha)(\nu - \varrho \alpha)}{(\tau+ \nu)^2 } \left[ \frac{\alpha}{(\tau + \nu)  } +  O \left(\frac{1}{m}\right)\right]\\
 & \quad = m \left[ \mathfrak{m}_{\alpha, \lambda} \frac{(\tau+\varrho \alpha)(\nu - \varrho \alpha)}{(\tau + \nu)} + \mathfrak{s}_{\alpha, \lambda}\frac{(\tau+\varrho \alpha)^2}{(\tau + \nu)} + \mathfrak{m}^2_{\alpha, \lambda}\frac{ \alpha (\tau+\varrho \alpha)(\nu - \varrho \alpha)}{(\tau + \nu)^3} \right] +O(1)\\
  & \quad = m \cdot \mathscr{S}_{\alpha, \tau, \nu, \varrho} + O(1)
\end{align*}

\subsection{Proof of \eqref{mom_post_2} in the case $\alpha = 0$} \label{sec:proof_asy_2}
We resort to representation \eqref{bino_rap0} which, in the regime $\theta = \tau m ,\,  n = \nu m$, becomes
$$  K_m^{(n)} \stackrel{\text{d}}{ =} Q\left(K_m^*, \, \frac{\tau}{\lambda}\right)$$
By the law of total expectation, and recalling the expression for the mean of the binomial distribution, the asymptotics for the mean of $K_m^*$ in \eqref{mom2} and the definitions of $\mathfrak{m}_{0, \lambda}$  and $\mathscr{M}_{0, \tau, \nu, \varrho}$, 
\begin{align*}
 \mathbb{E} \left[K_m^{(n)} \right] &= \mathbb{E} \left[ \mathbb{E} \left[Q\left(K_m^*, \, \frac{\tau}{\lambda}\right) \, \bigg| \, K_m^* \right] \right]\\
 &= \mathbb{E} \left[ K_m^* \cdot \frac{\tau }{\lambda}\right]\\
 & = m \cdot \mathfrak{m}_{0, \lambda} \,  \frac{\tau }{\lambda} + O(1)\\
 & = m \cdot \mathscr{M}_{0, \tau, \nu, \varrho} + O(1)
\end{align*}
By the law of total variance, and recalling the expression for the moments of the binomial distribution, the asymptotics in \eqref{mom2} and the definitions of $\mathfrak{m}_{0, \lambda}, \mathfrak{s}_{0, \lambda}$ and $\mathscr{S}_{0, \tau, \nu, \varrho}$, 
\begin{align*}
 \operatorname{Var}\left(K_m^{(n)} \right) &= \mathbb{E} \left[ \operatorname{Var} \left(Q\left(K_m^*, \, \frac{\tau}{\lambda}\right) \, \bigg| \, K_m^* \right) \right] + \operatorname{Var} \left( \mathbb{E}\left[Q\left(K_m^*, \, \frac{\tau}{\lambda} \right) \, \bigg| \, K_m^* \right] \right)\\
 &  = \mathbb{E} \left[ K_m^* \cdot \frac{\tau \nu}{\lambda^2} \right]  +  \operatorname{Var} \left( K_m^* \cdot \frac{\tau}{\lambda }\right)\\
 & =   m \cdot \mathfrak{m}_{0, \lambda}  \cdot \frac{\tau \nu}{\lambda^2} + O(1)  + m \cdot \mathfrak{s}_{0, \lambda} \,  \frac{\tau^2}{\lambda^2} + O(1)\\
  & = m \cdot \mathscr{S}_{0, \tau, \nu, \varrho} + O(1)
\end{align*}

\subsection{Proof of the LLN \eqref{lln_post}} \label{sec:lln_proof}
To prove the LLN \eqref{lln_post}, we fix $\varepsilon >0$ and combine Chebychev inequality with the (variance) asymptotic expansion \eqref{mom_post_2}. In particular, we write
\begin{align*}
P\left[\left| \frac{K_m^{(n)} - \mathbb{E}\left[K_m^{(n)}\right]}{m} \right|> \varepsilon \right]& = P \left[\left| K_m^{(n)} - \mathbb{E}\left[K_m^{(n)}\right] \right| > m\varepsilon \right]\\
& \le \frac{\operatorname{Var}\left(K_m^{(n)}\right)}{m^2 \varepsilon^2}\\
&= O\left(\frac{1}{m}\right) \rightarrow 0
\end{align*}
as $m \rightarrow +\infty$. Since the asymptotic expansion \eqref{mom_post_1} of $\E\left[K_m^{(n)}\right]$ implies that $m^{-1}\mathbb{E}\left[K_m^{(n)}\right]\rightarrow \mathscr{M}_{\alpha, \tau, \nu, \varrho}$ as $m \rightarrow+\infty$, the proof is concluded by means of Slutsky's theorem.

\section{Proof of the CLT \eqref{clt_post}\label{proof2}}\label{sec:d}
We begin by providing some detail on the sketch of proof presented in section \ref{sec3}. Equality \eqref{F=I1+I2} follows from the chain of equalities
\begin{align*}
F_m(x) & =\int_0^{+\infty} P\left[Q_m(z) \le m\mathscr{M}_{\alpha, \tau, \nu, \varrho} + \sqrt{m} \mathscr{S}_{\alpha, \tau, \nu, \varrho} \, x \right] \ \mu_{\frac{K_m^*}{m}}(\mathrm{d} z) \\
& = \int_0^{+\infty} P\left[V_m(z) \le \frac{\sqrt{m}\, \left[\mathscr{M}_{\alpha, \tau, \nu, \varrho} - \mu(z)\right] + \mathscr{S}_{\alpha, \tau, \nu, \varrho} \, x}{\sigma(z)}\right] \ \mu_{\frac{K_m^*}{m}}(\mathrm{d} z) \\
& = \int_0^{+\infty}\Phi\left( \frac{\sqrt{m}\, \left[\mathscr{M}_{\alpha, \tau, \nu, \varrho} - \mu(z)\right] + \mathscr{S}_{\alpha, \tau, \nu, \varrho} \, x}{\sigma(z)} \right)\ \mu_{\frac{K_m^*}{m}}(\mathrm{d} z) \\
& \quad+ \int_0^{+\infty} \left\{F_{V_m(z)}\left( \frac{\sqrt{m}\, \left[\mathscr{M}_{\alpha, \tau, \nu, \varrho} - \mu(z)\right] + \mathscr{S}_{\alpha, \tau, \nu, \varrho} \, x}{\sigma(z)} \right) \right.\\
&\quad\quad\quad\quad\quad \left. - \Phi\left( \frac{\sqrt{m}\, \left[\mathscr{M}_{\alpha, \tau, \nu, \varrho} - \mu(z)\right] + \mathscr{S}_{\alpha, \tau, \nu, \varrho} \, x}{\sigma(z)} \right)\right\}\ \mu_{\frac{K_m^*}{m}}(\mathrm{d} z) \\
& = : \mathcal{I}^{(m)}_1(x) + \mathcal{I}^{(m)}_2(x).
\end{align*}

The following two lemmas make use Proposition \ref{p2} and the CLT for $K_m^*$,  i.e. \eqref{clt} of theorem \ref{thm_main}, and are instrumental to the proof of \eqref{rv_1} trough Proposition \ref{lem:lem3_ter}.

\begin{lem}
\label{lem:lem1_ter}
If $Y$ is a Gaussian random variable with mean $0$ and variance $\mathfrak{s}_{\alpha, \lambda}$, then
\begin{displaymath}
\mathbb{E} \left[\Phi\left(\frac{\mu'(\mathfrak{m}_{\alpha, \lambda}) Y + \mathscr{S}_{\alpha, \tau, \nu, \varrho} \, x}{\sigma\left(\mathfrak{m}_{\alpha, \lambda} \right)} \right)\right] = \Phi(x)
\end{displaymath}
for every $x\in \mathbb{R}$.
\end{lem}
\begin{proof}
We introduce a standard Gaussian random variable $Z$, with $Z$ independent from the random variable $Y$. By a standard property of conditional probability, 
\begin{displaymath}
\Phi\left(\frac{\mu'(\mathfrak{m}_{\alpha, \lambda}) Y + \mathscr{S}_{\alpha, \tau, \nu, \varrho} \, x}{\sigma\left(\mathfrak{m}_{\alpha, \lambda} \right)} \right) = P \left[ Z \le \frac{\mu'(\mathfrak{m}_{\alpha, \lambda}) Y + \mathscr{S}_{\alpha, \tau, \nu, \varrho} \, x}{\sigma\left(\mathfrak{m}_{\alpha, \lambda} \right)} \, \bigg| \,Y \right],
\end{displaymath}
which implies
\begin{displaymath}
\mathbb{E}\left[ \Phi\left(\frac{\mu'(\mathfrak{m}_{\alpha, \lambda}) Y + \mathscr{S}_{\alpha, \tau, \nu, \varrho} \, x}{\sigma\left(\mathfrak{m}_{\alpha, \lambda} \right)}\right)\right]= P \left[ \frac{\sigma(\mathfrak{m}_{\alpha, \lambda}) Z - \mu'(\mathfrak{s}_{\alpha, \lambda}) Y}{\mathscr{S}_{\alpha, \tau, \nu, \varrho} }\le x\right].
\end{displaymath}
To conclude, notice that $\frac{\sigma(\mathfrak{m}_{\alpha, \lambda}) Z - \mu'(\mathfrak{s}_{\alpha, \lambda}) Y}{\mathscr{S}_{\alpha, \tau, \nu, \varrho} }$ is a linear combination of independent Gaussian random variables, hence it is Gaussian with mean $0$ and variance
\begin{displaymath}
 \frac{ \sigma^2\left(\mathfrak{m}_{\alpha, \lambda}\right) + \mathfrak{s}_{\alpha, \lambda}^{2}\cdot \left(\mu'\left(\mathfrak{m}_{\alpha, \lambda}\right)\right)^2 }{\mathscr{S}_{\alpha, \tau, \nu, \varrho}^2} =1,
\end{displaymath}
where the last identity follows from Proposition \ref{p2}. This completes the proof. 
\end{proof}

\begin{lem}
\label{lem:lemma2_ter}
Let $\psi:[0, +\infty) \to \mathbb{R}$ be a continuous function such that $\psi \in C^1((0, +\infty))$. If the function $\psi$ has bounded (first) derivative, then as $n \to + \infty$
\begin{displaymath}
\sqrt{n} \left[\psi \left(\frac{K_m^*}{m}\right) - \psi(\mathfrak{m}_{\alpha, \lambda}) \right] \stackrel{\text{w}}{\longrightarrow}\mathcal{N}\left(0, (\psi'(\mathfrak{m}_{\alpha, \lambda}))^2 \mathfrak{s}_{\alpha, \lambda}^2\right)
\end{displaymath}
\end{lem}
\begin{proof}
By means of the fundamental theorem of calculus, we can write that
\begin{equation}\label{part21}
\sqrt{n} \left[\psi \left(\frac{K_m^*}{m}\right) - \psi(\mathfrak{m}_{\alpha, \lambda}) \right] = \sqrt{n} \left(\frac{K_m^*}{m} - \mathfrak{m}_{\alpha, \lambda}\right)\int_0^1 \psi'\left(\mathfrak{m}_{\alpha, \lambda} + t \left[ \frac{K_m^*}{m} - \mathfrak{m}_{\alpha, \lambda} \right]\right) \mathrm{d}t.
\end{equation}
By \eqref{lln}, we have that $m^{-1}K_m^*-\mathfrak{m}_{\alpha, \lambda} \stackrel{\text{p}}{\longrightarrow}0$, as $m\rightarrow+\infty$. Since $\psi'$ is bounded, as $m\rightarrow+\infty$
\begin{equation}\label{part22}
\int_0^1 \psi'\left(\mathfrak{m}_{\alpha, \lambda} + t \left[ \frac{K_m^*}{m} - \mathfrak{m}_{\alpha, \lambda} \right]\right) \mathrm{d}t \stackrel{\text{p}}{\longrightarrow} \psi'(\mathfrak{m}_{\alpha, \lambda})
\end{equation}
and, by \eqref{clt}, as $m\rightarrow+\infty$
\begin{equation}\label{part23}
\sqrt{n} \left(\frac{K_m^*}{m} - \mathfrak{m}_{\alpha, \lambda}\right) \stackrel{\text{w}}{\longrightarrow}\mathcal{N}\left(0,  \mathfrak{s}_{\alpha, \lambda}^2\right).
\end{equation}
From \eqref{part21}, with \eqref{part22} and \eqref{part23}, the proof completed by means of Slutsky's theorem.
\end{proof}

\subsection{Proof of Proposition \ref{BE_Q} in the case $\alpha = 0$}\label{sec:d1}
\begin{proof}[\underline{Proof of \eqref{mean_Q} and \eqref{var_Q}}]
By a simple computation,
\begin{align*}
\mathbb{E}[Q_m(z)] &= \left\lfloor mz \right\rfloor  \cdot \frac{\tau}{\tau + \nu}\\
& = m z \cdot \frac{\tau}{\tau + \nu} - \left(m z - \left\lfloor mz \right\rfloor\right) \cdot\frac{\tau}{\tau + \nu}
\end{align*}
Since $m z - \left\lfloor mz \right\rfloor<1$ by definition, this proves the first equality. For the second equaility, 
\begin{align*}
 \operatorname {Var} \left(Q_m(z)\right)&= \left\lfloor mz \right\rfloor  \cdot \frac{\tau \nu}{(\tau + \nu)^2}\\
& = m z \cdot \frac{\tau \nu}{(\tau + \nu)^2}- \left(m z - \left\lfloor mz \right\rfloor\right) \cdot\frac{\tau \nu}{(\tau + \nu)^2}\\
& = m \sigma^2(z) + O(1)
\end{align*}
so that the $O(1)$ again accounts for the fact that we are discarding the floor function.
\end{proof}

\begin{proof}[\underline{Proof of \eqref{eq:be_Qm}}]
This is the standard Berry-Essen bound for the binomial distribution - see \citep[Chapter V, Theorem 4]{Pet(75)}.
\end{proof}

\subsection{Proof of Proposition \ref{BE_Q} in the case $\alpha \in (0, 1)$}	\label{sec:d2}
\begin{proof}[\underline{Proof of \eqref{mean_Q} and \eqref{var_Q}}]
By the law of total expectation, 
\begin{align*}
\mathbb{E}[Q_m(z)]&  = \mathbb{E} \left[\mathbb{E}\left[Q_m(z) \,  \bigg| \, B_{\left(\frac{\tau}{\alpha} + \varrho\right)\, m, \,  \left(\frac{\nu}{\alpha} - \varrho\right)\, m } \right]\right] \\
& = \left\lfloor mz \right\rfloor \mathbb{E}\left[ B_{\left(\frac{\tau}{\alpha} + \varrho\right)\, m, \,  \left(\frac{\nu}{\alpha} - \varrho\right)\, m } \right] \\
& = m z \cdot \frac{\tau + \varrho \alpha}{\tau + \nu} - \left(m z - \left\lfloor mz \right\rfloor\right) \cdot \frac{\tau + \varrho \alpha}{\tau + \nu}
\end{align*}
Since $m z - \left\lfloor mz \right\rfloor<1$ by definition, this proves the first equality. By the law of total variance,
\begin{align*}
 \operatorname {Var} \left(Q_m(z)\right)&=\mathbb{E} \left[\operatorname {Var} \left(Q_m(z) \, \bigg| \,  B_{\left(\frac{\tau}{\alpha} + \varrho\right)\, m, \,  \left(\frac{\nu}{\alpha} - \varrho\right)\, m }\right)\right]+\operatorname {Var} \left(\mathbb{E} \left[Q_m(z) \, \bigg| \,  B_{\left(\frac{\tau}{\alpha} + \varrho\right)\, m, \,  \left(\frac{\nu}{\alpha} - \varrho\right)\, m }\right]\right)   \\
 & = \left\lfloor mz \right\rfloor \mathbb{E}\left[ B_{\left(\frac{\tau}{\alpha} + \varrho\right)\, m, \,  \left(\frac{\nu}{\alpha} - \varrho\right)\, m } -  B^2_{\left(\frac{\tau}{\alpha} + \varrho\right)\, m, \,  \left(\frac{\nu}{\alpha} - \varrho\right)\, m } \right] + \left\lfloor mz \right\rfloor \operatorname{Var}\left( B_{\left(\frac{\tau}{\alpha} + \varrho\right)\, m, \,  \left(\frac{\nu}{\alpha} - \varrho\right)\, m }\right)\\
 & = m z \,\frac{\tau + \varrho \alpha}{\tau +\nu}   \left[  1 - \frac{\tau + \varrho \alpha}{\tau + \nu} + \alpha z\, \frac{ \nu - \varrho \alpha}{\left(\tau + \nu\right)^2}\right] + O(1)
\end{align*}
where the $O(1)$ accounts for the fact that we are discarding the floor function and that we are substituting the exact expression for the variance of $B_{\left(\frac{\tau}{\alpha} + \varrho\right)\, m, \,  \left(\frac{\nu}{\alpha} - \varrho\right)\, m }$  with its asymptotic principal part:
\begin{align*}
    m\, \operatorname{Var}\left(B_{\left(\frac{\tau}{\alpha} + \varrho\right)\, m, \,  \left(\frac{\nu}{\alpha} - \varrho\right)\, m }\right)  & = \frac{\alpha \, (\tau + \varrho \alpha)(\nu - \varrho \alpha)}{ (\tau + \nu)^2 (\tau + \nu + \alpha/m) } \\
    & = \frac{\alpha \, (\tau + \varrho \alpha)(\nu - \varrho \alpha)}{ (\tau + \nu)^3 \left[1 + \alpha/(\tau m + \nu m) \right] } \\
    & =  \frac{\alpha \, (\tau + \varrho \alpha)(\nu - \varrho \alpha)}{ (\tau + \nu)^3  } \, \left[1  + O\left(\frac{1}{m}\right)\right]
\end{align*}   
\end{proof}

\begin{proof}[\underline{Proof of \eqref{eq:be_Qm}}] The outline of the proof is as follows: fix $\delta \in (0, 1/4)$ and a constant $\mathcal{C}$ and start with the well-known inequality \citep[Chapter V, Theorem 2]{Pet(75)}
\begin{displaymath}
    \left\|F_{V_m(z)} - \Phi \right\|_\infty \le \int_{| \xi| \le \mathcal{C}\,\sigma(z) \, m^\delta} \left|\frac{\varphi_{V_m(z)} (\xi) - e^{- \frac{\xi^2}{2}}}{\xi} \right| \, \mathrm{d} \xi + \tilde{\mathcal{C}}{m^{-\delta}} 
    \end{displaymath}
    where $\varphi_{V_m(z)}$ denotes the characteristic function of $V_m(z)$ and $\tilde{\mathcal{C}} =\max_{z \in [\zeta_0, \zeta_1]} \sigma(z)\,  \mathcal{C}$. For notational convenience, let $B_m = B_{\left(\frac{\tau}{\alpha} + \varrho\right)\, m, \,  \left(\frac{\nu}{\alpha} - \varrho\right)\, m }$ and define, for $z \in [\zeta_0, \zeta_1]$,
 \begin{align*}
G_m(z) &= \sqrt{m} \left[\mu(z)  - z B_m \right] \\
S_m (z)&= \frac{z}{\sigma^2(z)} \left[B_m - B_m^2 \right] 
\end{align*}
and 
\begin{displaymath}
S(z) = \frac{z}{\sigma^2(z)} \cdot \frac{(\tau + \varrho \alpha)(\nu - \varrho \alpha)}{(\tau + \nu)^2}.
\end{displaymath}
Further, let $G(z)$ denote a Gaussian random variable with mean $0$ and variance 
$$s^2(z) = z^2 \cdot \frac{\alpha (\tau + \varrho \alpha)(\nu - \varrho \alpha)}{(\tau + \nu)^3}, $$ 
independent of $B_m$. In steps 1--3 we rewrite the right-end side of the inequality as 
\begin{align*}
   & \left\|F_{V_m(z)} - \Phi \right\|_\infty \\
   &  \quad \le \int_{| \xi| \le \mathcal{C}\,\sigma(z) \, m^\delta} \left|\frac{\mathbb{E}\left[ e^{ - i  \frac{\xi } {\sigma(z)} G_m(z)- \frac{\xi^2}{2} S_m(z)} \right] - \mathbb{E}\left[ e^{ - i  \frac{\xi } {\sigma(z)} G(z)- \frac{\xi^2}{2} S(z)} \right] }{\xi} \right| \, \mathrm{d} \xi \\
   & \quad \quad +  \tilde{C_2}{m^{2\delta - \frac{1}{2}}} +
 \tilde{\mathcal{C}}{m^{-\delta}} 
    \end{align*}
 Then, making use of the triangular inequality, we  further split the problem in two parts:
 \begin{align*}
   & \left\|F_{V_m(z)} - \Phi \right\|_\infty \\
   &  \quad \le \int_{| \xi| \le \mathcal{C}\,\sigma(z) \, m^\delta} \left|\frac{\mathbb{E}\left[ e^{ - i  \frac{\xi } {\sigma(z)} G_m(z)- \frac{\xi^2}{2} S_m(z)} \right] - \mathbb{E}\left[ e^{ - i  \frac{\xi } {\sigma(z)} G_m(z)- \frac{\xi^2}{2} S(z)} \right] }{\xi} \right| \, \mathrm{d} \xi \\
   & \quad 	\quad + \int_{| \xi| \le \mathcal{C}\,\sigma(z) \, m^\delta} \left|\frac{\mathbb{E}\left[ e^{ - i  \frac{\xi } {\sigma(z)} G_m(z)- \frac{\xi^2}{2} S(z)} \right] - \mathbb{E}\left[ e^{ - i  \frac{\xi } {\sigma(z)} G(z)- \frac{\xi^2}{2} S(z)} \right] }{\xi} \right| \, \mathrm{d} \xi\\
   & \quad \quad +  \tilde{C_2}{m^{2\delta - \frac{1}{2}}} +
 \tilde{\mathcal{C}}{m^{-\delta}} \\
 & \quad = \mathcal{I}_m^{(1)} (z) + \mathcal{I}_m^{(2)}(z) +  \tilde{C_2}{m^{2\delta - \frac{1}{2}}} +
 \tilde{\mathcal{C}}{m^{-\delta}} 
    \end{align*}
In steps 4 and 5 we bound $\mathcal{I}_m^{(1)}(z) $ and $\mathcal{I}_m^{(2)} (z)$ for every $z \in [\zeta_0, \zeta_1]$ respectively by
\begin{displaymath}
\mathcal{I}_m^{(1)}(z) \le c_1 \, m^{2\delta - \frac{1}{2}} + c m^{-3/2}
\end{displaymath}
for some suitable constants $c_1, c >0$ and 
\begin{displaymath}
\mathcal{I}_m^{(2)}(z) \le c_2 \, m^{-\gamma}
\end{displaymath}
for any $\gamma \in \left(0, \frac{1}{2}\right)$ and some constant $c_2>0$ depending only on $\gamma$. Thus,  
\begin{displaymath}
 \left\|F_{V_m(z)} - \Phi \right\|_\infty \le c_1 \, m^{2\delta - \frac{1}{2}} + c m^{-3/2}
 + c_2 \, m^{-\gamma}+ \tilde{C}_2{m^{2\delta - \frac{1}{2}}} +
 \tilde{\mathcal{C}}{m^{-\delta}}.
 \end{displaymath}
To conclude, it is easy to check that for every $\delta \in \left(0, \frac{1}{4}\right)$ and every $\gamma \in \left(\frac{1}{6}, \frac{1}{2}\right)$, 
\begin{displaymath}
\min\left(\frac{3}{2}, \delta,   -2\delta +\frac{1}{2} , \gamma \right) \ge \frac{1}{6}, 
\end{displaymath}
 whence for every $z \in [\zeta_0, \zeta_1]$ it holds 
\begin{displaymath}
 \left\|F_{V_m(z)} - \Phi \right\|_\infty \le  \bar{C} \, m^{-\frac{1}{6}}
\end{displaymath}
for some constant $\bar{C}>0$. This concludes the proof of \eqref{eq:be_Qm}. 
     
\noindent \underline{Step 1.} Let $f_{B_m}$ denote the density of $B_m$ and write
\begin{align*}
\varphi_{V_m(z)} (\xi) & = e^{- i \xi \sqrt{m} \frac{\mu(z)}{\sigma(z)}} \cdot \varphi_{Q_m(z)}\left(  \frac{ \xi }{\sqrt{m} \sigma(z)}\right) \\
& = e^{- i \xi \sqrt{m} \frac{\mu(z)}{\sigma(z)}} \cdot \mathbb{E} \left[\varphi_{Q\left(\lfloor mz \rfloor, B_m\right)}\left(  \frac{ \xi}{\sqrt{m} \sigma(z)}\right)\right]\\
& = \int_{0}^1 \exp \left\{ - i \xi \sqrt{m} \frac{\mu(z)}{\sigma(z)} +  \lfloor{mz}\rfloor  \log \left[ 1 + p \left(e^{ i \xi  \frac{1}{\sqrt{m} \sigma(z)}} -1 \right)\right]\right\} \, f_{B_m}(p) \, \mathrm{d}p
\end {align*}
If $\xi$ satisfies 
\begin{equation}
\label{eq: BE_xi_bound}
| \xi| \le \mathcal{C}\,\sigma(z) \, m^\delta, 
\end{equation}
then \citep[Chapter IV, Lemma 5]{Pet(75)} guarantees that, for every $p \in [0, 1]$, 
\begin{displaymath}
   \left|   p  \left(e^\frac{i \, \xi }{\sqrt{m} \sigma(z)}-1\right)\right| \le p \, \left| \frac{\xi }{\sqrt{m} \sigma(z)}\right| \le C \, m^{\delta - \frac{1}{2}}.
\end{displaymath} 
Then, we can apply Taylor's formula to $\log \left[ 1 + p \left(e^{  \frac{i \xi }{\sqrt{m}\sigma(z)}} -1 \right)\right]$ and then to $\left(e^\frac{i \, \xi }{\sqrt{m} \sigma(z)} - 1\right)$ to obtain
\begin{align*}
\varphi_{V_m(z)} (\xi)  &=  \int_0^1 \exp \left\{ - i  \frac{\xi } {\sigma(z)} \sqrt{m} \left[\mu(z)  - z p \right] - \frac{\xi^2}{2} \frac{z}{\sigma^2(z)} \left[p - p^2 \right]\right\}  \, \mathfrak{R}_m(p, z, \xi) \, f_{B_m}(p) \, \mathrm{d}p\\
& =  \mathbb{E} \left[ e^{ - i  \frac{\xi } {\sigma(z)} G_m(z)- \frac{\xi^2}{2} S_m(z)}  \cdot \mathfrak{R}_m\left( B_m, z, \xi\right) \right]
\end{align*}
where
\begin{align*}
\mathfrak{R}_m(p, z, \xi) &= \exp \left\{(mz - \lfloor mz \rfloor ) \, \log \left[ 1 + p \left(e^{ i \xi  \frac{1}{\sqrt{m} \sigma(z)}} -1 \right) \right]  \right.\\
& \quad\quad\quad  +mz p \left[\left(e^\frac{i \, \xi }{\sqrt{m} \sigma(z)} - 1\right) -\left(\frac{i \, \xi }{\sqrt{m} \sigma(z)} - \frac{\xi^2 }{2 m \sigma^2(z)}\right) \right]\\
& \quad\quad\quad + \frac{1}{2}\, mz  p^2 \left[ \left(e^\frac{i \, \xi }{\sqrt{m} \sigma(z)} - 1\right)^2 + \frac{\xi^2}{m \sigma^2(z)}\right]\\
   &\quad\quad\quad \left. + m z \, p^3  \left(e^\frac{i \, \xi }{\sqrt{m} \sigma(z)} - 1\right)^3 \int_0^1 \left(2 + tp \left(e^\frac{i \, \xi }{\sqrt{m} \sigma(z)} - 1\right)\right)^{-2} \, (1-t)^2 \, \mathrm{d} t\right\}.
\end{align*}
Making use of elementary properties of the exponential together with \citep[Chapter IV, Lemma 5]{Pet(75)}, we can prove that there exists a constant $\tilde{C}_1$ such that for every $z \in [\zeta_0, \zeta_1]$,  every $p \in [0, 1]$ and every $\xi$ satisfying \eqref{eq: BE_xi_bound}
$$
\left|  \mathfrak{R}_m(p, z, \xi) -1 \right|   \le \tilde{C}_1 \, 
| \xi | \, m^{2\delta - \frac{1}{2}}.
$$
It follows that 
\begin{equation}
\label{char_Vm}
\varphi_{V_m(z)} (\xi)  =  \mathbb{E}\left[ e^{ - i  \frac{\xi } {\sigma(z)} G_m(z)- \frac{\xi^2}{2} S_m(z)} \right] + \mathcal{R}_1(m)
\end{equation}
 with 
\begin{align*}
\left|\mathcal{R}_1(m)\right| &= \left| \varphi_{V_m(z)} (\xi)  -  \mathbb{E}\left[ e^{ - i  \frac{\xi } {\sigma(z)} G_m(z)- \frac{\xi^2}{2} S_m(z)} \right] \right|  \\
&= \left|  \mathbb{E} \left[ e^{ - i  \frac{\xi } {\sigma(z)} G_m(z)- \frac{\xi^2}{2} S_m(z)}  \cdot \left[\mathfrak{R}_m\left( B_m, z, \xi\right) -1 \right]\right] \right| \\
& \le \int_0^1  \left| \exp \left\{ - i  \frac{\xi } {\sigma(z)} \sqrt{m} \left[\mu(z)  - z p \right] - \frac{\xi^2}{2} \frac{z}{\sigma^2(z)} \left[p - p^2 \right]\right\}  \right| \cdot \left|\mathfrak{R}_m\left( p, z, \xi\right) -1 \right|  \, f_{B_m}(p) \mathrm{d}p \\
& \le \int_0^1 \tilde{C}_1 \, |\xi| \, m^{2\delta - \frac{1}{2}} \, f_{B_m}(p) \mathrm{d}p  = \tilde{C}_1 \, |\xi| \, m^{2\delta - \frac{1}{2}} 
\end{align*}

\noindent \underline{Step 2.} Now note that
$$s^2(z) = \sigma^2(z) \cdot \frac{\alpha }{\alpha z + \tau + \nu } $$
and therefore
\begin{align}
\label{char_phi}
\nonumber\mathbb{E}\left[ e^{ - i  \frac{\xi } {\sigma(z)} G(z)- \frac{\xi^2}{2} S(z)} \right]& = \varphi_{G(z)} \left(\frac{\xi}{\sigma(z)}\right)  \cdot e^{\frac{\xi^2}{2} S(z)}\\
&= \nonumber  \exp\left\{-\frac{\xi^2}{2} \cdot \left[ \frac{\alpha z}{\alpha z + \tau + \nu} +\frac{\tau + \nu}{\alpha z + \tau + \nu }  \right]\right\}\\
&= e^{-\frac{\xi^2}{2}}
\end{align}

\noindent \underline{Step 3.} Using \eqref{char_Vm} and \eqref{char_phi}, write
\begin{align*}
     &\int_{| \xi| \le \mathcal{C}\,\sigma(z) \, m^\delta} \left|\frac{\varphi_{V_m(z)} (\xi) - e^{- \frac{\xi^2}{2}}}{\xi} \right| \, \mathrm{d} \xi   \\
    & \quad \quad = \int_{| \xi| \le \mathcal{C}\,\sigma(z) \, m^\delta} \left|\frac{\mathbb{E}\left[ e^{ - i  \frac{\xi } {\sigma(z)} G_m(z)- \frac{\xi^2}{2} S_m(z)} \right] + \mathcal{R}_1(m) - \mathbb{E}\left[ e^{ - i  \frac{\xi } {\sigma(z)} G(z)- \frac{\xi^2}{2} S(z)} \right] }{\xi} \right| \, \mathrm{d} \xi \\
    & \quad \quad \le \int_{| \xi| \le \mathcal{C}\,\sigma(z) \, m^\delta} \left|\frac{\mathbb{E}\left[ e^{ - i  \frac{\xi } {\sigma(z)} G_m(z)- \frac{\xi^2}{2} S_m(z)} \right] - \mathbb{E}\left[ e^{ - i  \frac{\xi } {\sigma(z)} G(z)- \frac{\xi^2}{2} S(z)} \right] }{\xi} \right| \, \mathrm{d} \xi 
    + \tilde{C}_2 \, m^{2\delta - \frac{1}{2}}
        \end{align*}
where $\tilde{C}_2 =2 \, \mathcal{C}\,  \tilde{C}_1\,  \max_{z \in [\zeta_0, \zeta_1]}  \sigma(z)$.

\noindent \underline{Step 4.} 
By definition of $\mathcal{I}_m^{(1)} (z)$,  Jensen's inequality and the triangular inequality,  
 \begin{align*}
 \label{Im1}  \mathcal{I}_m^{(1)} (z) &= \int_{- \mathcal{C}\sigma(z) m^\delta}^{\mathcal{C}\sigma(z)  m^\delta} \left|\frac{\mathbb{E}\left[ e^{ - i  \frac{\xi } {\sigma(z)} G_m(z)} \left(e^{- \frac{\xi^2}{2} S_m(z)}  - e^{- \frac{\xi^2}{2} S(z)}\right)  \right] }{\xi} \right| \, \mathrm{d} \xi \\
 &\le  \int_{- \mathcal{C}\sigma(z) m^\delta}^{\mathcal{C}\sigma(z)  m^\delta}  \frac{\mathbb{E}\left[  \left|e^{- \frac{\xi^2}{2} S_m(z)}  - e^{- \frac{\xi^2}{2} S(z)}\right|  \right] }{|\xi|}  \, \mathrm{d} \xi
 \end{align*}
Now note that, for every $z \in [\zeta_0, \zeta_1]$, $S_m(z) \ge 0 $ a.s., and $S(z) \ge 0$. Since the function $[0, +\infty) \to (0, 1]; \, x \mapsto e^{-x}$ is 1-Lipschitz, the last term of the above inequality can be bounded by
\begin{displaymath}
\int_{- \mathcal{C}\sigma(z) m^\delta}^{\mathcal{C}\sigma(z)  m^\delta}  \frac{\mathbb{E}\left[  \left|e^{- \frac{\xi^2}{2} S_m(z)}  - e^{- \frac{\xi^2}{2} S(z)}\right|  \right] }{|\xi|}  \, \mathrm{d} \xi\le   \int_{- \mathcal{C}\sigma(z) m^\delta}^{\mathcal{C}\sigma(z)  m^\delta} \frac{|\xi|}{2} \, \mathbb{E}\left[ \left|S_m(z) - S(z)\right|\right]  \, \mathrm{d} \xi,
\end{displaymath}
such that
\begin{equation}\label{number}
 \mathcal{I}_m^{(1)} (z)\le \sqrt{\mathbb{E}\left[ \left(S_m(z) - S(z)\right)^2\right]}  \cdot \mathcal{C}^2 \sigma^2(z) \, m^{2\delta}.
\end{equation}
A simple computation shows
\begin{displaymath}
\mathbb{E}[S_m(z)] = S(z) \left[ 1 + O \left(m^{-1}\right)\right], 
\end{displaymath}
which in turn entails
\begin{displaymath}
\mathbb{E}\left[ \left(S_m(z) - S(z)\right)^2\right] = \left[ \mathbb{E}\left[S_m^2(z)\right]- S^2(z)\right] \cdot \left[1 + O \left(m^{-2}\right) \right].
\end{displaymath}
The second moment of $S_m(z)$ can be written as follows
\begin{displaymath}
 \mathbb{E}\left[S_m^2(z)\right] = \frac{z^2}{\sigma^4(z)} \cdot \left\{\mathbb{E}\left[B_m^2\right] -2 \mathbb{E}\left[B_m^3\right] + \mathbb{E}\left[B_m^4\right] \right\},
\end{displaymath}
reducing the problem to the study of the moments of $B_m$. By standard results regarding the beta distribution it is known that, for $k \in \mathbb{N}$, 
\begin{displaymath}
\mathbb{E} \left[B_m^k \right] = \frac{\Gamma \left(\frac{\tau + \varrho \alpha}{\alpha}\,  m  + k\right)\, \Gamma \left(\frac{\tau + \nu}{\alpha}\,  m \right)}{\Gamma \left(\frac{\tau + \varrho \alpha}{\alpha}\,  m  \right) \, \Gamma \left(\frac{\tau + \nu }{\alpha}\,  m  + k\right)},
\end{displaymath}
so that, for $m \to +\infty$, a straightforward application of \citep[formula (1)]{TE(51)} yields
\begin{align*}
\mathbb{E} \left[B_m^k \right] &  = \left(\frac{\tau + \varrho \alpha}{\alpha} \, m\right)^k \left[1+ \frac{k(k-1) \alpha}{2 (\tau + \varrho \alpha) \, m } + O \left(m^{-2}\right) \right]\\
& \quad + \left(\frac{\tau + \varrho \alpha}{\alpha} \, m\right)^k \left[1- \frac{k(k-1) \alpha}{2 (\tau + \nu) \, m } + O \left(m^{-2}\right) \right]\\
&= \left( \frac{\tau+ \varrho \alpha}{\tau + \nu} \right)^k \, \left[ 1 + \frac{k(k-1)}{m}  \, \eta
+ O \left(m^{-2}\right) \right]
\end{align*}
where $\eta = \eta (\tau, \nu, \varrho, \alpha) := \frac{ \alpha (\nu - \varrho \alpha) }{2 (\tau +\varrho \alpha) (\tau+\nu)}$ is a constant not depending on $k$ or $m$.
Then, 
\begin{align*}
 \mathbb{E}\left[S_m^2(z)\right] &= \frac{z^2}{\sigma^4(z)} \cdot  \left( \frac{\tau+ \varrho \alpha}{\tau + \nu} \right)^2 \cdot \left\{  \left[ 1 + \frac{2	\eta}{m} \right]  -2 \frac{\tau+ \varrho \alpha}{\tau + \nu} \left[ 1 + \frac{6\eta}{m}\right] \right.  \\
 & \quad \quad \quad \quad \quad \quad \quad \quad  \quad \quad \quad \left. +  \left( \frac{\tau+ \varrho \alpha}{\tau + \nu} \right)^2 \left[ 1 + \frac{12 \eta}{m}\right]+ O \left(m^{-2}\right)  \right\}\\
 & =  \frac{z^2}{\sigma^4(z)} \cdot  \left( \frac{\tau+ \varrho \alpha}{\tau + \nu} \right)^2 \cdot \left\{ \left(\frac{\nu - \varrho \alpha}{\tau + \nu}\right)^2 + \frac{1}{m} \cdot \left[ 2 \eta - 12 \eta \frac{\tau+\varrho \alpha}{\tau + \nu} + 12 \eta \left(\frac{\tau+\varrho \alpha}{\tau + \nu} \right)^2\right] \right.
 \\
 &\quad \quad \quad \quad \quad \quad \quad \quad  \quad \quad \quad  \left.+ O\left(m^{-2}\right)\right\}\\
 & = S^2(z) + \frac{1}{m} \, g(z, \tau, \nu, \varrho, \alpha) + O\left(m^{-2}\right), 
\end{align*}
with $g(z, \tau, \nu, \varrho, \alpha): =   \frac{z^2}{\sigma^4(z)} \cdot  \left( \frac{\tau+ \varrho \alpha}{\tau + \nu} \right)^2 \cdot \left[ 2 \eta - 12 \eta \frac{\tau+\varrho \alpha}{\tau + \nu} + 12 \eta \left(\frac{\tau+\varrho \alpha}{\tau + \nu} \right)^2\right]$. We can conclude that
\begin{displaymath}
\sqrt{\mathbb{E}\left[ \left(S_m(z) - S(z)\right)^2\right]}  = \frac{1}{\sqrt{m}} \, g(z, \tau, \nu, \varrho, \alpha) \left[1 +r_m\right]
\end{displaymath} 
with $r_m = O\left(m^{-1}\right)$. Since all the asymptotic expansions above hold uniformly on compact sets, in the sense of \citep[Definition A.9]{Con(24)}, and the function $z \to g(z, \tau, \nu, \varrho, \alpha)$ is continuous, one can resort to \citep[Lemma A.11]{Con(24)} to obtain
\begin{displaymath}
|R_m| : = \left| \sqrt{\mathbb{E}\left[ \left(S_m(z) - S(z)\right)^2\right]}  - \frac{1}{\sqrt{m}} \, g(z, \tau, \nu, \varrho, \alpha)\right| \le \frac{c}{m^{3/2}}
\end{displaymath}
fro some suitable positive constant $c$. To conclude, plug the above result in  \eqref{number} to obtain that for all $z \in [\zeta_0, \zeta_1]$, 
\begin{displaymath}
\mathcal{I}_m^{(1)}(z) \le c_1 \,  m^{2\delta - \frac{1}{2}} + \frac{c}{m^{3/2}}
\end{displaymath}
 for some suitable constant $c_1>0$. 

\noindent \underline{Step 5.} By definition, 
\begin{align*}
\mathcal{I}_m^{(2)}(z) &= \int_{-\mathcal{C} \sigma(z) m^\delta}^{\mathcal{C} \sigma(z) m^\delta} e^{-\frac{\xi^2 S(z)}{2}} \, \left| \int_{\R}   \frac{e^{i \xi x }-1}{\xi} \, \left[f_{G_m(z)}(x) - f_{G(z)}(x) \right] \, \mathrm{d} x\, \right| \, \mathrm{d} \xi \\
&\le  \int_{-\mathcal{C} \sigma(z) m^\delta}^{\mathcal{C} \sigma(z) m^\delta} e^{-\frac{\xi^2 S(z)}{2}} \, \int_{\R}  \left|  \frac{e^{i \xi x }-1}{\xi} \right|\, \left|f_{G_m(z)}(x) - f_{G(z)}(x) \right| \, \mathrm{d} x \, \mathrm{d} \xi \\
&=   \int_{\R} \left[ \int_{-\mathcal{C} \sigma(z) m^\delta}^{\mathcal{C} \sigma(z) m^\delta} e^{-\frac{\xi^2 S(z)}{2}} \,  \left|  \frac{e^{i \xi x }-1}{\xi} \right| \, \mathrm{d} \xi \right] \, \left|f_{G_m(z)}(x) - f_{G(z)}(x) \right| \, \mathrm{d} x 
\end{align*}
Making use of \citep[Chapter IV, Lemma 5]{Pet(75)} and of the elementary properties of the Gaussian density write the bound
\begin{align*}
 \int_{-\mathcal{C} \sigma(z) m^\delta}^{\mathcal{C} \sigma(z) m^\delta} e^{-\frac{\xi^2 S(z)}{2}} \,  \left|  \frac{e^{i \xi x }-1}{\xi} \right| \, \mathrm{d} \xi  & \le  \int_{-\mathcal{C} \sigma(z) m^\delta}^{\mathcal{C} \sigma(z) m^\delta} e^{-\frac{\xi^2 S(z)}{2}} \,  \left|  x \right| \, \mathrm{d} \xi\\
 & \le  |x| \,  \int_{\R} e^{-\frac{\xi^2 S(z)}{2}}  \, \mathrm{d} \xi \\
 & =  |x| \cdot \sqrt{\frac{2 \pi}{S(z)}}
\end{align*}
Hence, letting $c(z) =  \sqrt{\frac{2 \pi}{S(z)}}$, 
\begin{equation}
\label{bound_dens}
\mathcal{I}_m^{(2)} \le   c(z) \cdot  \int_{\R} |x| \cdot  \left|f_{G_m(z)}(x) - f_{G(z)}(x) \right| \, \mathrm{d} x 
\end{equation}
For $x \in \mathbb{R}$, 
\begin{displaymath}
f_{G_m(z)}(x) = \begin{cases} 0 \, \,  \quad  \quad \quad \quad \quad  \quad \quad \text{if } x \notin \left[- \sqrt{m} \mu(z) , \, \sqrt{m} (z- \mu(z) ) \right]& \\
\frac{1}{\sqrt{m} z} \, \frac{\Gamma \left( \frac{\tau + \nu}{\alpha} \, m\right)}{\Gamma \left( \frac{\tau + \varrho \alpha }{\alpha} \, m\right)\, \Gamma \left( \frac{\nu - \varrho \alpha }{\alpha} \, m\right)}\cdot  \left(\frac{\tau + \varrho \alpha}{\tau + \nu} - \frac{x}{\sqrt{m} z}\right)^{ \frac{\tau + \varrho \alpha }{\alpha} \, m -1} \cdot  \left(\frac{\nu - \varrho \alpha}{\tau + \nu} + \frac{x}{\sqrt{m} z}\right)^{ \frac{\nu - \varrho \alpha }{\alpha} \, m -1}  & \\
 \quad =: \psi_m(x)   \quad  \quad \quad \text{if } x \notin \left[- \sqrt{m} \mu(z) , \, \sqrt{m} (z- \mu(z) ) \right]& 
\end{cases}
\end{displaymath}
and 
\begin{displaymath}
f_{G(z)}(x) = \frac{1}{\sqrt{2 \pi s^2(z)}} \exp\left\{- \frac{x^2}{2 s^2(z)} \right\}.
\end{displaymath}
Introduce $\gamma \in (0, 1/2)$ and let $A_m = c \,  m^\gamma$, with $c$ a positive constant such that $-\mu(z) \sqrt{m}< -A_m < 0< A_m < (z- \mu(z)) \sqrt{m}$ for all $z \in [\zeta_0, \zeta_1]$, and write the integral in the RHS of \eqref{bound_dens} as
\begin{align*}
&  \int_{\R} |x| \cdot  \left|f_{G_m(z)}(x) - f_{G(z)}(x) \right| \, \mathrm{d} x   \\
  & \quad \quad =   \int_{-\infty}^{- \sqrt{m} \, \mu(z)}-x \, f_{G(z)}(x) \, \mathrm{d} x   +  \int^{-A_m}_{- \sqrt{m} \mu(z)} |x|  \cdot \left| \psi_m(x) -  f_{G(z)}(x)\right| \, \mathrm{d} x \\
&\quad \quad  \quad +  \int^{ -A_m}_{A_m} |x|  \cdot \left| \psi_m(x) -  f_{G(z)}(x)\right|  \, \mathrm{d} x \\
&\quad  \quad \quad+  \int^{ \sqrt{m}\, ( z-  \mu(z))}_{A_m}|x|  \cdot \left| \psi_m(x) -  f_{G(z)}(x)\right| \, \mathrm{d} x  +  \int_{ \sqrt{m}\, ( z-  \mu(z))}^{+\infty}x   \, f_{G(z)}(x) \, \mathrm{d} x \\
&\quad  \quad \le 		 \int_{-\infty}^{- A_m}-x \, f_{G(z)}(x) \, \mathrm{d} x   +  \int^{-A_m}_{- \sqrt{m} \, \mu(z)}|x| \, \psi_m(x) \, \mathrm{d} x \\
&\quad \quad \quad+  \int^{ -A_m}_{A_m} |x|  \cdot \left| \psi_m(x) -  f_{G(z)}(x)\right|  \, \mathrm{d} x \\
& \quad \quad \quad +  \int^{ \sqrt{m}\, ( z-  \mu(z))}_{A_m}|x |  \, \psi_m(x) \, \mathrm{d} x  +  \int_{ \sqrt{m}\, ( z-  \mu(z))}^{+\infty}x   \, f_{G(z)}(x) \, \mathrm{d} x \\
& \quad \quad = \mathrm{(I)} +\mathrm{(II)} +\mathrm{(III)} + \mathrm{(IV)} + \mathrm{(V)}
\end{align*}
The terms (I) and (V) vanish exponentially fast as $m \to +\infty$, in fact for any $\zeta>0$
\begin{displaymath}
 \int_{\zeta}^{+\infty}x   \, f_{G(z)}(x) \, \mathrm{d} x = \frac{s(z)}{\sqrt{2\pi}} \, e^{ -\frac{\zeta^2}{s^2(z) }}, 
\end{displaymath}
whence
\begin{displaymath}
 \mathrm{(I)} =  \frac{s(z)}{\sqrt{2\pi}} \, e^{ -\frac{c^2 \, m^{2\gamma}}{s^2(z) }}
\end{displaymath}
and 
\begin{displaymath}
 \mathrm{(V)} =  \frac{s(z)}{\sqrt{2\pi}} \, e^{ -\frac{c^2 \, m^{2\gamma}}{s^2(z) }}.
\end{displaymath}
We now study (II) and (IV). By a straightforward application of Cauchy-Schwartz inequality, 
\begin{align*}
\mathrm{(II)} &= \mathbb{E} \left[|G_m(z)| \cdot \mathds{1}_{\left\{G_m(z) \in [-\sqrt{m} \, \mu(z), -A_m ]\right\}}\right] \\
& \le \sqrt{\operatorname{Var}\left(G_m(z)\right)} \cdot \sqrt{P\left[G_m(z) \in [-\sqrt{m} \, \mu(z), -A_m ] \right]}.
\end{align*}
Now, 
\begin{displaymath}
\operatorname{Var}\left(G_m(z)\right) = mz^2 \operatorname{Var}\left(B_m\right) =  z^2\, \frac{ \alpha \, (\tau + \varrho \alpha)(\nu - \varrho \alpha)}{ (\tau + \nu)^3  } \, \left[1  + O\left(\frac{1}{m}\right)\right] \le  z^2 \, \tilde{c}_1 
\end{displaymath}
for some suitable constant $\tilde{c}_1$. By Chebichev's inequality
\begin{align*}
P\left(G_m(z) \in [-\sqrt{m} \, \mu(z), -A_m ] \right) & \le P\left(G_m(z) \le  -A_m  \right) + P\left(G_m(z)  \ge A_m  \right) \\
& \le \frac{\operatorname{Var}(G_m(z))}{A_m^2}\\
&  \le z^2 \, \tilde{c}_2 \, m^{-2\gamma}
 \end{align*}
for a suitable constant $\tilde{c}_2$. Combining these results we obtain
\begin{displaymath}
\mathrm{(II)} \le z^2 \, \tilde{c} \, m^{-\gamma}
\end{displaymath}
for $	\tilde{c} = \sqrt{\tilde{c}_1 \, \tilde{c}_2}$. The same argument allows to prove
\begin{displaymath}
\mathrm{(IV)} \le z^2 \, \tilde{c} \, m^{-\gamma}
\end{displaymath}
To assess the behavior of (III), use the well--known asymptotic expansion of the Gamma function for large argument \citep[Equation \href{https://dlmf.nist.gov/5.11.E10}{(5.11.3)}]{nist} and some simple algebraic rearrangements to write that, as $m \to +\infty$, 
\begin{align*}
& \psi_m(x)\\
& = \frac{1}{\sqrt{2\pi}} \, \sqrt{\frac{(\tau + \varrho \alpha)(\nu - \varrho \alpha)}{z^2 \alpha (\tau + \nu)}} \cdot \left(\frac{\tau+\varrho\alpha}{\tau + \nu} - \frac{x}{\sqrt{m} \, z}\right)^{-1} \cdot \left(\frac{\nu - \varrho\alpha}{\tau + \nu} + \frac{x}{\sqrt{m} \, z}\right)^{-1} \\
&\quad  \cdot \left\{\left(\frac{\tau + \nu}{\tau + \varrho \alpha}\right)^{\frac{\tau + \varrho \alpha}{\alpha}} \, \left(\frac{\tau + \nu}{\nu - \varrho \alpha}\right)^{\frac{\nu - \varrho \alpha}{\alpha}}  \cdot \left(\frac{\tau+\varrho\alpha}{\tau + \nu} - \frac{x}{\sqrt{m} \, z}\right)^{\frac{\tau + \varrho \alpha}{\alpha}} \cdot \left(\frac{\nu - \varrho\alpha}{\tau + \nu} + \frac{x}{\sqrt{m} \, z}\right)^{\frac{\nu -  \varrho \alpha}{\alpha}}\right\}^m \\
& \quad 	\cdot \mathfrak{R}_1^{(m)}\\
& =   \exp \left\{ m\cdot  \left[ \frac{\tau + \varrho \alpha}{\alpha} \log \left(1- \frac{x}{\sqrt{m}\, \mu(z)}\right)+ \frac{\nu - \varrho \alpha}{\alpha} \log \left(1+ \frac{x}{\sqrt{m}\, (z- \mu(z))}\right)\right] \right\} \\
& \quad \cdot \frac{1}{\sqrt{2\pi}} \, \Psi_m(z, x)\cdot \mathfrak{R}_1^{(m)} 
\end{align*}
where
\begin{displaymath}
\Psi_m(z, x) =\sqrt{\frac{(\tau + \varrho \alpha)(\nu - \varrho \alpha)}{ z^2 \, \alpha (\tau + \nu)}} \cdot \frac{1}{\left(\frac{\tau+\varrho\alpha}{\tau + \nu} - \frac{x}{\sqrt{m} \, z}\right) \cdot \left(\frac{\nu - \varrho\alpha}{\tau + \nu} + \frac{x}{\sqrt{m} \, z}\right)}
\end{displaymath}
and
\begin{displaymath}
\mathfrak{R}_1^{(m)} = \frac{\Gamma^* \left( \frac{\tau + \nu}{\alpha} \, m\right)}{\Gamma^* \left( \frac{\tau + \varrho \alpha }{\alpha} \, m\right)\, \Gamma^* \left( \frac{\nu - \varrho \alpha }{\alpha} \, m\right)}.
\end{displaymath}
We focus first on the behavior of these last two terms as $m \to +\infty$; simple algebraic rearrangements allow to write
\begin{align*}
\Psi_m(z, x) &=\sqrt{\frac{ \alpha(\tau + \nu)^3}{z^2 (\tau+\varrho\alpha)(\nu - \varrho\alpha)}}\,  \cdot \frac{1}{1 + \frac{x (\tau+\nu)}{\sqrt{m} \, z} \, \left(\frac{1}{\nu -  \varrho \alpha} - \frac{1}{\tau + \varrho \alpha} \right)- \frac{x^2 (\tau + \nu)^2}{m z (\tau + \varrho \alpha)(\nu - \varrho \alpha)} }\\
&= \frac{1}{s(z)} \,  \left[ 1+O \left(\frac{x}{m^{1/2}}\right)\right].
\end{align*}
uniformly for $z \in [\zeta_0, \zeta_1]$. Resorting again to \citep[Equations \href{https://dlmf.nist.gov/5.11.E10}{(5.11.3)} and \href{https://dlmf.nist.gov/5.11.E10}{(5.11.4)}]{nist}, we also obtain
\begin{displaymath}
\mathfrak{R}_1^{(m)}  = \frac{1+ \frac{\alpha}{12 \, (\tau+\nu) \, m } + O\left(m^{-2}\right)}{1+ \frac{\alpha \, (\tau+\nu)}{12 \, (\tau+\varrho \alpha)(\nu - \varrho \alpha) \, m } + O\left(m^{-2}\right)} = 1+  O\left(m^{-1}\right).
\end{displaymath}
uniformly for $z \in [\zeta_0, \zeta_1]$. For the exponential term, 
use Taylor's expansion of the logarithm around 1 to write
\begin{align*}
& \exp \left\{ m \cdot \left[\frac{\tau + \varrho \alpha}{\alpha} \log \left(1- \frac{x}{\sqrt{m}\, \mu(z)}\right)+ \frac{\nu - \varrho \alpha}{\alpha} \log \left(1+ \frac{x}{\sqrt{m}\, (z- \mu(z))}\right) \right]\right\} \\
& = \exp \left\{ m \cdot \left[ \frac{\tau + \varrho \alpha}{\alpha}  \left(- \frac{x}{\sqrt{m}\, \mu(z)} - \frac{x^2}{2m \, \mu^2(z)} \right)\right. \right.\\
& \left. \left.  \quad \quad \quad + \frac{\nu - \varrho \alpha}{\alpha}  \left( \frac{x}{\sqrt{m}\, (z- \mu(z))} - \frac{x^2}{2m\, (z- \mu(z))^2}\right)\right]\right\} \cdot \mathfrak{R}_2^{(m)}(x, z)\\
& = \exp \left\{ - \frac{x^2}{2} \cdot \frac{1}{s^2(z)}\right\} \cdot \mathfrak{R}_2^{(m)}(x, z)
\end{align*}
where
\begin{align*}
\mathfrak{R}_2^{(m)}(x, z)& = \exp \left\{ m \cdot \left[ \frac{\tau + \varrho \alpha}{\alpha}  \left(\log \left(1- \frac{x}{\sqrt{m}\, \mu(z)}\right) +\frac{x}{\sqrt{m}\, \mu(z)}  +\frac{x^2}{2m \, \mu^2(z)} \right)\right. \right.\\
& \left. \left.  \quad \quad \quad + \frac{\nu - \varrho \alpha}{\alpha}  \left( \log \left(1+ \frac{x}{\sqrt{m}\, (z- \mu(z))}\right)-\frac{x}{\sqrt{m}\, (z- \mu(z))} +\frac{x^2}{2m\, (z- \mu(z))^2}\right)\right]\right\} 
\end{align*}
As $m \to +\infty$,  
\begin{displaymath}
\mathfrak{R}_2^{(m)}(x, z) = 1+ O \left(\frac{x^3}{m^{1/2}}\right)
\end{displaymath}
uniformly for $z \in [\zeta_0, \zeta_1]$. Combining all the above results 
we can write
\begin{displaymath}
\psi_m(x) = \frac{1}{\sqrt{2\pi}\, s(z)} \,\exp \left\{ - \frac{x^2}{2} \cdot \frac{1}{s^2(z)}\right\} \left[1 + \mathcal{R}_m(x, z)\right] = f_{G(z)}(x) \left[1 + \mathcal{R}_m(x, z)\right]
\end{displaymath}
with
\begin{displaymath}
\mathcal{R}_m(x, z) =  O \left(\frac{x^3}{m^{1/2}}\right)
 \end{displaymath}
uniformly for $z \in [\zeta_0, \zeta_1]$. Resorting again to \citep[Lemma A.11]{Con(24)}, this in turn implies
\begin{displaymath}
\left|\psi_m(x)- f_{G(z)}(x) \right| \le \frac{C\, x^3}{m^{1/2}}
\end{displaymath}
for every $x \in [-A_m, A_m]$ and $z \in [\zeta_0, \zeta_1]$, for some constant $C>0$, whence 
\begin{align*}
\mathrm{(III)} &=  \int^{ -A_m}_{A_m} |x|  \cdot \left| \psi_m(x) -  f_{G(z)}(x)\right|  \, \mathrm{d} x\\
& \le  \int^{ -A_m}_{A_m} |x|  \cdot\frac{C\, x^3}{m^{1/2}}  \, \mathrm{d} x \\
& \le   \frac{\mathcal{C} \, \mathbb{E} \left[ G(z)^4\right]}{m^{1/2}}\\
\end{align*}
In conclusion, 
\begin{align*}
\mathcal{I}_m^{(2)}(z) &\le \frac{s(z)}{\sqrt{2\pi}} \left[ e^{ -\frac{c^2 \, m^{2\gamma}}{s^2(z) }} + e^{ -\frac{c^2 \, m^{2\gamma}}{s^2(z) }} \right] + 2z^2 \, c \,  m^{- \gamma} +\mathcal{C} \, \mathbb{E} \left[ G(z)^4\right] \, m^{-1/2}
\end{align*}
Recalling that $\gamma \in \left(0, \frac{1}{2}\right)$, we conclude
\begin{align*}
\mathcal{I}_m^{(2)} (z)&\le C_2(z) \, m^{-\gamma} \le c_2 \, m^{-\gamma} 
\end{align*}
for some continuous function $C_2: [\zeta_0, \zeta_1] \to (0, +\infty)$ depending only on $\gamma$ and $c_2 = \max_{z \in [\zeta_0, \zeta_1]} C_2(z)$.
\end{proof}

\section{Additional numerical illustrations} \label{app:E}
This section collects additional figures displaying the performance gap between the Mittag-Leffler credible intervals and the Gaussian confidence intervals, and its behavior with respect to the additional sample size $m$ on the datasets, both synthetic and real, considered in Section \ref{sec4}. 

\subsection{Synthetic data}
Figure \ref{fig:coverage_s} complements the analysis of the synthetic datasets in Section \ref{sec4} (Table \ref{tab:1} and Table \ref{tab:2}); in particular, it displays the coverage of the Mittag-Leffler credible interval (blue) and of the Gaussian credible interval (red) as a function of $m \in [0, 5n]$. The coverages are evaluated at a uniform mesh of $50$ points over $[0, 5n]$, as for Figure \ref{fig1_intro}. Monte Carlo algorithms to obtain exact credible intervals and Mittag-Leffler credible intervals apply $2000$ Monte Carlo samples.

\begin{figure}[h!]
\begin{minipage}{0.5 \textwidth}
\begin{center}
\medskip
A) Zipf
\medskip

\includegraphics[width = \textwidth]{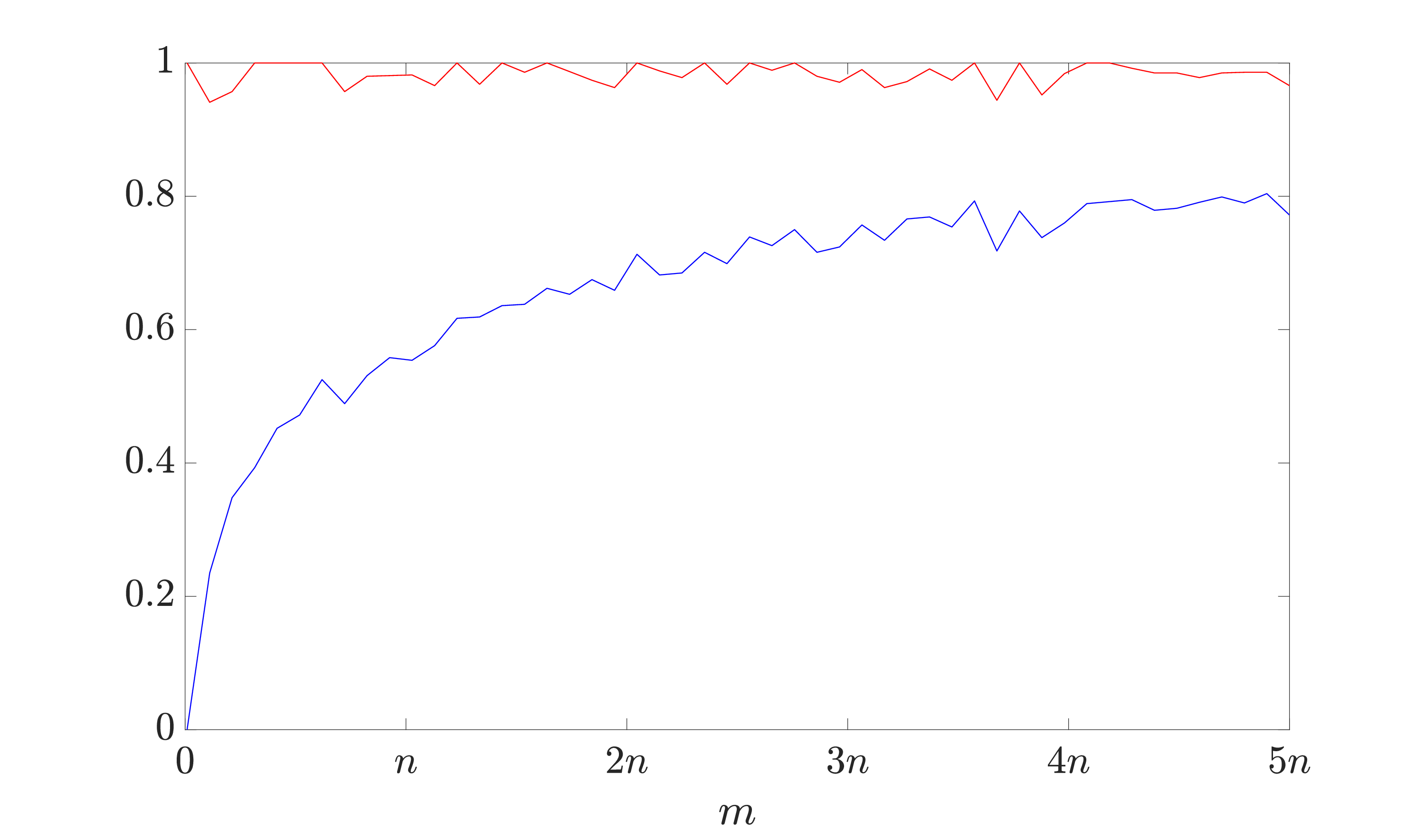}
\end{center}
\end{minipage}
\begin{minipage}{0.5 \textwidth}
\begin{center}
\medskip
B) Zipf
\medskip

\includegraphics[width = \textwidth]{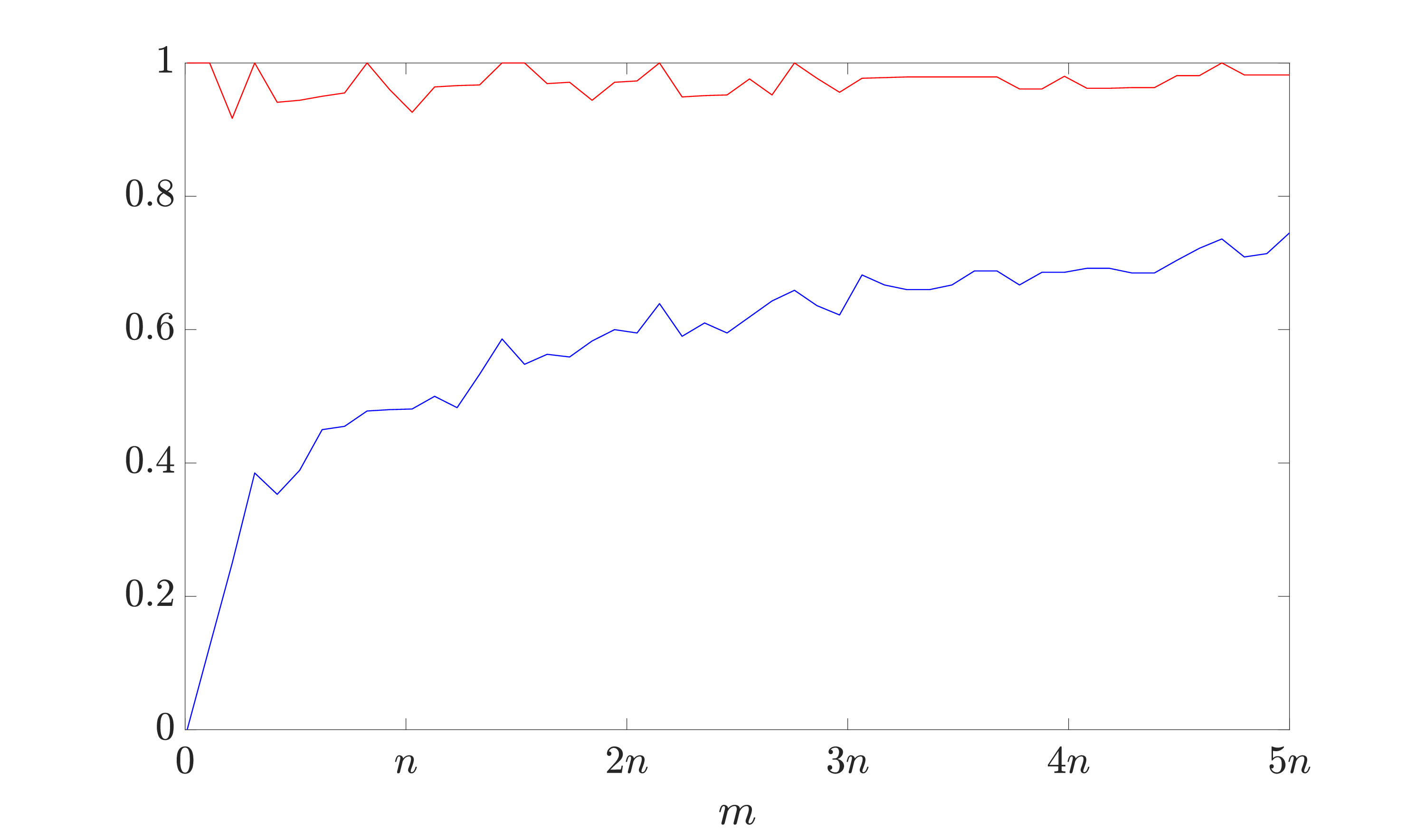}
\end{center}
\end{minipage}

\begin{minipage}{0.5 \textwidth}
\begin{center}
\medskip
C) P\'olya
\medskip

\includegraphics[width = \textwidth]{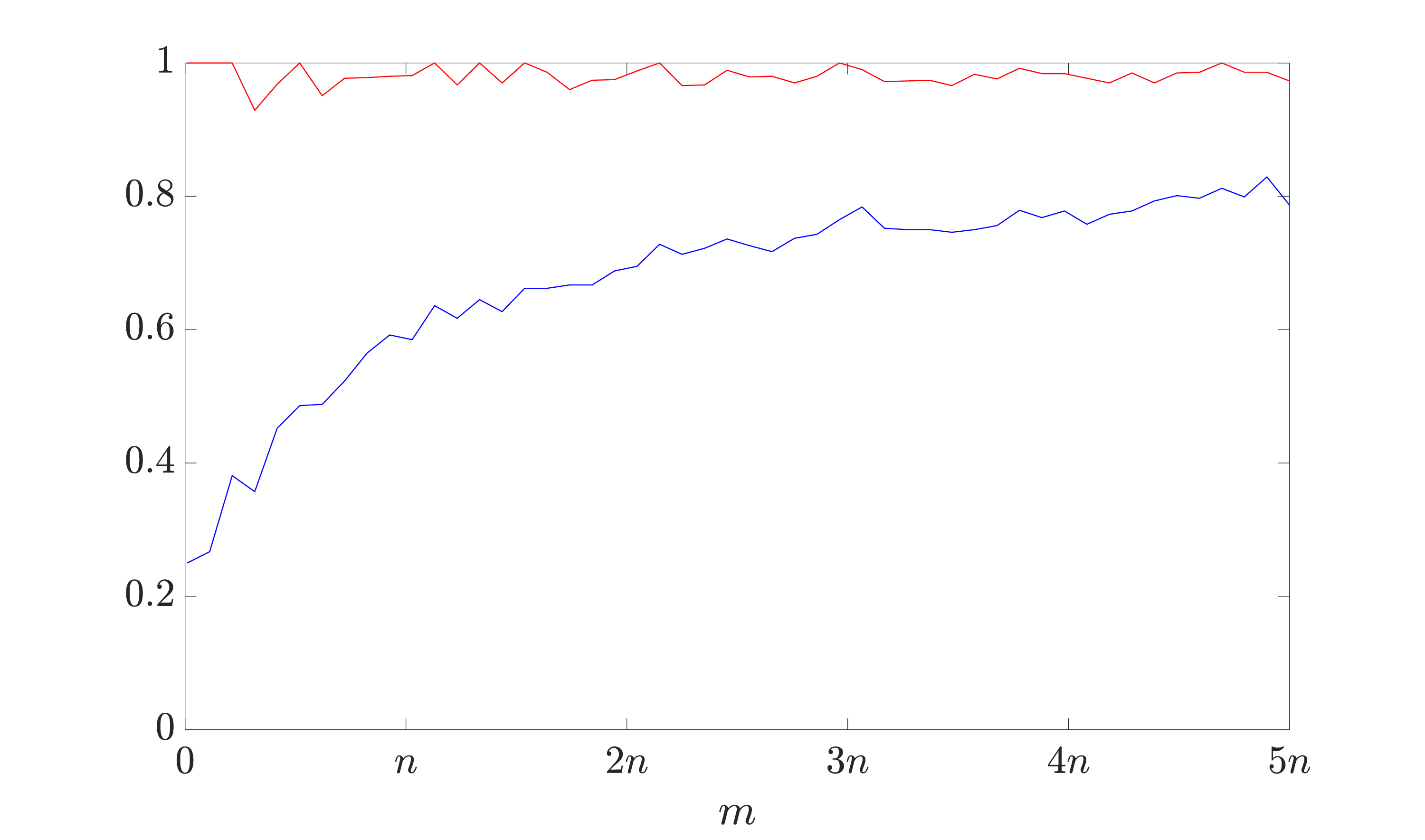}
\end{center}
\end{minipage}
\begin{minipage}{0.5 \textwidth}
\begin{center}
\medskip
D) Uniform
\medskip

\includegraphics[width = \textwidth]{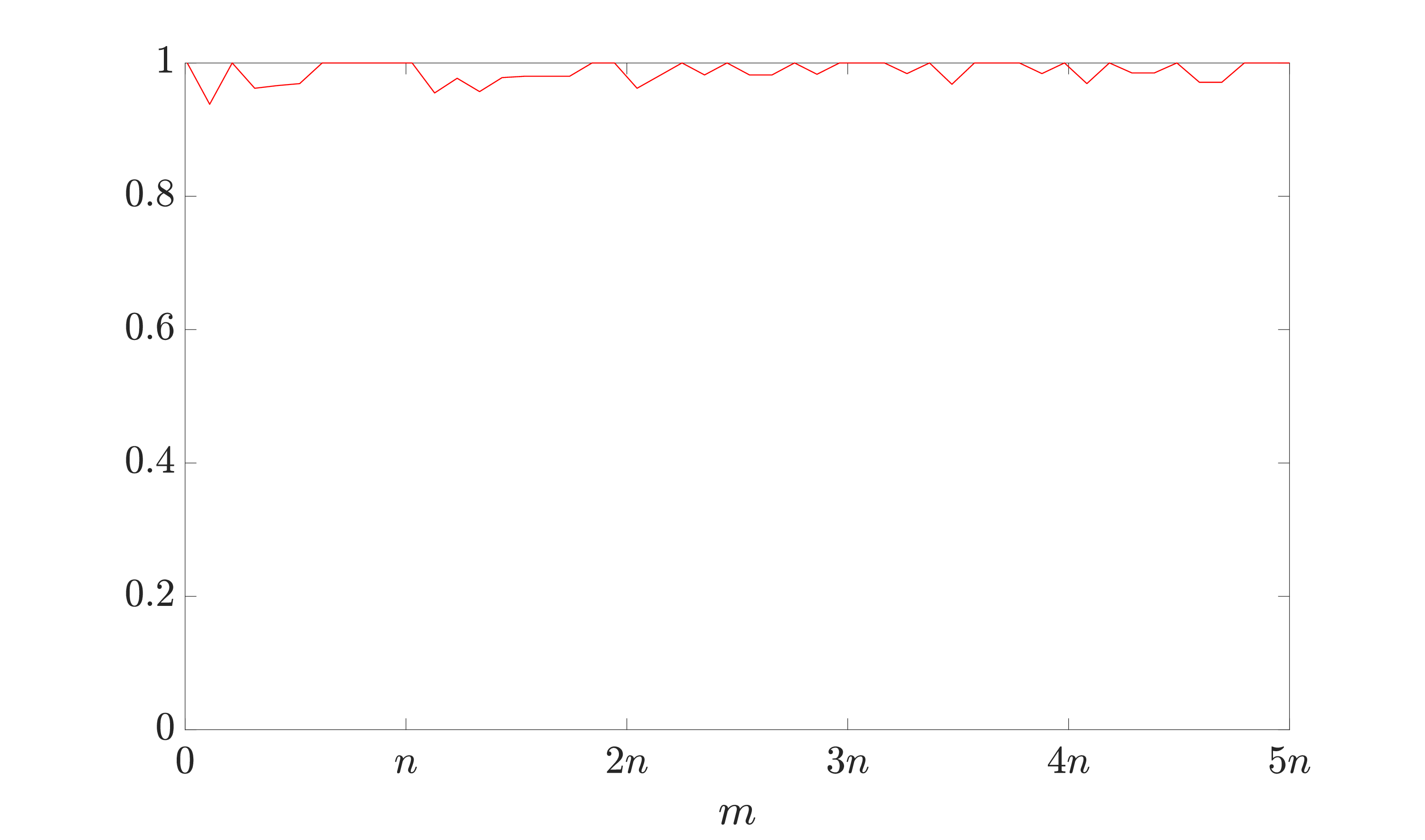}
\end{center}
\end{minipage}
\caption{\scriptsize{Coverage of the Mittag-Leffler credible interval (blue) and of the Gaussian credible interval (red) as a function of $m \in [0, 5n]$. }}
\label{fig:coverage_s}
\end{figure}

Figure \ref{fig:coverage_s} confirms the behavior of the coverage observed in Table \ref{tab:1}:  the coverage of the Gaussian credible intervals is nearly constant in $m$, with values oscillating between 95\% and 100\%. Instead, the coverage of the Mittag-Lefffler credible intervals is, for all values of $m$, lower than that of the corresponding Gaussian credible intervals; such a coverage increases in $m$.


\subsection{Real data}

Figure \ref{fig:naegAA} complements the analysis of the real EST datasets in Section \ref{sec4} (Table \ref{tab:3} and Table \ref{tab:4}); in particular, for the tomato flower, \emph{Mastigamoeba}, \emph{Mastigamoeba} normalized, and \emph{Naegleria} anaerobic EST datasets, it displays BNP estimates of $K_{n,m}$ with 95\% exact credible intervals, Mittag-Leffler credible intervals and Gaussian credible intervals as a function of $m \in [0, 5n]$. Credible intervals are evaluated at a uniform mesh of $50$ points over $[0, 5n]$, as for Figure \ref{fig:2}. Monte Carlo algorithms to obtain exact credible intervals and Mittag-Leffler credible intervals apply $2000$ Monte Carlo samples.

 \begin{figure}[h!]
\begin{center}
Tomato flower
\medskip

\includegraphics[width = 0.9\textwidth]{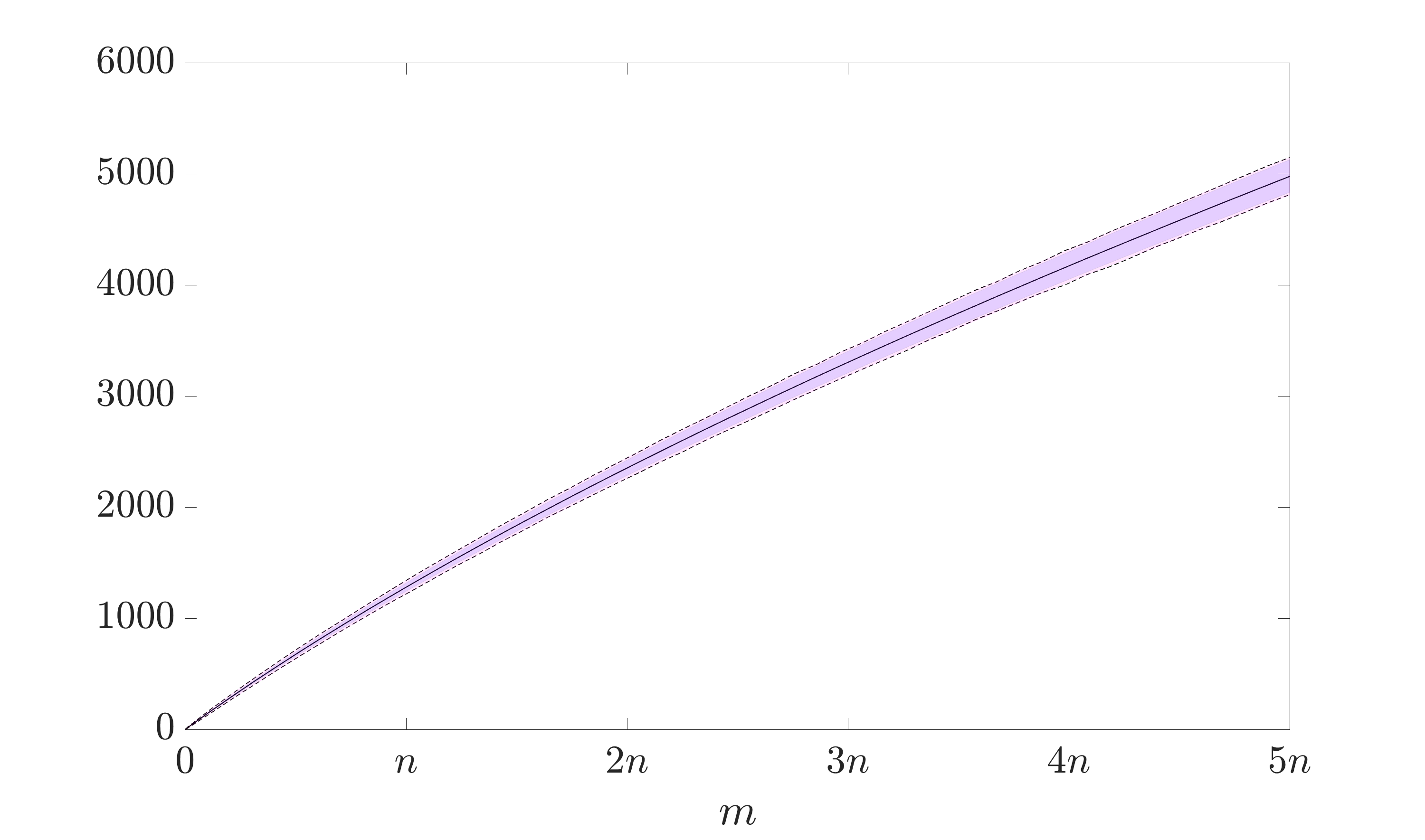}
\end{center}
\end{figure}\label{fig:tf}

 \begin{figure}[h!]
\begin{center}
\textit{Mastigamoeba}
\medskip

\includegraphics[width =0.9\textwidth]{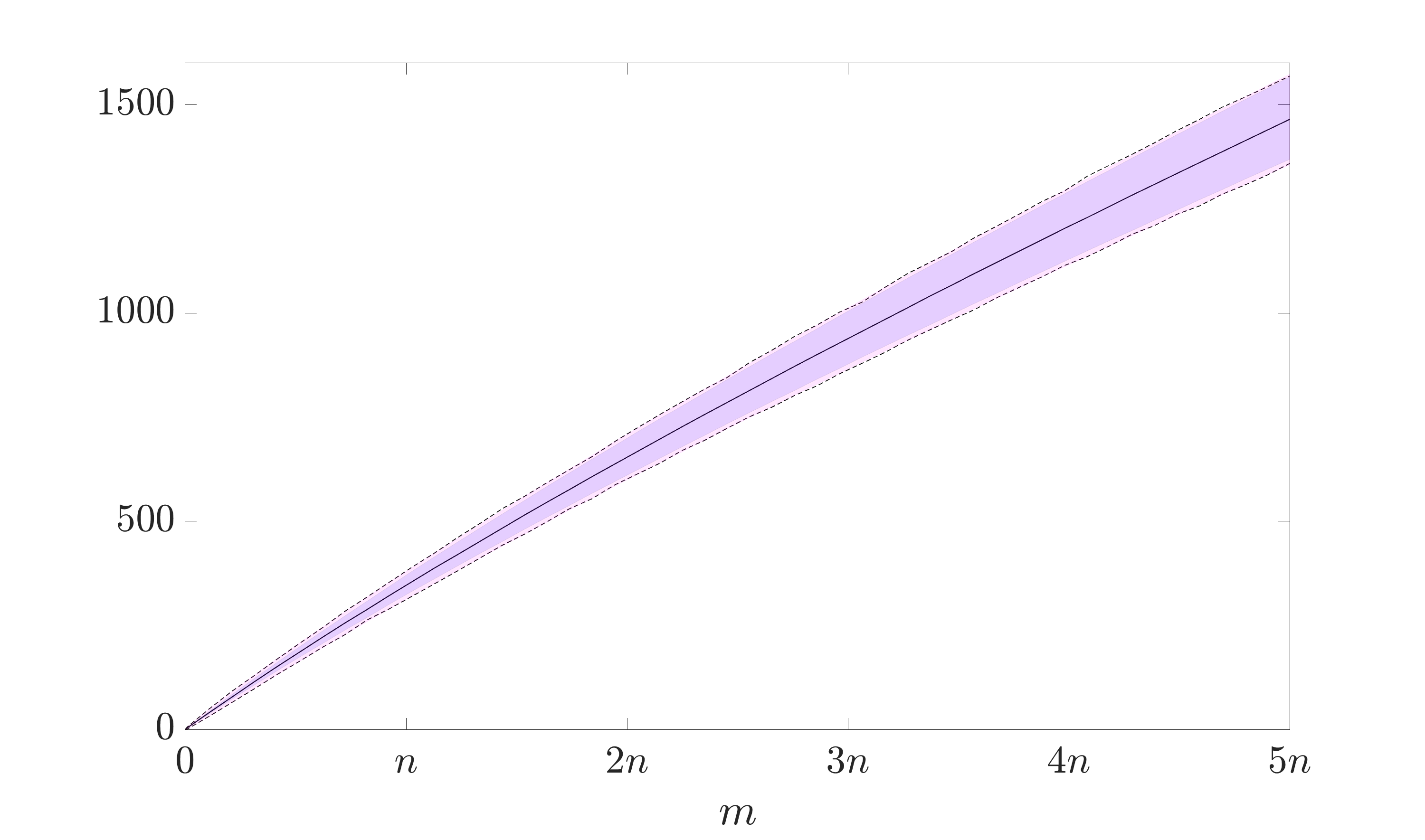}
\end{center}
\end{figure}\label{fig:mast}

 \begin{figure}[h!]
\begin{center}
\textit{Mastigamoeba}, normalized
\medskip

\includegraphics[width = 0.9\textwidth]{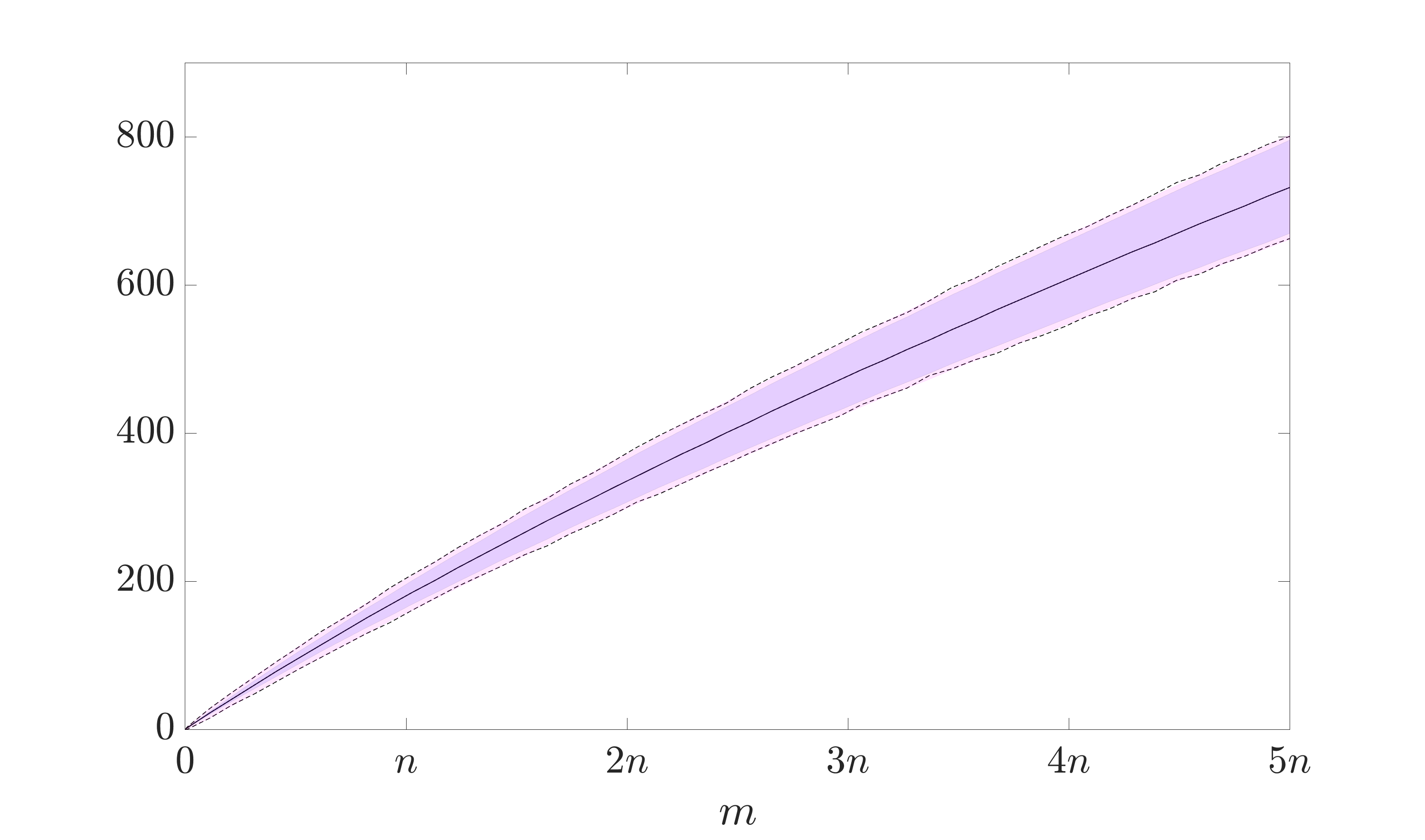}
\end{center}
\end{figure}\label{mastN}

 \begin{figure}[h!]
\begin{center}
\textit{Naegleria} anaerobic
\medskip

\includegraphics[width = 0.9\textwidth]{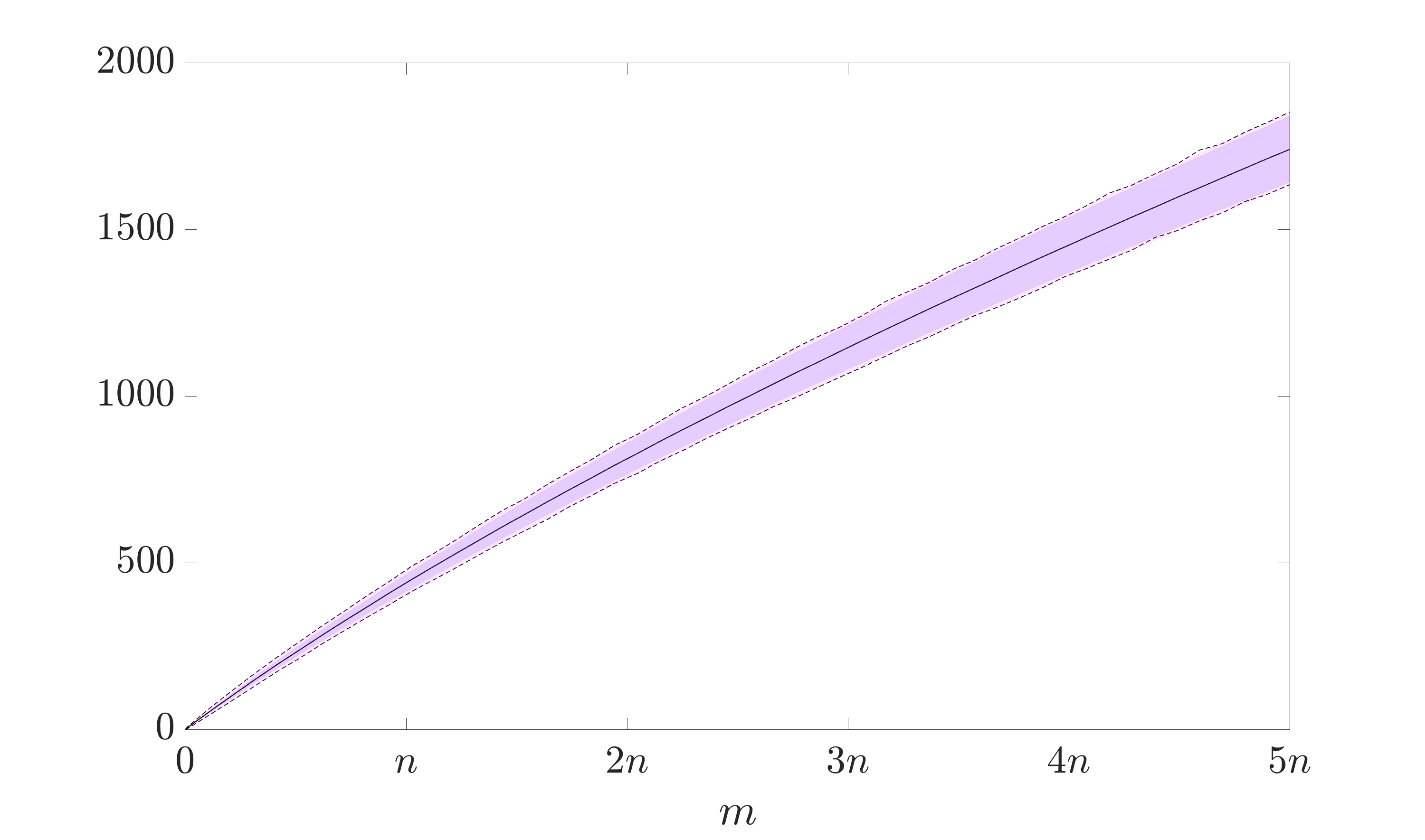}
\end{center}
\caption{\scriptsize{BNP estimates of $K_{n,m}$ (solid line --) with 95\% exact credible intervals (dashed line - -), Mittag-Leffler credible intervals (violet) and Gaussian credible intervals (pink), as a function of $m \in [0, 5n]$.}}
\label{fig:naegAA}
\end{figure}

For the tomato flower, \emph{Mastigamoeba}, \emph{Mastigamoeba} normalized,  \emph{Naegleria} aerobic and \emph{Naegleria} anaerobic EST datasets, Figure \ref{fig:coverage_r} shows the coverage of the Mittag-Leffler credible interval (blue) and of the Gaussian credible interval (red) as a function of $m \in [0, 5n]$. The coverages are evaluated at a uniform mesh of $50$ points over $[0, 5n]$, the same values of $m$  considered for Figure \ref{fig:2} and Figure \ref{fig:naegAA}. 

\begin{figure}[h!]
\begin{minipage}{0.5 \textwidth}
\begin{center}
\medskip
\textit{Mastigamoeba}
\medskip

\includegraphics[width = \textwidth]{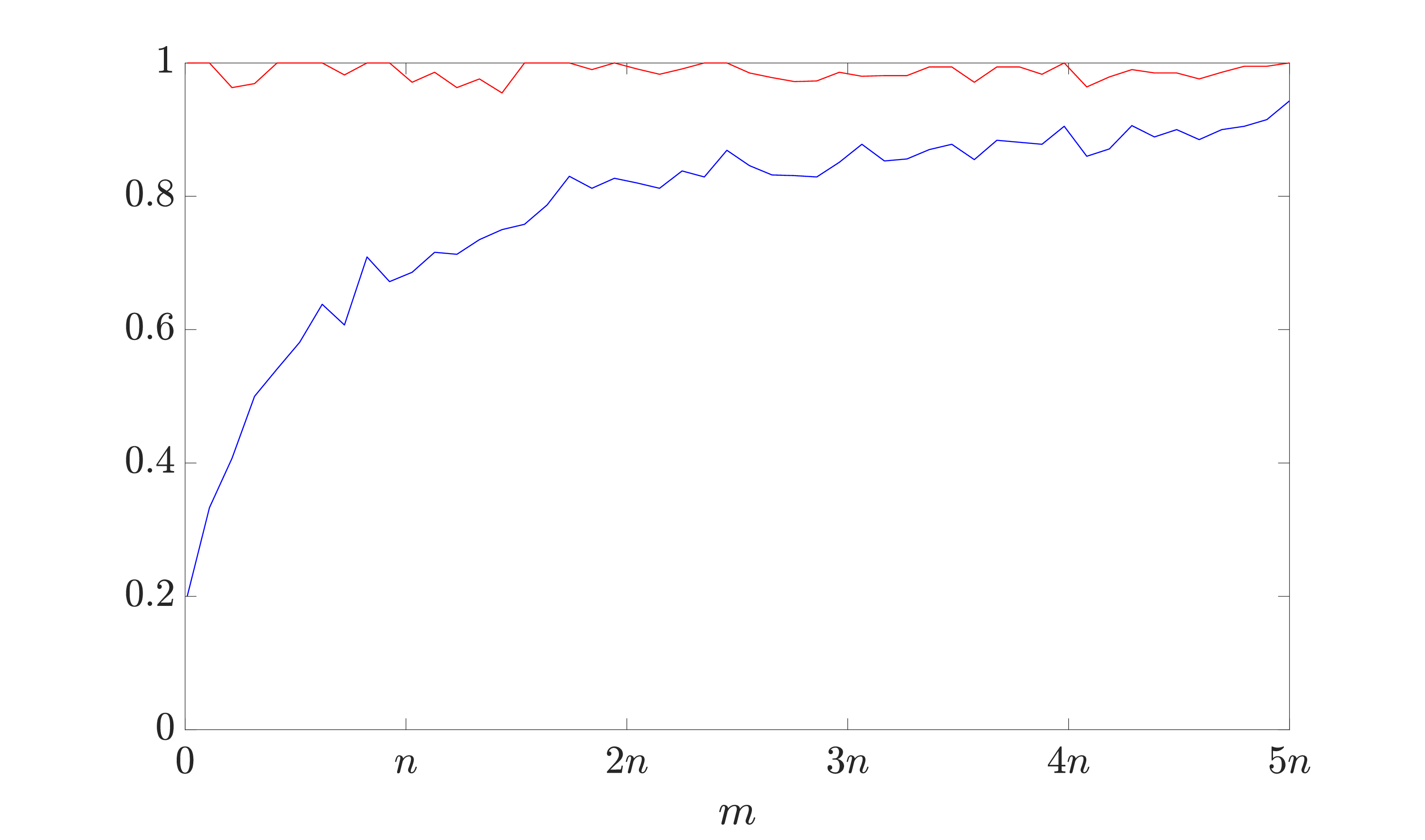}
\end{center}
\end{minipage}
\begin{minipage}{0.5 \textwidth}
\begin{center}
\medskip
\textit{Mastigamoeba}, normalized
\medskip

\includegraphics[width = \textwidth]{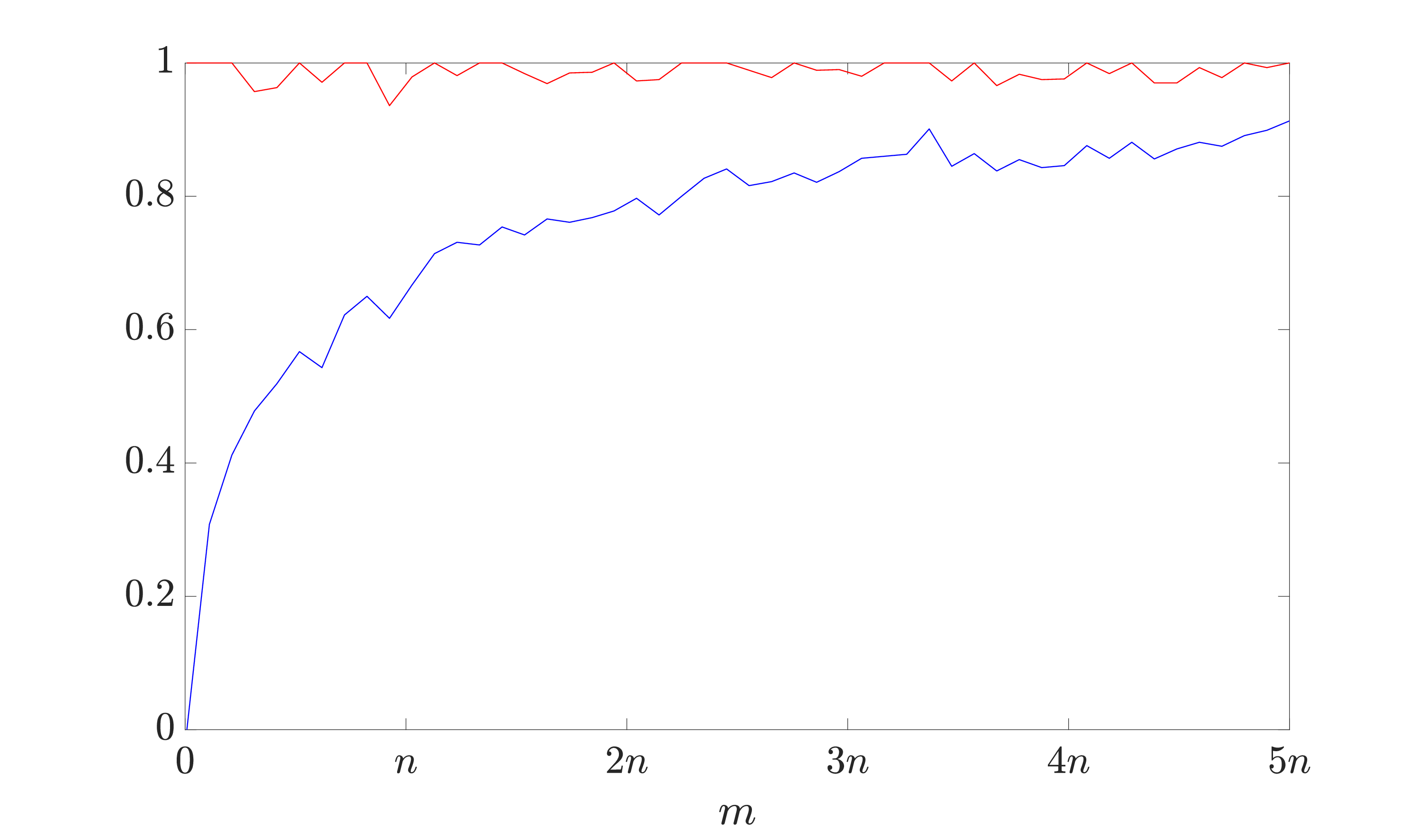}
\end{center}
\end{minipage}

\begin{minipage}{0.5 \textwidth}
\begin{center}
\medskip
\textit{Naegleria} aerobic
\medskip

\includegraphics[width = \textwidth]{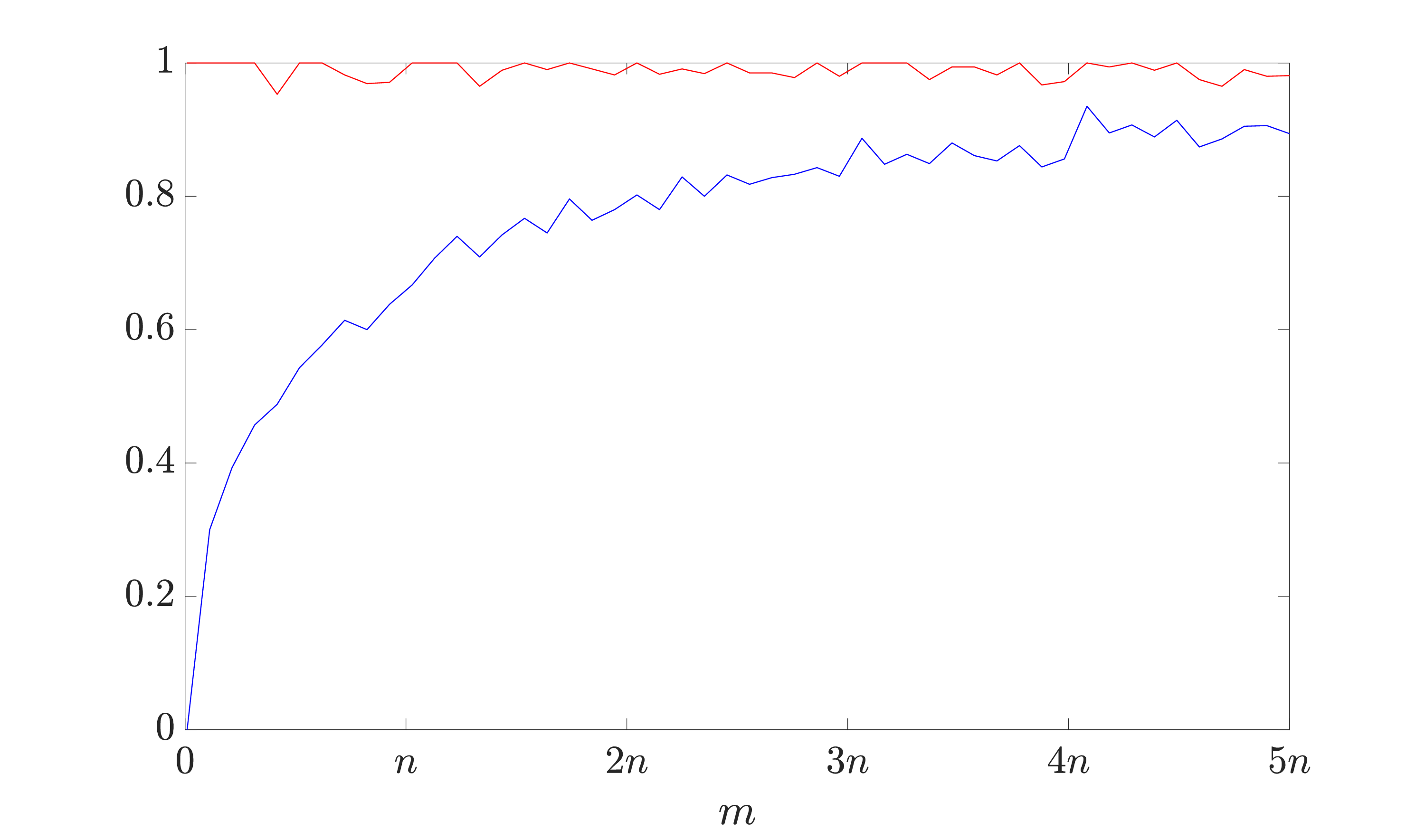}
\end{center}
\end{minipage}
\begin{minipage}{0.5 \textwidth}
\begin{center}
\medskip
\textit{Naegleria} anaerobic
\medskip

\includegraphics[width = \textwidth]{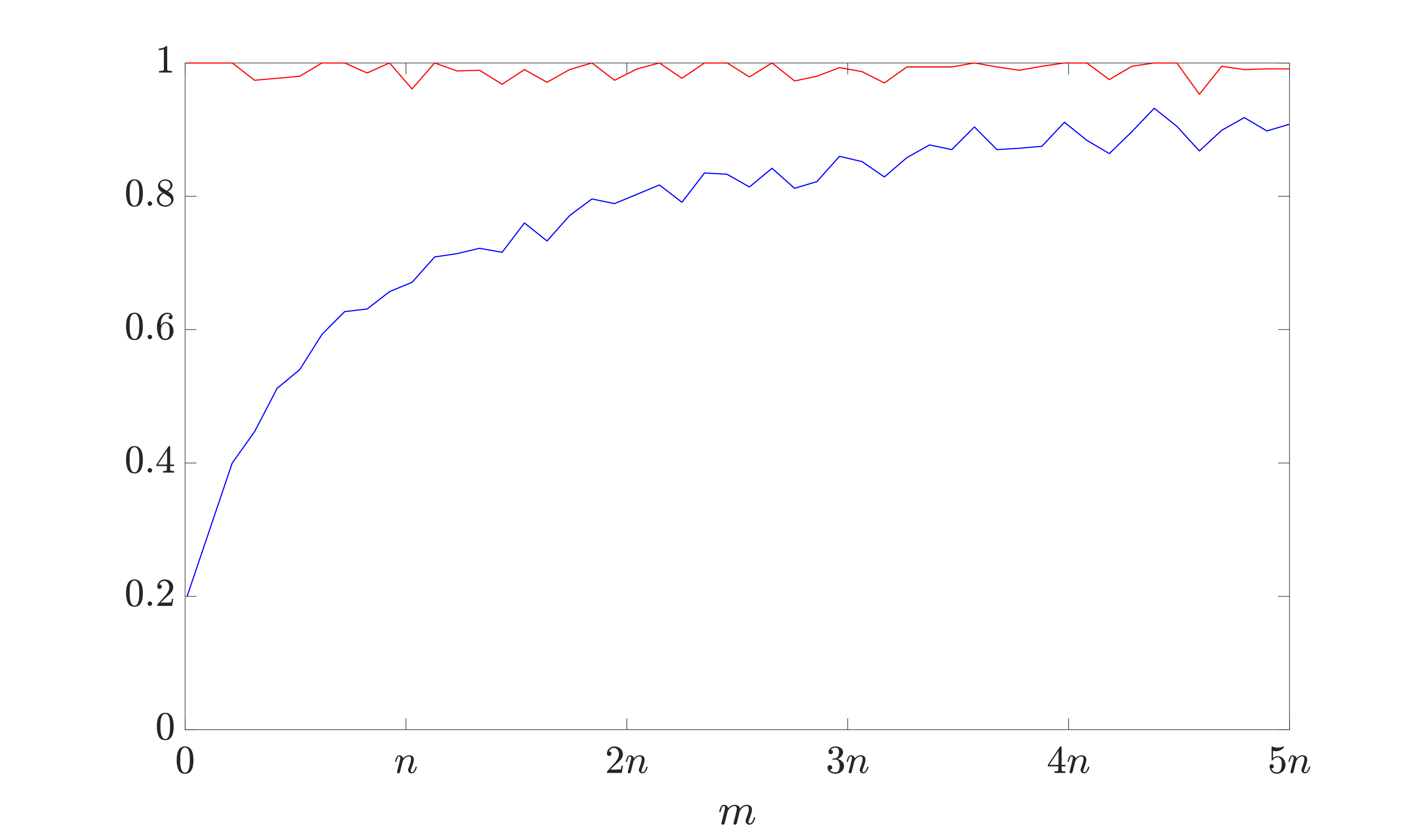}
\end{center}
\end{minipage}
\begin{minipage}{ \textwidth}
\begin{center}
\medskip
Tomato flower
\medskip

\includegraphics[width = 0.5\textwidth]{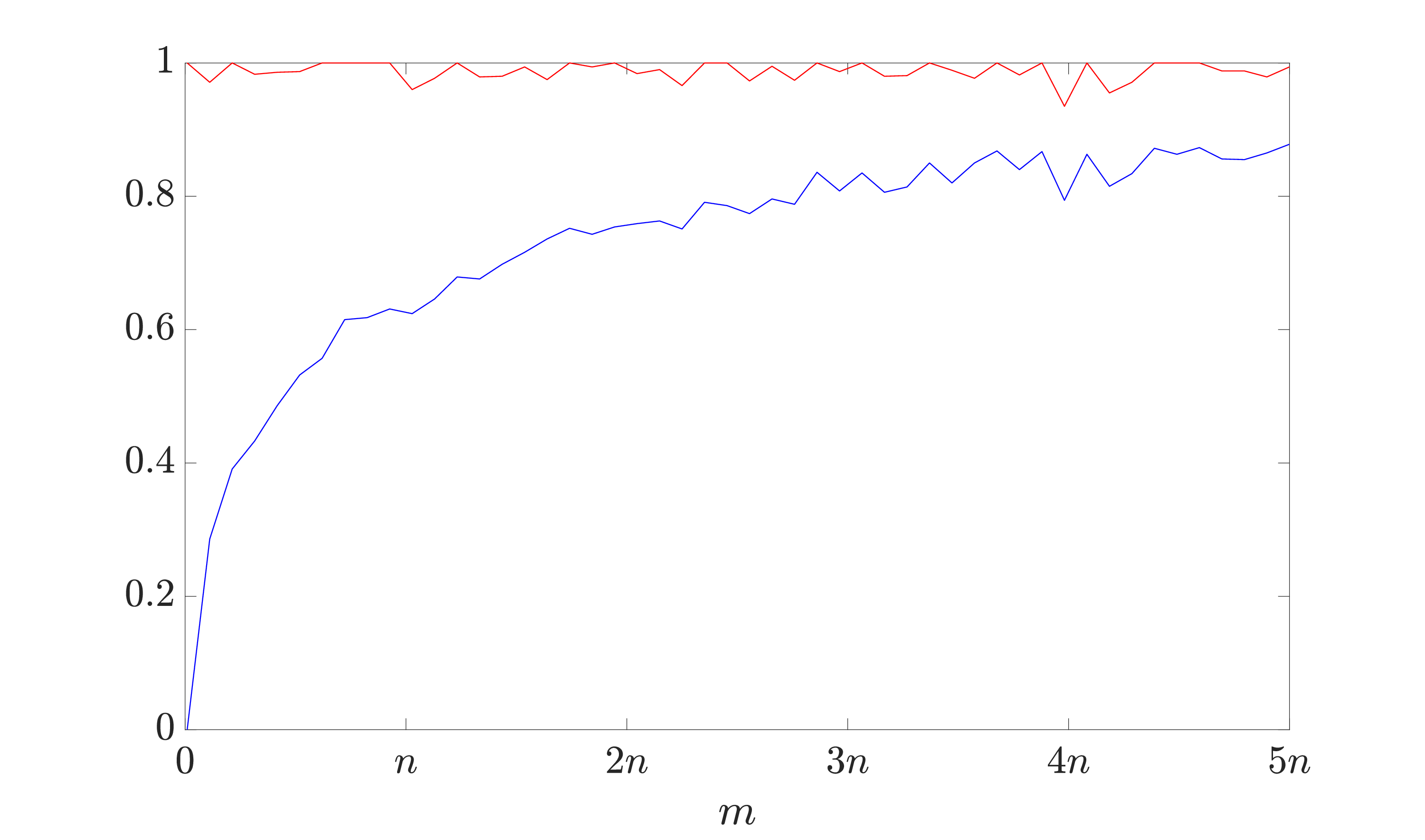}
\end{center}
\end{minipage}
\caption{\scriptsize{Coverage of the Mittag-Leffler credible interval (blue) and of the Gaussian credible interval (red) as a function of $m \in [0, 5n]$. }}
\label{fig:coverage_r}
\end{figure}

Figure \ref{fig:coverage_r} confirms the behavior of the coverage observed in Table \ref{tab:3}:  the coverage of the Gaussian credible intervals is nearly constant in $m$, with values oscillating between 95\% and 100\%. Instead, the coverage of the Mittag-Lefffler credible intervals is, for all values of $m$, lower than that of the corresponding Gaussian credible intervals; such a coverage increases in $m$.


\end{document}